\documentclass[12pt, letterpaper]{article}

\usepackage[T1]{fontenc}
\usepackage[utf8]{inputenc}
\usepackage{textcomp} 
\usepackage{amsmath,amssymb}  
\usepackage{eucal}  
\usepackage{amsthm}
\usepackage{bm}      
\usepackage{dsfont}  

\usepackage[american]{babel}
\usepackage{csquotes}

\theoremstyle{plain}

\newtheorem{propn}{Proposition}
\newtheorem{cor}{Corollary}
\newtheorem*{cor*}{Corollary}
\newtheorem*{construct*}{Construction}
\newtheorem*{hypoth*}{Hypothesis}
\newtheorem*{lem*}{Lemma}
\theoremstyle{definition}
\newtheorem{defn}{Definition}
\newtheorem*{defn*}{Definition}
\theoremstyle{remark}
\newtheorem*{rmk*}{Remark}
\newtheorem*{rmks*}{Remarks}
\newtheorem*{note*}{Note}
\newtheorem*{obs*}{Observation}
\newtheorem*{cav*}{Caveat}

\newcounter{examplectr}
\newenvironment{example}%
  {\refstepcounter{examplectr}\par\bfseries Example \arabic{examplectr}.\mdseries}%
  {\par}

\DeclareMathOperator{\mean}{mean}
\DeclareMathOperator{\diag}{diag}
\DeclareMathOperator{\rank}{rank}

\DeclareMathOperator{\zeros}{zeros}
\DeclareMathOperator{\avg}{avg}
\DeclareMathOperator{\rms}{rms}

\usepackage{verbatim}  

\usepackage{graphicx}
\usepackage[table,rgb]{xcolor}
  \colorlet{approx}{yellow!30}
  \colorlet{productive}{green!30}
  \colorlet{nonproductive}{orange!30}
  \colorlet{octavecomment}{red!50!blue}
  \colorlet{first70}{-blue}
  \colorlet{middle100}{green}
  \colorlet{last30}{-red}
  \colorlet{uniform}{gray!50}
  \colorlet{flat}{blue}
  \colorlet{feasible}{blue!15}
  \colorlet{feasible1}{blue!15}
  \colorlet{feasible2}{blue!5!gray!15}
  \colorlet{obtainable}{blue!15}
  \colorlet{obtainable1}{blue!15}
  \colorlet{obtainable2}{blue!5!gray!15}
  \colorlet{boundary}{blue!90!black}
  \colorlet{boundary1}{blue!50!black!75}
  \colorlet{boundary2}{blue!50!black!60}
  \colorlet{connection1}{blue!40!black!30}
  \colorlet{url}{blue!70!black}
  \colorlet{darkred}{red!50!black}
  \colorlet{xcol}{black!30!}
  \colorlet{ecol}{green!70!black}
\usepackage{hhline}
\doublerulesepcolor{white}
\usepackage{geometry}
\usepackage{calc}
\usepackage{float}
\usepackage[rflt]{floatflt}
\usepackage[justification=centering]{caption}
\usepackage{paralist}    
\usepackage{footmisc}  
\usepackage{dcolumn}  
    \newcolumntype{.}{D{.}{.}{-1}}
\usepackage{tikz}
\usetikzlibrary{arrows}
\usepackage{listings}
\newcounter{aecount}
\newenvironment{aenum}
  {\begin{list}{\arabic{aecount})}
  { \usecounter{aecount}
    \setlength\labelwidth{3.0ex}
    \setlength\leftmargin{\labelwidth+\labelsep+2.0ex}%
    \setlength\rightmargin{3.0ex}%
    \setlength\parskip{-1.5ex}%
    \setlength\itemsep{-1.0ex}%
  }}%
  {\end{list}}
\newenvironment{adescript}
  {\begin{list}{}
  {
    \setlength\labelwidth{50pt}%
    \setlength\leftmargin{\labelwidth+\labelsep+3.0ex}%
    \setlength\rightmargin{3.0ex}%
    \setlength\parskip{-1.5ex}%
    \setlength\itemsep{-1.5ex}%
  }}%
  {\end{list}}

\newcommand{\ticker}[1]{\texttt{#1}}
\newcommand{\BigFig}[1]{\parbox{12pt}{\Huge #1}}
\newcommand{\BigZero}{\BigFig{0}}

\newcounter{saveeqn}

\usepackage{parskip}   
\usepackage[hyperfootnotes=false,%
   colorlinks=true,citecolor=blue,urlcolor=url,runcolor=run,linkcolor=red%
   ]{hyperref}

\begin{document}

\tolerance=550  
\emergencystretch=2ex 

\captionsetup{justification=centering}  

\thispagestyle{empty}
\begin{center}

\vspace*{0.5in}
\parbox{4.5in}{\LARGE
\centering
An algorithm for the orthogonal decomposition of financial
return data
}

\vspace*{0.5in}
Vic Norton\\
Department of Mathematics and Statistics\\
Bowling Green State University\\
\url{mailto:vic@norton.name}\\
\url{http://vic.norton.name}

\vspace*{0.25in}
12-Jun-2012\\
revised 18-May-2013

\vspace*{1.0in}
ABSTRACT\\[2.5ex]
\begin{minipage}{4.1in}
We present an algorithm for the decomposition of periodic financial
return data into orthogonal factors of expected return and ``systemic'',
``productive'', and ``nonproductive'' risk. Generally, when the number
of funds does not exceed the number of periods, the expected return of a
portfolio is an affine function of its productive risk.

\vspace*{1.0in}
\end{minipage}

\end{center}

\clearpage


\section*{Preface}
\thispagestyle{empty}

This is a paper about our \textbf{rtndecomp} algorithm, an algorithm for
decomposing financial return data into expected returns and principal
components of risk. A complete listing of the algorithm appears in
Appendix \ref{octave-listing}. Section \ref{rtndecomp_function}
describes exactly what the algorithm does. The rest of the paper is
background---more or less.

The paper is accompanied by three ancillary text files:
\begin{itemize}
  \item rtndecomp.m -- The GNU Octave function.
  \item GPLv3.txt -- The GNU General Public License governing the
  use of the \mbox{rtndecomp.m} code.
  \item AdjustedClosingPrices\_2010-2011.csv --
  The adjusted closing prices, in tab-separated-value (spreadsheet)
  format, of 22 iShares exchange traded funds on the 505 market days
  from 2009-12-31 to 2011-12-30 inclusive. These prices are normalized
  at 100.000 on 2010-12-31. This means that the proportions in a
  notional portfolio $\mathbf{p}=(p_1,\ldots,p_{22})$ represent the
  proportions of the 22 securities in an actual investment portfolio at
  the close of 2010-12-31. The security proportions in the same
  investment portfolio are typically different at the close of any one
  of the other 504 market days under consideration.
\end{itemize}

Section \ref{output-examples}, \textbf{Examples of output}, illustrates
the application of the algorithm to real world data. All computations in
this section are based on the adjusted closing prices in
``AdjustedClosingPrices\_2010-2011.csv.''


\newpage
\thispagestyle{empty}
\tableofcontents


\newpage
\setcounter{page}{1}
\section{The standard model}
\label{stdmodel}

We start with a synopsis of the ``standard mean-variance
portfolio selection model'' (\cite[pp. 3--5]{Markowitz:1987wd}) for
ex post return data.

Given an $M \times n$ matrix $R = [\mathbf{r}_1,\ldots,\mathbf{r}_n]$ of
successive periodic returns ($M$ returns for each of $n$ securities), an
investor is to choose the proportions $\mathbf{p} = [p_1,\ldots,p_n]^T$
invested in each security, the proportions being subject to the
constraints $p_j\ge0~ (j=1,\ldots,n),~ \sum_{j=1}^n p_j = 1.$ We assume
that the periodic returns, $\mathbf{r}_\textbf{p}\in\mathds{R}^M$, of the
corresponding \emph{investment portfolio} satisfy the \emph{linear
hypothesis}
\begin{equation}\label{linear_hypothesis}
  \mathbf{r}_\textbf{p} = \sum_{j=1}^n \mathbf{r}_j p_j = R \mathbf{p}.
\end{equation}
We also assume that \emph{expected} periodic return is a linear function
of periodic return or, in other words, the expected periodic return of
security $j$ is given by $e_j=\bm{\omega}^T\bm{r}_j$ for $j=1,\dots,n$.
Here the weight vector $\bm{\omega}\in\mathds{R}^M$ should satisfy
$\omega_i > 0~ (i=1,\dots,M)$ and $\sum_{i=1}^M\omega_i = 1.$

Under these assumptions, the expected periodic return of
the investment portfolio corresponding to $\mathbf{p}$ is
\begin{equation}\label{expected_return}
  e_\textbf{p} = \sum_{j=1}^n e_j p_j = E \mathbf{p},
\end{equation}
with~ $E = [e_1,\ldots,e_n] = \bm{\omega}^T R$, and the \emph{variance}
of portfolio return is
\begin{equation}\label{var1}
  v_\textbf{p} =
  \sum_{j=1}^n\sum_{k=1}^n v_{jk} p_j p_k = \mathbf{p}^T V \mathbf{p},
\end{equation}
where the $n\times n$ \emph{covariance} matrix ~$V = [v_{jk}]$~ is given
by
\begin{equation}\label{cov1}
  v_{jk}= \sum_{i=1}^M \omega_i z_{ij}z_{ik}
  \quad (j, k = 1,\ldots n),
\end{equation}
the deviation or ``risk'' vectors~
$\mathbf{z}_j\in\mathds{R}^M~ (j=1,\ldots,n)$~
being defined by
\begin{equation}\label{risk-vectors}
  \mathbf{z}_j = \mathbf{r}_j - \mathbf{1}_M e_j,
\end{equation}
with $\mathbf{1}_M\in\mathds{R}^M$ representing the constant return
vector of all 1's.\\

\vspace*{1.0ex}
\begin{cav*}
  If the periodic returns in $R$ are \emph{normalized linear returns},
  then the normalized linear returns of each investment portfolio in the
  $n$ securities satisfy the linear hypothesis \eqref{linear_hypothesis}
  with respect to some $\mathbf{p}=[p_1,\ldots,p_n]^T$, and all of the
  above statements follow (\cite{Norton:2011fk}). More typically, when
  compound periodic returns are used, the linear hypothesis cannot be
  satisfied by any $\mathbf{p}$, and the arguments of this paper do not
  apply.
\end{cav*}


\section{Geometry}
\label{geometry}

We will consider the ex post standard model from a geometric
standpoint. The a priori weights, $\bm{\omega}$, of section
\ref{stdmodel} induce a Euclidean metric on the space of consecutive
periodic returns, $\mathds{R}^M$:
\begin{equation}\label{inner-product}
  \langle\mathbf{x},\mathbf{y}\rangle_\omega
    = \sum_{i=1}^M\omega_i x_i y_i, \quad
  \|\mathbf{x}\|_\omega
    = \sqrt{\langle\mathbf{x},\mathbf{x}\rangle_\omega}\,,
    \quad\text{for}\quad
  \mathbf{x},\mathbf{y}\in\mathds{R}^M.
\end{equation}
Two return vectors $\mathbf{x}$ and $\mathbf{y}$ are \emph{orthogonal}
(perpendicular to each other) if
$\langle\mathbf{x},\mathbf{y}\rangle_\omega=0$.

The vector of all ones, $\mathbf{1}_M$, is a unit vector in this
Euclidean space since $\sum_{i=1}^M\omega_i=1$. The expected
return axis, the $E$-axis, points in the $\mathbf{1}_M$-direction. The
$E$-coordinate of any periodic return vector $\mathbf{r}\in\mathds{R}^M$,
\begin{equation}\label{e-coordinate}
  e=\langle\mathbf{1}_M,\mathbf{r}\rangle_\omega=\sum_{i=1}^M\omega_i r_i,
\end{equation}
is its expected return.

\begin{minipage}{2.0in}
  \captionof{figure}{$\mathbf{r}=\mathbf{z}+\mathbf{1}_Me$~~~}%
  \label{rze}
\begin{tikzpicture}[scale=1, >={angle 60}]
  \draw[->] (0,-0.5) -- (0,2.8) node[above]{$E$};
  \draw (-0.8,0) -- (3.5,0) node[right]{$\mathcal{Z}$};
  \draw[->, ultra thick] (0,0) -- (0,1.3);
  \draw (-0.2,1.3) node[left]{1} -- (0,1.3);
  \draw (0,0.5) node[left]{$\mathbf{1}_M$};
  \draw (0,0) rectangle (0.3,0.3);
  \draw[->, ultra thick] (0,0) -- (3,0);
  \draw (1.5,0) node[above]{$\mathbf{z}$};
  \draw (-0.2,2.2) node[left]{$e$} -- (0,2.2);
  \draw[dashed] (0,2.2) -- (3.0,2.2) -- (3.0,0);
  \fill[fill=black] (3.0,2.2) circle[radius=0.10] node[above=2.5] {$\mathbf{r}$};
  \draw (1.5,-0.02) node[below]{$\sigma(\mathbf{r})=\|\mathbf{z}\|_\omega$};
  \draw[->](0.3,-0.4) -- (0,-0.4);
  \draw[->](2.7,-0.4) -- (3.0,-0.4);
  \draw (3.0,0) -- (3.0,-0.5);
\end{tikzpicture}
\end{minipage}
\hfill\begin{minipage}{3.8in}
Each periodic return vector, $\mathbf{r}$, has an orthogonal
decomposition into its (scalar) expected-return component, $e$, and its
(vector) risk component,
\begin{equation}\label{zre}
  \mathbf{z} = \mathbf{r}-\mathbf{1}_M e,
\end{equation}
with expected return zero.
The standard deviation of periodic return
is simply the length or norm of the risk component,
\begin{equation}\label{stdv(r)}
  \sigma(\mathbf{r})=\|\mathbf{z}\|_\omega,
\end{equation}
and the variance of periodic return is its square norm,
\begin{equation}\label{var2}
  v(\mathbf{r})=\|\mathbf{z}\|_\omega^2.
\end{equation}
\end{minipage}

The covariance matrix $V=[v_{jk}]$ of \eqref{cov1} is the Gram matrix of
inner products of the security risk vectors
$Z=[\mathbf{z}_1,\ldots,\mathbf{z}_n]$ of \eqref{risk-vectors}:
\begin{equation}\label{cov2}
  v_{jk}=\langle\mathbf{z}_j,\mathbf{z}_j\rangle_\omega\quad
  (j,k=1,\ldots,n).
\end{equation}


\section{Linear subspaces and flats}\label{spaces-flats}

We are concerned with notional portfolios in $n$ specific securities.
The return vectors of these portfolios lie in the the linear subspace
$\mathcal{L}(R)$ of $\mathds{R}^M$ spanned by the return vectors,
$R=[\mathbf{r}_1,\ldots,\mathbf{r}_n]$, of the individual securities:
\begin{equation*}
  \mathcal{L}(R)=\{~\sum_{j=1}^n\mathbf{r}_jt_j: t_j\in\mathds{R}~\}.
\end{equation*}
The risk components of portfolio return vectors lie in the linear
subspace $\mathcal{L}(Z)$ of $\mathds{R}^M$ spanned by the risk
components, $Z=[\mathbf{z}_1,\ldots,\mathbf{z}_n]$, of the
$\mathbf{r}_j$:
\begin{equation*}
  \mathcal{L}(Z)=\{~\sum_{j=1}^n\mathbf{z}_jt_j: t_j\in\mathds{R}~\}.
\end{equation*}

Since the proportions of the securities in a notional portfolio must sum
to 1, portfolio return vectors and their risk components are contained
in the flats (affine subspaces of $\mathds{R}^M$) defined by
\begin{align*}
  \mathcal{F}(R)&=\{~\sum_{j=1}^n\mathbf{r}_jt_j: \sum_{j=1}^nt_j=1~\}\\
  \intertext{and}
  \mathcal{F}(Z)&=\{~\sum_{j=1}^n\mathbf{z}_jt_j: \sum_{j=1}^nt_j=1~\},
\end{align*}
respectively. We will refer to these as the $R$- and $Z$-flats.

Finally, we will be concerned with differences in periodic return
vectors, and the corresponding differences in their risk components,
from one notional portfolio to another. Such difference vectors reside
in the tangent spaces
\begin{align*}
  \mathcal{T}(R)&=\{~\sum_{j=1}^n\mathbf{r}_jt_j: \sum_{j=1}^nt_j=0~\}\\
  \intertext{and}
  \mathcal{T}(Z)&=\{~\sum_{j=1}^n\mathbf{z}_jt_j: \sum_{j=1}^nt_j=0~\}
\end{align*}
of the $R$- and $Z$-flats.\\

\begin{propn}\label{isomorph}
   If ~$\mathbf{1}_M\notin\mathcal{T}(R)$, then the risk component
   mapping $\mathbf{r}\mapsto\mathbf{z}$ defined by \eqref{e-coordinate}
   and \eqref{zre} restricts to a linear isomorphism of $\mathcal{T}(R)$
   onto $\mathcal{T}(Z)$.
\end{propn}

\begin{proof}
The mapping from $\mathcal{T}(R)$ onto $\mathcal{T}(Z)$ can be expresses
as $\bm{\Delta}\mathbf{r}\mapsto\bm{\Delta}\mathbf{z}$ with
\begin{equation*}
  \bm{\Delta}\mathbf{r}=\bm{\Delta}\mathbf{z}
  +\mathbf{1}_M\sum_{j=1}^ne_jt_j,\quad
  \bm{\Delta}\mathbf{z}=\sum_{j=1}^n\mathbf{z}_jt_j,
  \quad\text{and}\quad
  \sum_{j=1}^nt_j=0.
\end{equation*}
To show that this mapping is a linear isomorphism, we need to show that
$\bm{\Delta}\mathbf{r}=\mathbf{0}_M$ whenever
$\bm{\Delta}\mathbf{z}=\mathbf{0}_M$. But, if
$\bm{\Delta}\mathbf{z}=\mathbf{0}_M$, then
$\bm{\Delta}\mathbf{r}=\mathbf{1}_M\sum_{j=1}^ne_jt_j$. And then,
since $\mathbf{1}_M\notin\mathcal{T}(R)$, 
$\sum_{j=1}^ne_jt_j=0$ and $\bm{\Delta}\mathbf{r}=\mathbf{0}_M$.
\end{proof}

\vspace*{2.0ex}
\begin{cor*}
  If ~$\mathbf{1}_M\notin\mathcal{T}(R)$, then
  \begin{equation}\label{z_mapsto_r}
    \sum_{j=1}^n\mathbf{z}_jt_j\mapsto
    \sum_{j=1}^n\mathbf{r}_jt_j\quad\text{for}\quad
    \sum_{j=1}^nt_j=1
  \end{equation}
  is a well-defined mapping of $\mathcal{F}(Z)$ onto $\mathcal{F}(R)$.
  It is the inverse of the risk component mapping from
  $\mathcal{F}(R)$ onto $\mathcal{F}(Z)$.
\end{cor*}


\section{Components of portfolio risk}
\label{risk-components}

The \emph{total variance} of return of the periodic returns in
$R=[\mathbf{r}_1,\ldots,\mathbf{r}_n]$ is the sum of the variances of
return of the individual securities:
\begin{equation}\label{totalvar1}
  v_\text{T} = \sum_{j=1}^nv_{jj}
    = \sum_{j=1}^n\|\mathbf{z}_j\|_\omega^2.
\end{equation}
This is a measure of the volatility of the return data as a whole, of
the spread of the periodic returns in
$R=[\mathbf{r}_1,\ldots,\mathbf{r}_n]$ away from their expected values
$E=[e_1,\ldots,e_n]$.

Given a unit risk vector $\mathbf{u}\in\mathcal{L}(Z)$, the variance of
return of security $j$ \emph{in the $\mathbf{u}$-direction} is the
square of the $\mathbf{u}$-coordinate of its risk vector,
$\langle\mathbf{u},\mathbf{z}_j\rangle_\omega^2$. The total
variance of return in the $\mathbf{u}$-direction is the sum of the
$\mathbf{u}$-directional variances:
\begin{equation}\label{totalvaru}
  v_\textbf{u} = \sum_{j=1}^n\langle\mathbf{u},\mathbf{z}_j\rangle_\omega^2.
\end{equation}

If~ $\mathcal{U}\subset\mathcal{L}(Z)$ is an orthonormal basis for
$\mathcal{L}(Z)$ (a pairwise-orthogonal set of unit vectors that span
$\mathcal{L}(Z)$), then
\begin{equation}\label{var3}
  \|\mathbf{z}_j\|_\omega^2 =
    \sum_{\mathbf{u}\in\mathcal{U}}
    \langle\mathbf{u},\mathbf{z}_j\rangle_\omega^2 \quad
    (j=1,\ldots,n).
\end{equation}
Consequently
\begin{equation}\label{totalvar2}
  v_\text{T} = \sum_{j=1}^n\|\mathbf{z}_j\|_\omega^2 =
    \sum_{j=1}^n\sum_{\mathbf{u}\in\mathcal{U}}
    \langle\mathbf{u},\mathbf{z}_j\rangle_\omega^2 =
    \sum_{\mathbf{u}\in\mathcal{U}}\sum_{j=1}^n
    \langle\mathbf{u},\mathbf{z}_j\rangle_\omega^2 =
    \sum_{\mathbf{u}\in\mathcal{U}}v_\textbf{u}.
\end{equation}

For \emph{principal component analysis} (\cite{Wikipedia:2011rr}) one
attempts to choose the orthogonal basis $\mathcal{U}$ (orthogonal
coordinate system if you will) so that the sum of
$\mathbf{u}$-directional-total-variances on the right side of
\eqref{totalvar2} decomposes or ``explains'' the total variance, $v_\text{T}$,
in a particularly meaningful way. We are aiming for such a decomposition
of the total variance of return in this paper. Our idea of a
``particularly meaningful way'' will be defined in this section.


\subsection{Systemic risk}\label{systemic-risk}

\begin{lem*}
  Let $\mathbf{z}_0$ denote the point in the $Z$-flat
  that is closest to the origin:
  \begin{equation*}
    \|\mathbf{z}_0\|_\omega = \min\{~\|\mathbf{z}\|_\omega
    : \mathbf{z}\in\mathcal{F}(Z)~ \}.
  \end{equation*}
  Then
  \begin{equation}\label{z0perpZ-flat}
    \langle\mathbf{z}_0,\mathbf{z}-\mathbf{z}_0\rangle_\omega=0
    \quad\text{for all}\quad\mathbf{z}\in\mathcal{F}(Z)
  \end{equation}
\end{lem*}

\begin{proof} Given $\mathbf{z}\in\mathcal{F}(Z)$,
  $\mathbf{z}(t)=\mathbf{z}_0+(\mathbf{z}-\mathbf{z}_0)t$
  is in $\mathcal{F}(Z)$ for all $t\in\mathds{R}$. By definition
  $\|\mathbf{z}(t)\|_\omega^2$ achieves its minimum value
  of $\|\mathbf{z}_0\|_\omega^2$ at $t=0$. Consequently
  \begin{equation*} \frac{1}{2}
    \left.\frac{d}{dt}\right|_{t=0}\|\mathbf{z}(t)\|_\omega^2
    = \langle\mathbf{z}_0,\mathbf{z}-\mathbf{z}_0\rangle_\omega=0.
  \end{equation*}
\end{proof}

\begin{propn}\label{vjkhat}
  Let $f_0=\|\mathbf{z}_0\|_\omega$. Then
\begin{equation}\label{cov3}
  v_{jk}=f_0^2+\hat{v}_{jk}\quad\text{with}\quad \hat{v}_{jk}=
    \langle\mathbf{z}_j-\mathbf{z}_0,\mathbf{z}_k-\mathbf{z}_0
    \rangle_\omega\quad (j,k=1,\ldots,n).
\end{equation}
\end{propn}

\begin{proof}
\begin{alignat*}{3}
  v_{jk} &= \langle\mathbf{z}_j,\mathbf{z}_k\rangle_\omega
    &\quad& \eqref{cov2} \\
    &=\langle(\mathbf{z}_j-\mathbf{z}_0)+\mathbf{z}_0,
      (\mathbf{z}_k-\mathbf{z}_0)+\mathbf{z}_0\rangle_\omega\\
    &=\langle\mathbf{z}_j-\mathbf{z}_0,\mathbf{z}_k-\mathbf{z}_0\rangle_\omega
      + \langle\mathbf{z}_0,\mathbf{z}_0\rangle_\omega \\
    &\quad +\langle\mathbf{z}_j-\mathbf{z}_0,\mathbf{z}_0\rangle_\omega
      +\langle\mathbf{z}_0,\mathbf{z}_k-\mathbf{z}_0\rangle_\omega
    &\quad& \text{(bilinear expansion)}\\
    &=\langle\mathbf{z}_j-\mathbf{z}_0,\mathbf{z}_k-\mathbf{z}_0\rangle_\omega
      + \langle\mathbf{z}_0,\mathbf{z}_0\rangle_\omega + 0 + 0
    &\quad& \text{(by the lemma)}\\
    &=\hat{v}_{jk}+f_0^2.
\end{alignat*}
\end{proof}

\begin{cor}
  \begin{equation}\label{totalvar3}
  v_\textnormal{T}=n f_0^2+\hat{v}_\text{T}\quad\text{where}\quad
  \hat{v}_\textnormal{T} = \sum_{j=1}^n\hat{v}_{jj}
   = \sum_{j=1}^n\|\mathbf{z}_j-\mathbf{z}_0\|_\omega^2.
  \end{equation}
\end{cor}

This follows from \eqref{totalvar1} and \eqref{cov3}.
Then, by rewriting \eqref{cov3} in matrix form,
we see that\\

\begin{cor}
\begin{equation}\label{VeqV0+f0}
  V=f_0^2+\widehat{V} 
  \quad\text{where}\quad \widehat{V}=[\hat{v}_{jk}]
  ~(j,k=1,\ldots,n).
\end{equation}
\end{cor}
\vspace*{-2.0ex}
Here we adopt the convention that the sum of a scalar and a matrix is
the original matrix with the scalar
added to its every coefficient.

\setcounter{cor}{0}

\vspace*{1.0ex}
We refer to $f_0$ as the \emph{systemic portfolio risk}, $f_0^2$ is the
\emph{systemic portfolio variance}, and $n f_0^2$ is the \emph{total
systemic variance} of the system. The variance of return of any notional
portfolio $\mathbf{p}$ decomposes into its systemic and nonsystemic
parts:
\begin{equation}\label{var4}
  v_\textbf{p} = \mathbf{p}^T V \mathbf{p}
      = f_0^2 + \mathbf{p}^T \widehat{V} \mathbf{p}.
\end{equation}
The second equation follows from \eqref{VeqV0+f0} and
$\sum_{j=1}^np_j$ = 1.

 A \emph{minimum-variance portfolio} is a
notional portfolio $\mathbf{p}$ whose variance is less than or equal to
the variance of any other notional portfolio $\mathbf{q}$ with the same
expected return. Minimum-variance portfolios play a crutial role in
Markowitz's mean-variance analysis (\cite{Markowitz:1987wd}). By
\eqref{var4}
\begin{equation*}
  \mathbf{p}^TV\mathbf{p}\le\mathbf{q}^TV\mathbf{q}
  \quad\text{if and only if}\quad
  \mathbf{p}^T\widehat{V}\mathbf{p}\le\mathbf{q}^T\widehat{V}\mathbf{q}
\end{equation*}
for notional portfolios $\mathbf{p}$ and $\mathbf{q}$. Consequently, the
collection of all minimum-variance portfolios is completely determined
by the singular, \emph{nonsystemic covariance matrix} $\widehat{V}$ and
the expected return matrix $E$.

If $f_0=\|\mathbf{z}_0\|_\omega\neq0$, we will take
$\mathbf{u}_0=\mathbf{z}_0/f_0$ to be the first vector in our
orthonormal basis $\mathcal{U}$ for $\mathcal{L}(Z)$. This is the
\emph{direction of systemic risk}. Every vector
$\mathbf{z}\in\mathcal{F}(Z)$ has the same $\mathbf{u}_0$-coordinate,
$\langle\mathbf{u}_0,\mathbf{z}\rangle_\omega=f_0$, as can be
seen from the expansion
\begin{alignat*}{3}
  \langle\mathbf{u}_0,\mathbf{z}\rangle_\omega
  &=\langle\mathbf{u}_0,\mathbf{z}_0
    +(\mathbf{z}-\mathbf{z}_0)\rangle_\omega \\
  &=\langle\mathbf{u}_0,\mathbf{z}_0\rangle_\omega
    +\langle\mathbf{u}_0,\mathbf{z}-\mathbf{z}_0\rangle_\omega
    &\quad&\text{(linear expansion)} \\
  &=\langle\mathbf{u}_0,\mathbf{u}_0 f_0\rangle_\omega + 0
    &\quad&\text{(definition of $\mathbf{u}_0$
    and \eqref{z0perpZ-flat})} \\
  &=f_0.
\end{alignat*}

We will assume until further notice that
$\mathbf{1}_M\notin\mathcal{T}(R)$. Then
the mapping
$\mathbf{z}\mapsto\mathbf{r}$ from the $Z$-flat onto the $R$-flat is
well-defined by the corollary to Proposition \ref{isomorph}, and
$\mathbf{z}_0\in\mathcal{F}(Z)$ is the risk component of a unique
$\mathbf{r}_0\in\mathcal{F}(R)$. We will refer to
\begin{equation}\label{systemic-return}
  e_0=\langle\mathbf{1}_M,\mathbf{r}_0\rangle_\omega\
\end{equation}
as the
\emph{systemic portfolio return} of our system. Note that $e_0$ may not
be the expected return of any notional portfolio $\mathbf{p}$, all of
whose coefficients must be nonnegative.


\subsection{Productive risk}\label{productive-risk}

Equation \eqref{z0perpZ-flat} shows that the tangent space
$\mathcal{T}(Z)$ is the orthogonal complement of
$\mathbf{z}_0$ in $\mathcal{L}(Z)$.
Indeed $\mathcal{T}(Z)$ is spanned by the difference vectors
~$\mathbf{z}_j-\mathbf{z}_0~ (j=1,\ldots,n)$, and the nonsystemic
covariance matrix, $\widehat{V}$, is the Gram matrix of these
difference vectors. We will select the remaining orthonormal basis vectors
$\mathbf{u}_i~ (i=1,\ldots,m;~ m<n)$ from $\mathcal{T}(Z)$.
Then the \emph{total nonsystemic variance} $\hat{v}_\text{T}$ of
\eqref{totalvar3}) will decompose as the sum of the squares of the
$\mathbf{u}_i$-coordinates of the $\mathbf{z}_j-\mathbf{z}_0$,
\begin{equation}\label{total-nonsys-var}
  \hat{v}_\text{T}=\sum_{j=1}^n\|\mathbf{z}_j-\mathbf{z}_0\|_\omega^2
    =\sum_{i=1}^m\sum_{j=1}^n\langle\mathbf{u}_i,
      \mathbf{z}_j-\mathbf{z}_0\rangle_\omega^2
    =\sum_{i=1}^m\hat{v}_{\textbf{u}_i}\,,
\end{equation}
and $\widehat{V}$ will factor as $\widehat{V} = F^TF$, with the
coefficients of the $m\times n$ factor matrix $F$ given by
\begin{equation}\label{factor-matrix}
  f_{ij}=\langle\mathbf{u}_i,\mathbf{z}_j-\mathbf{z}_0\rangle_\omega
  \quad(i=1,\ldots,m;j=1,\ldots,n).
\end{equation}

\vspace*{1.0ex}
We will continue to assume that $\mathbf{1}_M\notin\mathcal{T}(R)$, so
that $\mathbf{r}\mapsto\mathbf{z}$ is a bijection of $\mathcal{F}(R)$
onto $\mathcal{F}(Z)$, and further suppose that the $n$ securities do
not all have the same expected return. Under these assumptions the
orthogonal projection of $\mathbf{1}_M$ onto $\mathcal{T}(R)$ is neither
$\mathbf{0}_M$ nor $\mathbf{1}_M$ itself.

\begin{minipage}{2.2in}
  
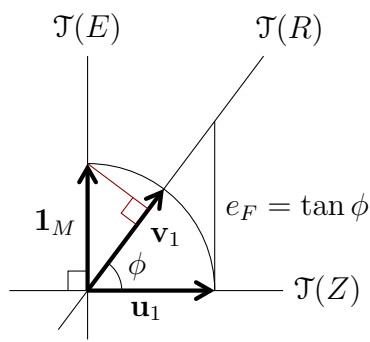
\captionof{figure}{The $(\mathbf{u}_1,\mathbf{1}_M)$-plane~~~}
  \label{u1-1m}
\begin{tikzpicture}[scale=1.3, >={angle 60}]
  \draw (0,-0.5) -- (0,2.4) node[above]{$\mathcal{T}(E)$};
  \draw (-0.8,0) -- (2.0,0) node[right]{$\mathcal{T}(Z)$};
  \draw[->, ultra thick] (0,0) -- (0,1.3);
  \draw (0,0.7) node[left]{$\mathbf{1}_M$};
  \draw[->, ultra thick] (0,0) -- (1.3,0);
  \draw (0.6,0) node[below]{$\mathbf{u}_1$};
  \draw[->, ultra thick] (0,0) -- (0.78,1.04);
  \draw (0.80,0.55) node{$\mathbf{v}_1$};
  \draw (1.3,0) arc (0:90:1.3);
  \draw (0.35,0) arc (0:53.13:0.35);
  \draw (0.30,0.25) node[right] {$\phi$};
  \draw (0,0.2) -- (-0.2,0.2) -- (-0.2,0);
  \draw (-0.3,-0.4) -- (1.8,2.4) node[above=10,right=-7]{$\mathcal{T}(R)$};
  \draw[darkred] (0,1.3) -- (0.624,0.832);
  \draw[darkred] (0.464,0.952) -- (0.344,0.792) -- (0.504,0.672);
  \draw (1.3,0) -- (1.3,1.733);
  \draw (1.3,0.85) node[right]{$e_F=\tan\phi$};
\end{tikzpicture}
\end{minipage}
\hfill\begin{minipage}{3.4in}
Let $\mathbf{v}_1\in\mathcal{T}(R)$ denote the unit vector in the
direction of the orthogonal projection of $\mathbf{1}_M$ onto
$\mathcal{T}(R)$, as shown in Figure \ref{u1-1m}. Then $\mathbf{v}_1$ is
the direction of steepest increase of expected return in the $R$-flat.
Changes in expected return depend only on changes of periodic return in
the $\mathbf{v}_1$ direction in the sense that
\begin{align}\notag
  \Delta e &= <\mathbf{1}_M,\bm{\Delta}\mathbf{r}>_\omega\\
    &= <\mathbf{1}_M,\mathbf{v}_1>_\omega
    <\mathbf{v}_1,\bm{\Delta}\mathbf{r}>_\omega
    \label{dedr}
\end{align}
for all $\bm{\Delta}\mathbf{r}\in\mathcal{T}(R)$.
\end{minipage}

Now set
\begin{align}\label{u1defn}
  \mathbf{u}_1&=\frac{\mathbf{v}_1
  -\mathbf{1}_M\langle\mathbf{1}_M,\mathbf{v}_1\rangle_\omega
  }{\|\mathbf{v}_1
  -\mathbf{1}_M\langle\mathbf{1}_M,\mathbf{v}_1\rangle_\omega
  \|_\omega}\in\mathcal{T}(Z)\\
\intertext{so that}\label{anglephi}
  \mathbf{v}_1&=\mathbf{u}_1\cos\phi+\mathbf{1}_M\sin\phi
  \quad\text{with}\quad 0<\phi<\frac{\pi}{2}
\end{align}
as shown in Figure \ref{u1-1m}.\\

\begin{propn}\label{deltae-u1}
  Let
  \begin{align}\notag
    \Delta e&=\langle\mathbf{1}_M,\bm{\Delta}\mathbf{r}\rangle_\omega,\\
    \bm{\Delta}\mathbf{z}&=\bm{\Delta}\mathbf{r} \notag
    - \mathbf{1}_M\Delta e,\\
    \intertext{for $\Delta\mathbf{r}\in\mathcal{T}(R)$. Then}
    \Delta e&= \label{deltae-deltaz}
    e_F\langle\mathbf{u}_1,\bm{\Delta}\mathbf{z}\rangle_\omega,
  \end{align}
  with the $\mathbf{u}_1$ of \eqref{u1defn} and $e_F=\tan\phi$ as in
  Figure \ref{u1-1m}.
\end{propn}

\begin{proof}
  The result follows from \eqref{dedr} and
  \begin{align*}
    \langle\mathbf{1}_M,\mathbf{v}_1\rangle_\omega&=\sin\phi, \\
    \langle\mathbf{v}_1,\bm{\Delta}\mathbf{r}\rangle_\omega&=
      \langle\mathbf{u}_1,\bm{\Delta}\mathbf{r}\rangle_\omega\cos\phi +
      \langle\mathbf{1}_M,\bm{\Delta}\mathbf{r}\rangle_\omega\sin\phi \\ &=
      \langle\mathbf{u}_1,\bm{\Delta}\mathbf{z}\rangle_\omega\cos\phi +
      \Delta e\sin\phi.
  \end{align*}
\end{proof}

Under the assumption $\mathbf{1}_M\notin\mathcal{T}(R)$, the change in
expected return, $\Delta e$, from one notional portfolio to another
depends only the change in risk component, $\bm{\Delta}\mathbf{z}$,
of the respective return vectors. This is a consequence of Proposition
\ref{isomorph}. Proposition \ref{deltae-u1} now shows that such a
change in expected return depends only on the change in the risk
component in the $\mathbf{u}_1$-direction. For this reason we refer to
the $\mathbf{u}_1$-direction of the Z-flat as the \emph{direction of
productive risk}. Changes in portfolio risk vectors in directions
orthogonal to the $\mathbf{u}_1$-direction have no effect on expected
reward. Such changes are \emph{nonproductive} in this sense.\\

\begin{cor}\label{cor-euz}
  \begin{equation}\label{euz}
    e=e_0+e_F\langle\mathbf{u}_1,\mathbf{z}\rangle_\omega\quad
    \text{for all}\quad
    \mathbf{r}=\mathbf{z}+\mathbf{1}_Me\in\mathcal{F}(R).
  \end{equation}
\end{cor}

\begin{proof}
Set $\bm{\Delta}\mathbf{r}=\mathbf{r}-\mathbf{r}_0$ in Proposition
\ref{deltae-u1}.
Then
\begin{align*}
  e-e_0&=e_F\langle\mathbf{u}_1,\mathbf{z}-\mathbf{z}_0\rangle_\omega\\
    &=e_F(\langle\mathbf{u}_1,\mathbf{z}\rangle_\omega
    -\langle\mathbf{u}_1,\mathbf{z}_0\rangle_\omega)\\
    &=e_F\langle\mathbf{u}_1,\mathbf{z}\rangle_\omega.
\end{align*}
Here $\langle\mathbf{u}_1,\mathbf{z}_0\rangle_\omega=0$
since $\mathbf{u}_1\in\mathcal{T}(Z)$ and $\mathbf{z}_0$
is orthogonal to $\mathcal{T}(Z)$.
\end{proof}

\vspace*{2.0ex}
\begin{cor}\label{cor_e0z_*}
  \begin{equation}\label{e0z_*}
    e_0 = e_*-e_F\langle\mathbf{u}_1,\mathbf{z}_*\rangle_\omega
  \end{equation}
  for any convenient
  $\mathbf{r}_*=\mathbf{z}_*+\mathbf{1}_Me_*\in\mathcal{F}(R)$.
\end{cor}
\setcounter{cor}{0}

\vspace*{1.0ex}
Let us now define the \emph{productive risk} of the system, $\tau_1$, as
\begin{equation}\label{defn-productive-risk}
  \tau_1=\sqrt{\sum_{j=1}^n
  \langle\mathbf{u}_1,\mathbf{z}_j-\mathbf{z}_0\rangle_\omega^2}
  =\sqrt{\sum_{j=1}^n\langle\mathbf{u}_1,\mathbf{z}_j\rangle_\omega^2}
\end{equation}
with $\tau_1^2$ being the \emph{productive variance}. We
include the middle, $\mathbf{z}_0$ expression in this definition to
emphasize that the productive risk is coming from the tangent space
$\mathcal{T}(Z)$, which is spanned by the $\mathbf{z}_j-\mathbf{z}_0$.
The middle expression collapses to the last expression because
$\langle\mathbf{u}_1,\mathbf{z}_0\rangle_\omega=0$.


\subsection{Nonproductive risk}\label{nonproductive-risk}

Each notional portfolio $\mathbf{p}$ has a corresponding risk
vector $\mathbf{z}=Z\mathbf{p}$.
By \eqref{euz} the expected return of the portfolio
is completely determined by the $\mathbf{u}_1$ coordinate of
$\mathbf{z}$ and the parameters $e_0$ and $e_F$.
However, the sum of the systemic and productive variances,
$f_0^2+\langle\mathbf{u_1},\mathbf{z}\rangle_\omega^2$, is just a
part of the portfolio variance. The remaining variance
is nonproductive, having no effect on the expected return
of the portfolio.\\

\vspace*{1.0ex}
\begin{defn}
\label{principal-npr}
We now define the \emph{principal nonproductive risks},
$\tau_i>0$, and the corresponding \emph{principal directions of
nonproductive risk}, $\mathbf{u}_i\in\mathcal{T}(Z)$
($\|\mathbf{u}_i\|_\omega=1$), for $i=2,\ldots,m$, where $m$ is the
dimension of $\mathcal{T}(Z)$, the rank of $\widehat{V}$. The definition
proceeds by induction:\\[1.0ex]
\hspace*{2.0ex}
for $i = 2,\ldots,m,$
\begin{align*}
  \hspace*{3.0ex}\tau_i^2 &=
  \sum_{j=1}^n\langle\mathbf{u}_i,\mathbf{z}_j\rangle_\omega^2 \\
  &= \max\left\{
  \sum_{j=1}^n\langle\mathbf{u},\mathbf{z}_j\rangle_\omega^2
  :\mathbf{u}\in\mathcal{T}(Z),\|\mathbf{u}\|_\omega=1,
  \langle\mathbf{u}_k,\mathbf{u}\rangle_\omega=0~
  (k=1,\ldots,i-1)\right\}
\end{align*}
\end{defn}

\begin{rmk*}
The $\tau_i=\sqrt{\tau_i^2}$ are uniquely determined, and, in the
generic case, when $\tau_2>\tau_3>\ldots>\tau_m$, the principal
directions of nonproductive risk are unique upto multiplication by $-$1.
We will assume this case to simplify the discussion. The
\textbf{rtndecomp} algorithm presented in Appendix \ref{octave-listing}
makes no such assumption.\\
\end{rmk*}

\begin{cor}\label{cor-total-variance-decomp}
The total nonsystemic variance can be decomposed into its
productive and nonproductive parts as
\begin{equation}\label{nonsys-variance-decomp}
  \hat{v}_\textnormal{T}=\tau_1^2+\sum_{i=2}^m\tau_i^2
\end{equation}
\end{cor}

\begin{cor}\label{cor-nonsyst-cov-factor}
  The nonsystemic covariance matrix $\widehat{V}$ factors as
  $\widehat{V}=F^TF$, where the coefficients of the $m\times n$
  factor matrix $F$ are given by
\begin{equation}\label{nonsyst-cov-factor}
  f_{ij}=\langle\mathbf{u}_i,\mathbf{z}_j-\mathbf{z}_0\rangle_\omega
  =\langle\mathbf{u}_i,\mathbf{z}_j\rangle_\omega\quad
  (i=1,\ldots,m;~ j=1,\ldots,n)
\end{equation}
\end{cor}

\begin{cor}\label{cor-portf-var-decomp}
The variance of return of any notional portfolio $\mathbf{p}$ can be
decomposed into its systemic, productive, and nonproductive parts as
\begin{equation}\label{portf-var-decomp}
  v_\textnormal{\textbf{p}}=f_0^2
    +(\sum_{j=1}^nf_{1j}p_j)^2+
    \sum_{i=2}^m(\sum_{j=1}^nf_{ij}p_j)^2,
\end{equation}
with the $f_{ij}$ of \eqref{nonsyst-cov-factor}.
\end{cor}
\setcounter{cor}{0}

These three corollaries of Definition \ref{principal-npr} follow
immediately from the preceding discussion. We refer to the factor matrix
$F$ of Corollary \ref{cor-nonsyst-cov-factor} as the \emph{nonsystemic risk
matrix}.


\subsection{Mean-variance analysis}\label{mv-analysis}

Let
\[
  \Phi:\mathds{R}^n\to\mathds{R}^2,\quad
  \mathbf{p}\mapsto(e_\textbf{p},v_\textbf{p}),
\]
denote the \emph{mean-variance mapping} defined by
\eqref{expected_return} and \eqref{var1}. In view of the preceding
discussion $\Phi$ can be factored as
\begin{equation}\label{factor_Phi}
  \Phi:\mathds{R}^n\xrightarrow{F}\mathds{R}^m\to\mathds{R}^2,\quad
  \mathbf{p}\mapsto\textbf{f}_\textbf{p}\mapsto(e_\textbf{p},v_\textbf{p}),
\end{equation}
with
\begin{equation}\label{e_of_x}
  e_\textbf{p} = e_0 + e_F x_\textbf{p},
\end{equation}
\begin{equation}\label{v_of_f}
  v_\textbf{p} = f_0^2 + \|\mathbf{f}_\textbf{p}\|^2 =
    f_0^2 + x_\textbf{p}^2 + \|\mathbf{y}_\textbf{p}\|^2,
\end{equation}
where
\begin{equation}\label{fpeqFp}
  \mathbf{f}_\textbf{p} =
  \begin{bmatrix}
    x_\textbf{p} \\ \mathbf{y}_\textbf{p}
  \end{bmatrix} =
  \begin{bmatrix}
    F(1, :) \\ F(2\!:\!m, :)
  \end{bmatrix} \mathbf{p} =
  F \mathbf{p}.
\end{equation}
Here $F(1, :)$ and $F(2\!:\!m, :)$ denote the productive and nonproductive
rows of the $m\times n$ nonsystemic risk matrix $F$, respectively.

We are primarily interested in the image, $\Phi(\bm{\Delta})$, of the
notional portfolio simplex
\[
  \bm{\Delta}=\{~\mathbf{p}\in\mathds{R}^n:
    \sum_{j=1}^np_j=1~\text{and}~p_j\ge0~\text{for}~j=1,\ldots,n\}.
\]
\cite{Markowitz:1987wd} refers to this image as the \emph{obtainable
$EV$ set}.\\

\begin{rmk*}
We apologize for the reuse of notation here. We have been using
$\bm{\Delta}$ to indicate a difference vector. Now $\bm{\Delta}$ is the
standard $(n-1)$-simplex in $\mathds{R}^n$. In the future we hope the
meaning of $\bm{\Delta}$ will be clear by its context.
\end{rmk*}

Example \ref{rd_example1} and the corresponding Figure
\ref{rd_factorize1} illustrate the factorization
\eqref{factor_Phi}--\eqref{v_of_f} in the $n=3, m=2$ case.

\begin{example}\label{rd_example1}
\begin{equation*}
  F =
  \left[\begin{array}{rrr}
    -4 &  2 &  \hspace{1.7ex}4 \\
     2 & -2 &  3 \\
  \end{array}\right]\quad\text{and}\quad \widehat{V} = F^TF=
  \left[\begin{array}{rrr}
    20 & -12 & -10 \\
   -12 &   8 &   2 \\
   -10 &   2 &  25 \\
  \end{array}\right].
\end{equation*}
\end{example}

\begin{figure}[ht]
  \centering
  \caption{\label{rd_factorize1}%
    \vspace{-3.0ex}%
    Factorization of the mean-variance mapping $\Phi$
  }
\begin{tikzpicture}[scale=0.9,>={angle 60}]
\small
\begin{scope}   
\begin{scope}
  \filldraw[fill=obtainable1,draw=boundary1] (0.0000,0.0000)
  -- (1.0000,-1.7321)
  -- (2.8000,-2.0785)
  -- (4.0000,0.0000) -- cycle;
  \filldraw[fill=obtainable2,draw=boundary2] (1.0000,-1.7321)
  -- (2.0000,-3.4641)
  -- (2.8000,-2.0785) -- cycle;
  \draw[->] (0.4553,-1.6272) -- (3.2500,-2.1651) node[right] {$x$};
  \draw[->] (1.2105,-2.7348) -- (2.1053,0.3646) node[above] {$y$};
  \draw[very thick] (0.0000,0.0000)
  -- (1.0000,-1.7321)
  -- (2.8000,-2.0785)
  -- (4.0000,0.0000);
  \fill (0.0000,0.0000) circle (2.5pt) node[above] {\textbf{A}(1, 0, 0)};
  \fill (2.0000,-3.4641) circle (2.5pt) node[below] {\textbf{B}(0, 1, 0)};
  \fill (4.0000,0.0000) circle (2.5pt) node[above] {\textbf{C}(0, 0, 1)};
  \filldraw[fill=red] (1.0000,-1.7321) circle (2.25pt) node[right=5,above=2] {\textbf{P}};
  \filldraw[fill=red] (2.8000,-2.0785) circle (2.25pt) node[left=5,above=1] {\textbf{Q}};
  \filldraw[fill=green] (1.4737,-1.8232) circle (2.25pt) node[right=6,below=2] {\textbf{E}};
\end{scope}
\begin{scope}[xshift=11.0cm,yshift=-3.5cm]
  \filldraw[fill=obtainable1,draw=boundary1] (-2.0000,3.7500)
  .. controls (-1.5000,1.7500) and (-1.0000,0.5625) .. (-0.5000,0.1875)
  .. controls (0.1333,-0.2875) and (0.7667,0.1400) .. (1.4000,1.4700)
  .. controls (1.6000,1.8900) and (1.8000,2.9625) .. (2.0000,4.6875)
  .. controls (0.6667,0.3125) and (-0.6667,0.0000) .. (-2.0000,3.7500);
  \filldraw[fill=obtainable2,draw=boundary2] (-0.5000,0.1875)
  .. controls (0.0000,-0.1875) and (0.5000,0.2500) .. (1.0000,1.5000)
  .. controls (1.1333,1.2000) and (1.2667,1.1900) .. (1.4000,1.4700)
  .. controls (0.7667,0.1400) and (0.1333,-0.2875) .. (-0.5000,0.1875);
  \draw[very thick] (-2.0000,3.7500)
  .. controls (-1.5000,1.7500) and (-1.0000,0.5625) .. (-0.5000,0.1875)
  .. controls (0.1333,-0.2875) and (0.7667,0.1400) .. (1.4000,1.4700)
  .. controls (1.6000,1.8900) and (1.8000,2.9625) .. (2.0000,4.6875);
  \fill (-2.0000,3.7500) circle (2.5pt) node[above=2.5] {\textbf{A}};
  \fill (1.0000,1.5000) circle (2.5pt) node[right=4,above=2] {\textbf{B}};
  \fill (2.0000,4.6875) circle (2.5pt) node[above=2] {\textbf{C}};
  \filldraw[fill=red] (-0.5000,0.1875) circle (2.25pt) node[left=8,below=0] {\textbf{P}};
  \filldraw[fill=red] (1.4000,1.4700) circle (2.25pt) node[right=2] {\textbf{Q}};
  \draw[->] (0,0) -- (0,-0.5) -- (1.0,-0.5);
  \draw (0.78,-0.5) node[above=8,left=-24] {efficient};
  \filldraw[fill=green] (0.0000,0.0000) circle (2.25pt) node[above=2] {\textbf{E}};
\end{scope}
\end{scope}  
\begin{scope}[xshift=2cm, yshift=-9.3cm]  
{ \normalsize
  \node at (4.40,9.0) {$\Phi$};
  \node at (4.75,8.5) {$\mathbf{p}\mapsto(e,v)$};
  \node at (0,4.6) {$\mathbf{p}$};
  \node at (-0.5,4) {$F$};
  \node[rotate=-90] at (0,4) {$\mapsto$};
  \node at (0,3.3) {$(x,y)$};
  \node at (4.75,2.5) {$(x,y)\mapsto(x,|y|)$};
  \node at (9,4.6) {$(e,v)$};
  \node[rotate=90] at (9,4) {$\mapsto$};
  \node at (9,3.3) {$(x,|y|)$};
  \draw[-] (4.00,7.8) -- (5.5,7.8);
  \node at (4.75,7.2) {$e=e_0+e_Fx$};
  \node at (4.75,6.6) {$v=f_0^2+x^2+\|\mathbf{y}\|^2$};
  \draw[<->] (7.0,5.5) node[above] {$v$} -- (7.0,5.0) -- (7.5,5.0) node[right] {$e$};
}
\begin{scope}[scale=0.5]
  \draw[step=1.0cm,help lines] (-4, -2) grid (4, 3);
  \filldraw[fill=obtainable1,draw=boundary1] (-4,2)
  -- (-1,0) -- (2.8,0) -- (4,3) -- cycle;
  \filldraw[fill=obtainable2,draw=boundary2] (-1,0)
  -- (2,-2) -- (2.8,0) -- cycle;
  \draw[->] (-4.5,0) -- (5.0,0) node[right] {$x$};
  \draw[->] (0,-2.5) -- (0,4.0) node[above] {$y$};
  \draw[very thick] (-4,2) -- (-1,0) -- (2.8,0) -- (4,3);
  \fill (-4,2) circle (5.0pt) node[left=1.5] {\textbf{A}};
  \fill (2,-2) circle (5.0pt) node[below=2] {\textbf{B}};
  \fill (4,3) circle (5.0pt) node[above=2] {\textbf{C}};
  \filldraw[fill=red] (-1,0) circle (4.55pt) node[right=5,above=1.5] {\textbf{P}};
  \filldraw[fill=red] (2.8,0) circle (4.55pt) node[left=5,above=0.5] {\textbf{Q}};
  \filldraw[fill=green] (0,0) circle (4.55pt) node[right=7.5,above=1.5] {\textbf{E}};
\end{scope}
\begin{scope}[scale=0.5,xshift=18cm]
  \draw[step=1.0cm,help lines] (-4, 0) grid (4, 3);
  \filldraw[fill=obtainable1,draw=boundary1] (-4,2)
  -- (-1,0) -- (2.8,0) -- (4,3) -- cycle;
  \filldraw[fill=obtainable2,draw=boundary2] (-1,0)
  -- (2,2) -- (2.8,0) -- cycle;
  \draw[->] (-4.5,0) -- (5.0,0) node[right] {$x$};
  \draw[->] (0,-0.0) -- (0,4.0) node[above] {$|y|$};
  \draw[very thick] (-4,2) -- (-1,0) -- (2.8,0) -- (4,3);
  \fill (-4,2) circle (5.0pt) node[left=1.5] {\textbf{A}};
  \fill (2,2) circle (5.0pt) node[right=1.5] {\textbf{B}};
  \fill (4,3) circle (5.0pt) node[above=2] {\textbf{C}};
  \filldraw[fill=red] (-1,0) circle (4.55pt) node[left=5,below=1.5] {\textbf{P}};
  \filldraw[fill=red] (2.8,-0) circle (4.55pt) node[right=5,below=1.5] {\textbf{Q}};
  \filldraw[fill=green] (0,0) circle (4.55pt) node[below=1.5] {\textbf{E}};
\end{scope}
\end{scope}   
\end{tikzpicture}
\end{figure}
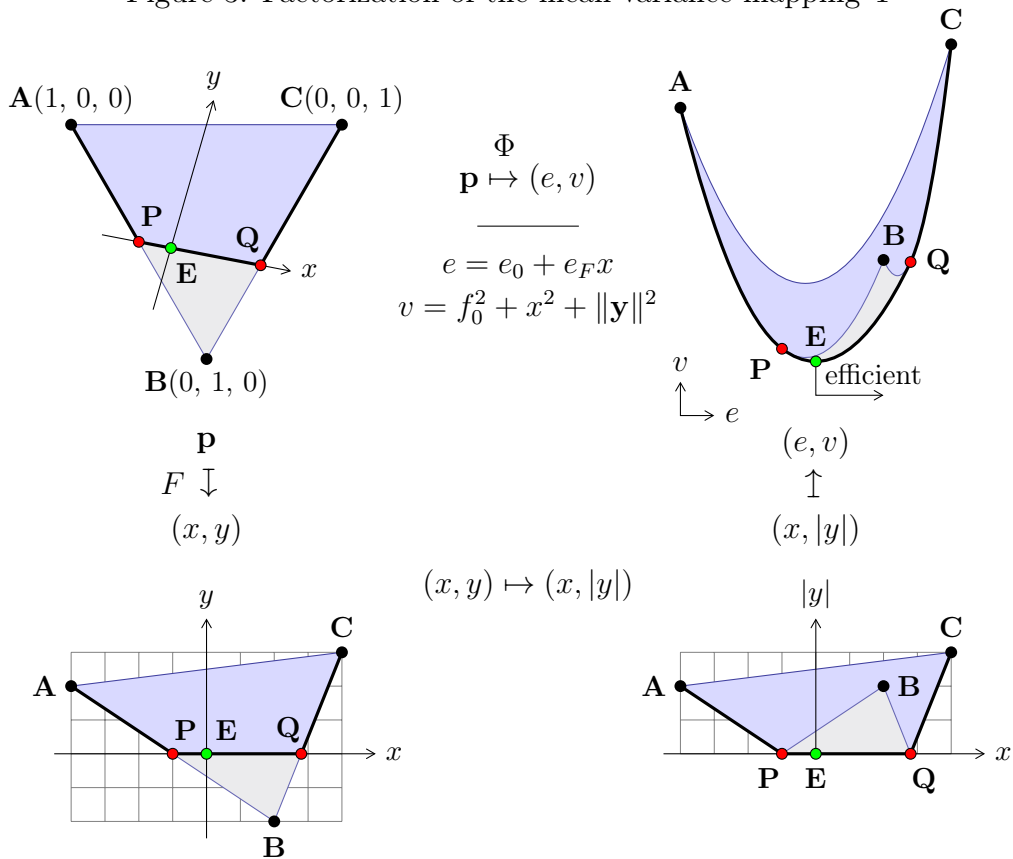

In Figure \ref{rd_factorize1} the set of minimum-variance
portfolios is specified by the piecewise linear path \textbf{APQC}
though the simplex $\bm{\Delta}$. The portfolios
\begin{align*}
  \textbf{P}&=50\%\,\textbf{A}+50\%\,\textbf{B}\\[-5.0ex]
\intertext{\vspace{-2.0ex}and}
  \textbf{Q}&=60\%\,\textbf{B}+40\%\,\textbf{C}
\end{align*}
are called \emph{corner portfolios} for obvious reasons.
Equation \eqref{v_of_f} implies that the portfolio \textbf{E} of absolute
minimum variance corresponds to $(x,y)=(0,0)$. Since $\textbf{P}$
and $\textbf{Q}$ have $(x,y)$-representations (-1, 0) and (2.8, 0),
respectively, we must have
\begin{equation*}
  \textbf{E}=\frac{14}{19}\,\textbf{P}+\frac{5}{19}\,\textbf{Q}.
\end{equation*}
The parameters $~e_0,~ e_F\ge0,$ and $f_0\ge0$ are inconsequential. The
set of minimum-variance portfolios is independent of these parameters.

Portfolios on the piecewise linear path \textbf{EQC} in $\bm{\Delta}$
are \emph{efficient}: besides having minimum variance for their expected
return, they have maximum expected return for their variance.

The $x$- and $y$-axes through the upper-left-hand simplex of
Figure \ref{rd_factorize1} are the preimages of the respective axes on
the lower-left $xy$-plane via the mapping $\mathbf{p}\mapsto
(x,y)=F\mathbf{p}$. The image axes are perpendicular to each other, but
the preimage axes are not. The preimage of the $x$-axis is the
\emph{critical line} of the mapping $\Phi|\{\sum_{j=1}^3p_j=1\}$. The
derivative of $\Phi|\{\sum_{j=1}^3p_j=1\}$ has rank 1 along this line
and rank 2 everywhere else. In effect $\Phi$ folds the
$\sum_{j=1}^3p_j=1$ plane over this critical line.

Returning to the general situation let us point out that the
factorization of the mean variance mapping in
\eqref{factor_Phi}--\eqref{v_of_f} leads to a natural, geometric
characterization of minimum-variance portfolios. First note that the
portfolio simplex $\bm{\Delta}$ is mapped onto a convex polytope
$P=F(\bm{\Delta})$ in $(x,\mathbf{y})$-space $\mathds{R}^m$ due to the
linearity of $\mathbf{p}\mapsto F\mathbf{p}$. If the $x$-axis (the
$\mathbf{u}_1$-axis) passes through this polytope, then every point in
the intersection of the $x$-axis and the polytope is the image of a
minimum-variance portfolio $\mathbf{p}$---simply because
$\mathbf{y}_\textbf{p}=\mathbf{0}$ and $v_\textbf{p}$ can't get any
smaller than $v_\textbf{p}=f_0^2+x_\textbf{p}^2$. More generally, let
$x_\text{min}$ and $x_\text{max}$ be the minimum and maximum values of
the coefficients in the ``productive'' row, $X = F(1, :)$, of $F$. Given
$x_*$ between $x_\text{min}$ and $x_\text{max}$, suppose
$\mathbf{y}_*\in\mathds{R}^{m-1}$ satisfies
\begin{equation}\label{y*}
  \|\mathbf{y}_*\| = \min\left\{\|\mathbf{y}\|:
  (x_*,\mathbf{y})\in P\cap\{x=x_*\}\right\}.
\end{equation}
Then any portfolio $\mathbf{p}\in\bm{\Delta}$ that
$F$ maps onto $(x_*,\mathbf{y}_*)$ (and there is at least one) is a
minimum-variance portfolio.


\subsection{Relaxing assumptions}\label{relaxing}

Since Section \ref{productive-risk} we have been assuming that the $n$
given securities do not all have the same expected return; however there
is no problem if the returns are identical. Then there is no productive
risk, all nonsystemic risk is nonproductive:
$\mathbf{u}_1,\ldots,\mathbf{u}_m$ are the principal directions of
nonproductive risk. This case is signalled by $e_F=0$, and \eqref{euz}
still holds with $e_0$ being the common expected return.

The case when $\mathbf{1}_M$ parallels the R-flat, when
$\mathbf{1}_M\in\mathcal{T}(R)$, is more problematic. This situation
typically arises when there are more securities than periods. Then there
is no unambiguous systemic return, $e_0$, and no well-defined gradient
of expected return, $\mathbf{g}=\mathbf{u}_1e_F\in\mathcal{T}(Z)$.

To handle the $\mathbf{1}_M\in\mathcal{T}(R)$ case we anchor ourselves
at the mean risk component,
$\bar{\mathbf{z}}=\frac{1}{n}\sum_{j=1}^n\mathbf{z}_j$, with mean
expected return, $\bar{e}=\frac{1}{n}\sum_{j=1}^ne_j$. Our
approximate gradient,
$\mathbf{g}=\sum_{j=1}^n(\mathbf{z}_j-\bar{\mathbf{z}})g_j$, is the
least-squares solution of \eqref{deltae-deltaz} in the form
\begin{equation}\label{lsq-deltae-deltaz}
  e_k-\bar{e}=\sum_{j=1}^n
  \langle\mathbf{z}_k-\bar{\mathbf{z}},\mathbf{z}_j-\bar{\mathbf{z}}
  \rangle_\omega\, g_j\quad(k=1,\ldots,n).
\end{equation}
We use this $\mathbf{g}$ to define the parameters 
\begin{equation}\label{param-approx}
  e_F=\|\mathbf{g}\|_\omega,\quad
  \mathbf{u}_1=\mathbf{g}/e_F,\quad
  e_0=\bar{e}-\langle\mathbf{g},\bar{\mathbf{z}}\rangle_\omega,
\end{equation}
for the approximate version of \eqref{euz}:
\begin{equation}\label{euz-approx}
    e\approx e_0+e_F\langle\mathbf{u}_1,\mathbf{z}\rangle_\omega.
\end{equation}
Note that the definition of $e_0$ is essentially \eqref{e0z_*}
with $e_*=\bar{e}$ and $\mathbf{z}_*=\bar{\mathbf{z}}$.

The remaining orthogonal directions of risk are defined inductively by
Definition \ref{principal-npr}, and the corollaries of that definition
continue to hold.


\subsection{Scaling output}\label{scaling}

Up to this point expected returns, $e_0$ and $E=[e_1,\ldots,e_n]$, and
risk coefficients, $f_0$ and $F=[f_{ij}]~ (i=1,\ldots,m;j=1,\ldots,n)$,
have been measured in the same percent-per-period units. While days or
weeks may be used for computational purposes, annualized,
percent-per-year output is usually preferred to daily or weekly
percentages.

To compensate for this preference we add a periods-per-unit-of-time
parameter $\rho$ to the periodic returns $R$ and weights $\omega$
required by our algorithm. Then, at the end of the computations,
percent-per-period expected returns and risks are scaled to
percent-per-unit-of-time units as follows:
\begin{equation*}
\begin{array}{ccccc}
  \text{percent}&&&&\text{percent}\\
  \text{per-unit-of-time}&&&&\text{per-period}\\
  E   &\leftarrow&\rho&\times& E   \\
  e_0 &\leftarrow&\rho&\times& e_0 \\
  e_F &\leftarrow&\sqrt{\rho}&\times& e_F \\
  F   &\leftarrow&\sqrt{\rho}&\times& F \\
  f_0 &\leftarrow&\sqrt{\rho}&\times& f_0 
\end{array}
\end{equation*}

The idea behind this scaling is statistical. Assume, for example, that
daily returns are independent random variables from one market-day to
the next and there are (typically) $\rho=252$ market-days per year. The
annual return is the sum of $\rho$ daily returns; so the expected value
of annual return is $\rho$ times the daily expected value. This accounts
for the $\rho$ multipliers above. The variance of annual return is
$\rho$ times the daily variance due to the independence assumption, but
risk or standard deviation is the square root of variance; consequently
$\sqrt{\rho}$ is the appropriate multiplier of $F$ and $f_0$. Finally
$e_F$ is the rate of change of expected return to risk; so
$\sqrt{\rho}=\rho/\sqrt{\rho}$~ is the appropriate multiplier.


\section{The rtndecomp function -- arguments and relationships}
\label{rtndecomp_function}

The \href{http://www.gnu.org/software/octave/}{GNU Octave}
listing of the \textbf{rtndecomp} function appears in
Appendix \ref{octave-listing}. In this section we give the function
header and describe its arguments. The relationships between the
output arguments were derived in the last section.

\makebox[1.0in][l]{\textbf{function:}}%
$[E, F, f_0, e_0, e_F] = \textbf{rtndecomp}(R, \bm{\omega}, \rho)$

\makebox[1.0in][l]{\textbf{purpose}}%
\parbox[t]{4.5in}{%
To decompose financial return data into orthogonal risk-factors.}

\textbf{input}
\begin{adescript}
  \item[R] $M\times n$ matrix of periodic returns.
  \item[\bm{\omega}] $M$-vector of positive weights or a scalar.\\
    If $\bm{\omega}$ is a scalar or if $R$ is the only input argument,\\
    then $\bm{\omega}$ defaults to\\
    \hspace*{5.0ex}$\omega_i = 1/M$ ~for~ $i = 1,\ldots,M$.
  \item[\rho] periods per unit time. ($\rho\ge1,$ default: $\rho = 1$)\\
    e.g., $\rho=252$ market-days per year.
\end{adescript}

\textbf{output}
\begin{adescript}
  \item[E] $1\times n$ matrix of expected returns.
  \item[F] $m\times n$ matrix of risk coefficients.\\
    $\rank(F)=m$ unless $F=\zeros(1,n)$.
  \item[f_0] systemic risk. ($f_0\ge0$)
  \item[e_0] systemic expected return.
  \item[e_F] expected return per unit of productive risk.\\
    ($e_F\ge0$; if $e_F=0$ there is no productive risk)
\end{adescript}

\textbf{global output}
\begin{adescript}
  \item[\textit{eflag}] \textbf{true} if and only if a nonzero, constant
  $M$-vector is\\parallel to the returns flat, $\mathcal{F}(R)$, or,
  said another way,\\if and only if $\mathbf{1}_M\in\mathcal{T}(R)$.
\end{adescript}

\textbf{relationships}
\begin{aenum}
  \item $E = \rho\,\bm{\omega}^T R.$
  \item The $n\times n$ covariance of returns matrix, $V,$ is given by\\
    \hspace*{3.0ex}$V = Z^T\diag(\bm{\omega}\rho)\,Z,$ where
    $Z = R - \mathbf{1}_M\bm{\omega}^T R.$
  \item $V = f_0^2 + F^TF.$ (the scalar $f_0^2$ is added to each
    coefficient of $F^TF$)
  \item $E = e_0 + e_F X$ unless \textit{eflag} is \textbf{true}.\\
    Here $X = F(1, 1:n)$ denotes the first row of $F,$ and, in the
    equation, $e_0$ is added to each coefficient of $e_F X.$ When
    \textit{eflag} is \textbf{true}, the equation is an approximation,
    but\\
    \hspace*{3.0ex}$\mean(E) = e_0 + e_F\mean(X)$\\remains true.
  \item $v_T = \sum_{j=1}^n V(j. j)
    = nf_0^2 + \tau_1^2 + \sum_{i=2}^m\tau_i^2,$
    with $\tau_i = \sqrt{F(i,1:n) F(i,1:n)^T}$\\for
    $i=1,\ldots,m.$
    The right-hand side of the equation for $v_T$ represents the
    decomposition of total variance into systemic, productive, and
    nonproductive parts (though $\tau_1^2$ is nonproductive
    if $e_F=0$).
  \item $\tau_2\ge\ldots\ge\tau_m>0.$
    (principal nonproductive risks)
  \item $F(2:m,1:n)F(2:m, 1:n)^T = \diag(\tau_2^2,\ldots,\tau_m^2).$
  \item If $e_F=0$ and $F\ne[0,\ldots,0]$,
    then $\tau_1\ge\ldots\ge\tau_m>0$\\
    and
    $F(1:m,1:n)F(1:m, 1:n)^T = \diag(\tau_1^2,\ldots,\tau_m^2).$
\end{aenum}


\section{Examples of output}
\label{output-examples}


\subsection{Five large ETFs -- 2010}
\label{5largeETFs}

In our first example we apply the \textbf{rtndecomp} algorithm (listed
in Appendix \ref{octave-listing}) to 2010 daily returns from the five
iShares exchange traded funds (ETFs)
\begin{compactenum}\label{fivefunds}
  \item \ticker{IEF} -- iShares Barclays 7-10 Year Treasury Bond Fund
  \item \ticker{IWB} -- iShares Russell 1000 Index Fund
  \item \ticker{IWM} -- iShares Russell 2000 Index Fund
  \item \ticker{EFA} -- iShares MSCI EAFE Index Fund
  \item \ticker{EEM} -- iShares MSCI Emerging Markets Index Fund
\end{compactenum}
Figure \ref{growth5_2010} shows the growth of these securities in
2010. These are plots of adjusted closing prices against time. The
adjusted closing prices are normalized at 100 on 2010-12-31; thus
notional portfolios specify the closing proportions of actual investment
portfolios at the end of 2010 (\cite{Norton:2011fk}). If ex post
analysis were to deem a certain notional portfolio $\mathbf{p}$ as
optimal, one would buy
\begin{equation}\label{shares}
  s_j = 100\,p_j/a_j \quad
\end{equation}
shares of fund $j$, per \$100 invested, to invest in the optimal
portfolio at the end of 2010. Here $a_j$ is the 2010-12-31
closing price of fund $j~ (j=1,\dots,n;~ n=5)$.
\begin{figure}[H]
  \centering
  \caption{\label{growth5_2010}%
    2010 adjusted closing prices of five large ETFs\\
    prices normalized at 100 on 2010-12-31}
  \includegraphics{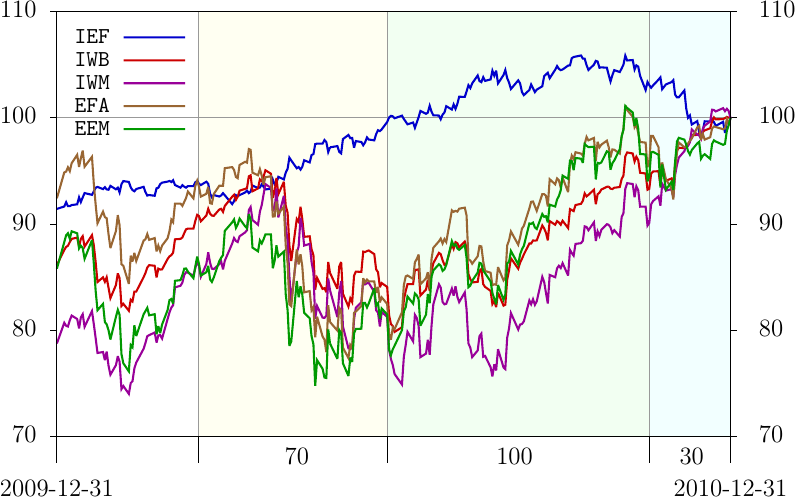}
\end{figure}


\subsubsection{Algorithm input}\label{input}

The graphs in Figure \ref{growth5_2010} correspond to a $253\times5$
matrix $A$ of adjusted closing prices for these 5 ETFs on the 253
market days from 2009-12-31 thru 2010-12-31 inclusive. The normalized,
linear, daily returns for these funds are given by the $252\times5$
matrix $R=\Delta A$, where $\Delta$ is the $252\times253$
difference operator
\begin{equation*}
  \Delta = \left [
\begin{array}{rrrrr}
    -1 & 1 & & \multicolumn{2}{c}{}\\
     & -1 & 1 & \multicolumn{2}{c}
     {\raisebox{1.5ex}[0pt]{\BigZero}}\\
     & & \ddots & \ddots &\\
     \multicolumn{2}{c}{\raisebox{1.5ex}[0pt]{\BigZero}}
     & & -1 & 1
\end{array}
  \right ].
\end{equation*}
We won't use all of $R$ in this example, just the last $M=200$ rows,
just the last 200 market-day returns---which correspond to the
colored-background portion of Figure \ref{growth5_2010}.

In addition to the $M\times n$ return matrix $R$, the
\mbox{\textbf{rtndecomp}} algorithm requires an $M$-vector of weights,
$\bm{\omega}$, and a scaling factor, $\rho$, that specifies the number
of periods per unit time. We use $\rho=252$ market-days per year
throughout this paper.

To see how the weights affect output we will consider two systems
of weights: gray, uniform weights, where each of the $M=200$
market days has the same importance, $\omega_i=1/200$, and the more
colorful \emph{late-heavy} weight system pictured in Figure
\ref{weights-figure}.

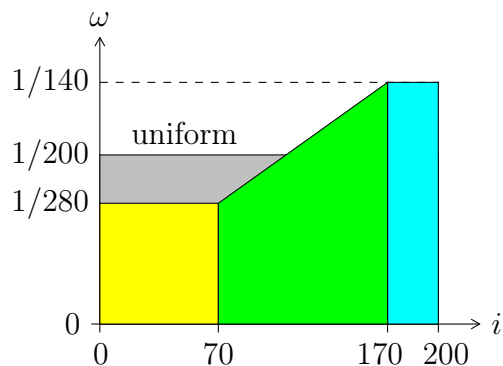
\begin{floatingfigure}[r]{7.0cm}
  \caption{Return weights for 200 
  market days. Late-heavy weights colored.}
  \label{weights-figure}\vspace{-3.0ex}
  \centering
\begin{tikzpicture}[scale=2.25, >={angle 60}]
  \draw [dashed] (0.0,1.429) -- (2.0,1.429);
  \draw[->] (0.0,-0.05) node [below=8,left=-7.5] {0}
    -- (0.0,1.7) node [above] {$\omega$};
  \draw[->] (-0.05,0.0) node [left] {0}
    -- (2.25,0.0) node [right] {$i$};
  \filldraw [fill = uniform] (0.0,0.0) rectangle (2.0,1.0);
  \draw (0.5,1.0) node [above] {uniform};
  \draw (0.0,1.0) node [left] {$1/200$};
  \filldraw [fill = first70] (0.0,0.0) rectangle (0.7,0.714);
  \draw (0.0,0.714) node[left] {$1/280$};
  \draw (0.7,0.0) -- (0.7,-0.05) node [below=8,left=-10] {70};
  \filldraw [fill = last30] (1.7,0.0) rectangle (2.0,1.429);
  \draw (0.0, 1.429) node[left] {$1/140$};
  \draw (1.7,0.0) -- (1.7,-0.05) node [below=8,left=-10] {170};
  \draw (2.0,0.0) -- (2.0,-0.05) node [below=8,right=-10] {200};
  \filldraw [fill = middle100] (0.7,0.0) -- (1.7,0.0) -- (1.7,1.429)
    -- (0.7,0.714) -- cycle;
\end{tikzpicture}
\end{floatingfigure}

The colors of the late-heavy system correspond to the colored regions of
Figure \ref{growth5_2010}. The returns of the last 30 market days
of 2010 are weighted by 1/140 each; returns for the first 70 days count
half as much as these, or 1/280 per day; and the weights for the middle
100 days,
\begin{equation*}
  \omega_{i+70} = \frac{1+i/101}{280}~~ \text{for}~~
  i=1,\ldots,100,
\end{equation*}
increase uniformly between these two extremes. The sum of these
late-heavy weights is the sum of the yellow, green, and blue areas in
Figure
\ref{weights-figure}:
\begin{equation*}
  70\times\frac{1}{280}+100\times\frac{1}{2}\times
  \left(\frac{1}{280}+\frac{1}{140}\right)
  +30\times\frac{1}{140}=1,
\end{equation*}
as required.

The idea behind late-heavy weighting is simple. Think of investing at the
end of a 200 market-day period. The recent performance of a group of
securities may be more important than their performance further back in the
past---%
as a predictor of their near-future performance. Thus an investment analysis
might weight the recent performance more heavily.


\subsubsection{Algorithm output}\label{algorithm-output}

The output arguments of the \textbf{rtndecomp} function
\begin{equation}\label{rtndecomp}
  [E, F, f_0, e_0, e_F] = \textbf{rtndecomp}(R,\bm{\omega},\rho)
\end{equation}
were described in the Section \ref{rtndecomp_function}. Tables
\ref{decomp5_2010L200} and \ref{decomp5_2010U200} show the output when
the function is applied to the $200\times5$ matrix of daily returns
described above. The late-heavy weights of Figure \ref{weights-figure}
were used for Table \ref{decomp5_2010L200} and uniform weights for Table
\ref{decomp5_2010U200}.

\begin{table}[H]
  \centering
  \caption{\label{decomp5_2010L200}%
    Decomposition of return data -- 5 large ETF universe\\
    last 200 market-days of 2010 -- late-heavy weights
  }
{ \small
  \newcolumntype{.}{D{.}{.}{2}}
  \newcolumntype{W}{|>{\columncolor{white}}c|}
$
\begin{array}{|c|.....||r.|}
\cline{1-6}\rule{0mm}{4mm}
  \text{fund}
  & \multicolumn{1}{r}{\parbox{6.5ex}{\hfill\ticker{IEF}\hspace*{0.2ex}}}
  & \multicolumn{1}{r}{\parbox{7.0ex}{\hfill\ticker{IWB}\hspace*{0.2ex}}}
  & \multicolumn{1}{r}{\parbox{7.0ex}{\hfill\ticker{IWM}\hspace*{0.2ex}}}
  & \multicolumn{1}{r}{\parbox{7.0ex}{\hfill\ticker{EFA}\hspace*{0.2ex}}}
  & \multicolumn{1}{r|}{\parbox{7.0ex}{\hfill\ticker{EEM}\hspace*{0.8ex}}} \\
\hline\rule[-1.5mm]{0mm}{6mm}
  E &   2.86 &  18.77 &  27.12 &  13.84 &  21.79
  & \multicolumn{2}{c|}{\widehat{V}_\text{T}} \\
\hline\rowcolor{productive}
  \multicolumn{1}{W}{\rule{0mm}{4mm}}
  & -2.79 & 7.80 & 13.36 & 4.52 & 9.81 & 364 & 24.8\% \\
\cline{2-8}
  \rowcolor{nonproductive}\multicolumn{1}{W}{\rule{0mm}{4mm}}
  & 4.14 & -11.81 & -13.48 & -19.57 & -16.81 & 1004 & 68.4\% \\
  \rowcolor{nonproductive}\multicolumn{1}{W}{\raisebox{1.5ex}[0pt]{$F$}}
  & 2.30 & -2.15 & -5.06 & 0.01 & 6.12 & 73 & 5.0\% \\
  \rowcolor{nonproductive}\multicolumn{1}{W}{}
  & -4.14 & 2.02 & -2.26 & -0.89 & 0.40 & 27 & 1.9\% \\
\hhline{|=|=====#==|}\rule{0mm}{4mm}
  & \multicolumn{1}{r@{\hspace*{1.1ex}}}{47}
  & \multicolumn{1}{r@{\hspace*{1.1ex}}}{209}
  & \multicolumn{1}{r@{\hspace*{1.1ex}}}{391}
  & \multicolumn{1}{r@{\hspace*{1.1ex}}}{404}
  & \multicolumn{1}{r@{\hspace*{1.1ex}}||}{416}
  & \multicolumn{1}{r}{1468}
  & \multicolumn{1}{r|}{100\%\hspace*{-0.8ex}} \\
\cline{8-8} \raisebox{1.6ex}[0pt]{$\widehat{V}_\text{T}$}
  & 3.2\% & 14.2\% & 26.6\% & 27.5\% & 28.4\%
  & \multicolumn{1}{r|}{100\%\hspace*{-0.5ex}}\\
\cline{1-7}
  \multicolumn{7}{c}{\rule{0mm}{5mm}
    \text{with}~f_0 = 5.07,~ e_0 = 7.05,~\text{and}~e_F = 1.502
  }
\end{array}
$
}
\end{table}

\begin{table}[H]
  \centering
  \caption{\label{decomp5_2010U200}%
    Decomposition of return data -- 5 large ETF universe\\
    last 200 market-days of 2010 -- uniform weights
  }
{ \small
  \newcolumntype{.}{D{.}{.}{2}}
  \newcolumntype{W}{|>{\columncolor{white}}c|}
$
\begin{array}{|c|.....||r.|}
\cline{1-6}\rule{0mm}{4mm}
  \text{fund}
  & \multicolumn{1}{r}{\parbox{6.5ex}{\hfill\ticker{IEF}\hspace*{0.2ex}}}
  & \multicolumn{1}{r}{\parbox{7.0ex}{\hfill\ticker{IWB}\hspace*{0.2ex}}}
  & \multicolumn{1}{r}{\parbox{7.0ex}{\hfill\ticker{IWM}\hspace*{0.2ex}}}
  & \multicolumn{1}{r}{\parbox{7.0ex}{\hfill\ticker{EFA}\hspace*{0.2ex}}}
  & \multicolumn{1}{r|}{\parbox{7.0ex}{\hfill\ticker{EEM}\hspace*{0.8ex}}} \\
\hline\rule[-1.5mm]{0mm}{6mm}
  E &   7.79 &  11.62 &  17.47 &  8.12 & 17.27
  & \multicolumn{2}{c|}{\widehat{V}_\text{T}} \\
\hline\rowcolor{productive}
  \multicolumn{1}{W}{\rule{0mm}{4mm}}
  & 0.00 & 3.02 & 7.62 & 0.27 & 7.46 & 123 & 7.4\% \\
\cline{2-8}
  \rowcolor{nonproductive}\multicolumn{1}{W}{\rule{0mm}{4mm}}
  & 5.41 & -15.08 & -18.63 & -20.92 & -19.15 & 1408 & 85.3\% \\
  \rowcolor{nonproductive}\multicolumn{1}{W}{\raisebox{1.5ex}[0pt]{$F$}}
  & -1.74 & 2.31 & 6.80 & -2.94 & -5.71 & 96 & 5.8\% \\
  \rowcolor{nonproductive}\multicolumn{1}{W}{}
  & -3.33 & 2.09 & -1.48 & -2.17 & 1.22 & 24 & 1.4\% \\
\hhline{|=|=====#==|}\rule{0mm}{4mm}
  & \multicolumn{1}{r@{\hspace*{1.1ex}}}{43}
  & \multicolumn{1}{r@{\hspace*{1.1ex}}}{246}
  & \multicolumn{1}{r@{\hspace*{1.1ex}}}{453}
  & \multicolumn{1}{r@{\hspace*{1.1ex}}}{451}
  & \multicolumn{1}{r@{\hspace*{1.1ex}}||}{456}
  & \multicolumn{1}{r}{1651}
  & \multicolumn{1}{r|}{100\%\hspace*{-0.8ex}} \\
\cline{8-8} \raisebox{1.6ex}[0pt]{$\widehat{V}_\text{T}$}
  & 2.6\% & 14.9\% & 27.5\% & 27.3\% & 27.7\%
  & \multicolumn{1}{r|}{100\%\hspace*{-0.5ex}}\\
\cline{1-7}
  \multicolumn{7}{c}{\rule{0mm}{5mm}
    \text{with}~f_0 = 4.86,~ e_0 = 7.78,~\text{and}~e_F = 1.271
  }
\end{array}
$
}
\end{table}

The $E, e_0, F$, and $f_0$ coefficients in these tables are in
percent-per-year units, the slope $e_F$ is unitless,
and the nonsystemic variance totals in the
$\widehat{V}_\text{T}$ sections are in
percent-per-year-squared units.

As described in Section \ref{risk-components}, the coefficients of each
risk matrix $F$ are the coordinates of the nonsystemic risk components
of the individual securities with respect to an orthonormal basis,
$\{\mathbf{u}_1,\ldots,\mathbf{u}_m\}$, for the
tangent space of the $Z$-flat. The sum of the squares of these
coefficients is the \emph{total nonsystemic variance} of the system. The
$\widehat{V}_\text{T}$-row of either table shows how this nonsystemic
variance is distributed among the individual funds. As one might expect,
the nonsystemic variance of the bond fund, \texttt{IEF}, is substantially
less than that of any of the equity funds---under either system of weights.

The green row, $X=[x_1,\ldots,x_n]=F(1,:)$, of each $F$ contains the
\emph{productive} risk coefficients, the $\mathbf{u}_1$-coordinates, of
the security risk vectors. Changes in productive risk from one notional
portfolio to another produce corresponding changes in expected return.
If there is no change in productive risk, there is no change in expected
return.

The orange rows, $F(2\!:\!m, :)$, of each $F$ contain the risk
coefficients in the principal directions of nonproductive risk. These
rows are pairwise orthogonal (up to roundoff error). The first orange
row, $Y=[y_1,\ldots,y_n]=F(2,:)$, represents the most significant or
major direction of nonproductive risk. It contains the
$\mathbf{u}_2$-coordinates of the security risk vectors. The
$\widehat{V}_\text{T}$-column of each table shows how the total
nonsystemic variance is decomposed into productive and nonproductive
components.

In addition to its nonsystemic variance each fund has a
\emph{systemic variance} of $f_0^2$ so that the \emph{total variance}
of the system is decomposed into its systemic, productive, and
nonproductive parts as
\begin{equation}\label{var5}
  v_\text{T} = n f_0^2 + \sum_{j=1}^nf_{1j}^2 +
  \sum_{i=2}^m\sum_{j=1}^nf_{ij}^2.
\end{equation}
Table \ref{totalvar5_2010} shows the total variance decompositions
corresponding to Tables
\ref{decomp5_2010L200} and \ref{decomp5_2010U200}.

\begin{table}[H]
  \centering
  \caption{\label{totalvar5_2010}%
    Decomposition of total variance\\
    5 large ETFs -- last 200 market-days of 2010
  }
  { \small
\begin{tabular}{|l|rr|rr|}
\cline{2-5}
  \multicolumn{1}{l|}{\rule{0mm}{4mm}}
    & \multicolumn{2}{c|}{late-heavy} & \multicolumn{2}{c|}{uniform} \\
  \multicolumn{1}{l|}{}
    & \multicolumn{2}{c|}{weights} & \multicolumn{2}{c|}{weights} \\
\hline\rule{0mm}{4mm}%
  systemic variance ($n f_0^2$) & 128 & 8.0\% & 118 & 6.7\% \\
  productive variance ($\sum x^2$) & 364 & 22.8\% & 123 & 7.0\% \\
  major nonproductive variance ($\sum y^2$) & 1004 & 62.9\% & 1408 & 79.6\% \\
  other nonproductive variance & 100 & 6.3\% & 120 & 6.8\% \\
\hline\rule{0mm}{4mm}%
  total variance & 1596 & 100.0\% & 1769 & 100.0\% \\
\hline
\end{tabular}
  }
\end{table}


\subsubsection{The \textit{XE}-plane}\label{xe-plane}

We will let
\begin{equation}\label{x-coordinate}
  x = X\mathbf{p} = F(1, :)\mathbf{p}
    = \sum_{j=1}^nf_{1j}p_j
\end{equation}
denote the productive risk coordinate of
$\mathbf{z}=\sum_{j=1}^n\mathbf{z}_jp_j$ and
\begin{equation}\label{E-coordinate}
  e = E\mathbf{p} = \sum_{j=1}^ne_jp_j
\end{equation}
be the corresponding expected return coordinate.
Figure \ref{xe-plane-1} shows the graphs of
\begin{equation}\label{ex-equation}
  e = e_0 + e_F x,
\end{equation}
corresponding to the late-heavy and uniform weight systems.
The plotted security points realize the \,$X$\, and \,$E$\, rows
of Tables \ref{decomp5_2010L200} and \ref{decomp5_2010U200}.
The grid scale is $5\times5$.

\begin{minipage}{3.0in}
  
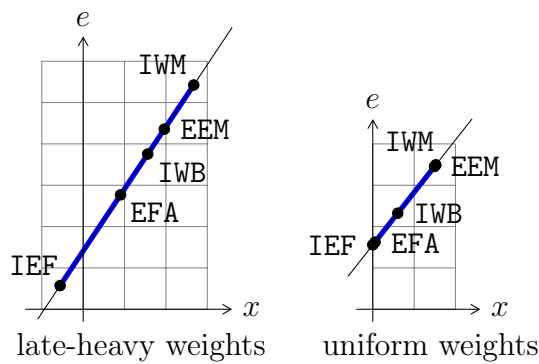
\captionof{figure}{\label{xe-plane-1}
    expected return as an\\
    affine function of productive risk\\
    last 200 market days of 2010
  }%
\begin{tikzpicture}[scale=0.11, >={angle 60}]
  \draw[help lines] (-5,0) grid[step=5] (15,30);
  \draw[line width=2pt,color=flat] (-2.79,2.86) -- (13.36,27.12);
  \draw (-5.5,-1.21) -- (18.0,34.09);
  \draw[->] (-7,0) -- (18,0) node[right] {$x$};
  \draw[->] (0, -1.5) -- (0,33) node[above] {$e$};
  \filldraw (-2.79,2.86) circle (18pt) node[left=10,above=0]{\texttt{IEF}};
  \filldraw (7.80,18.77) circle (18pt) node[right=13,below=-1]{\texttt{IWB}};
  \filldraw (13.36,27.12) circle (18pt) node[left=12,above=0]{\texttt{IWM}};
  \filldraw (4.52,13.84) circle (18pt) node[right=13,below=-1]{\texttt{EFA}};
  \filldraw (9.81,21.79) circle (18pt) node[right=2]{\texttt{EEM}};
  \draw (7.0,0) node[below=5]{late-heavy weights};
  \pgftransformxshift{35cm}
  \draw[help lines] (0,0) grid[step=5] (10,20);
  \draw[line width=2pt,color=flat] (0.00,7.79) -- (7.62,17.47);
  \draw (-3,3.97) -- (12,23.03);
  \draw[->] (-2,0) -- (13,0) node[right] {$x$};
  \draw[->] (0,-1.5) -- (0,23) node[above] {$e$};
  \filldraw (0.00,7.79) circle (18pt) node[left=2]{\texttt{IEF}};
  \filldraw (3.02,11.62) circle (18pt) node[right=2]{\texttt{IWB}};
  \filldraw (7.62,17.47) circle (18pt) node[left=10,above=1]{\texttt{IWM}};
  \filldraw (0.27,8.12) circle (18pt) node[below=0.5,right=2]{\texttt{EFA}};
  \filldraw (7.46,17.27) circle (18pt) node[right=2]{\texttt{EEM}};
  \draw (7,0) node[below=5]{uniform weights};
\end{tikzpicture}
\end{minipage}
\hfill
\begin{minipage}{3.0in}
These are 2-dimensional slices of the 6-dimensional $ZE$-spaces,
$\mathcal{L}(Z) \times \mathds{R}$, corresponding to the two weight
systems we are considering. Each periodic return vector,
$\mathbf{r}=\mathbf{z}+\mathbf{1}_Me$, corresponds to a point
$(\mathbf{z},e)$ in $ZE$-space. Figure \ref{xe-plane-1} shows the
orthogonal projection, $(x, e),~
x=\langle\mathbf{u}_1,\mathbf{z}\rangle_\omega$, of these points onto
the respective $XE$-planes.\\[-1.3ex]

In each picture the blue segment connecting the security points is the
projection of the portfolio polytope in $ZE$-space onto the $XE$-plane.
The line through this blue segment is the projection of the entire
$R$-flat.
\end{minipage}

Figure \ref{xe-plane-1} shows how the relationship between productive
risk and expected return can vary with weight system. Productive risk
accounts for 22.8\% of the total variance under the late-heavy system,
but only 7.0\% under the uniform system. This difference shows up in
the extra width of the late-heavy picture.

Looking at the vertical spread of Figure \ref{xe-plane-1} one notices
that the expected returns of the stock funds are higher and the bond
fund return lower under the late-heavy system. This is because the
negative returns of the stock funds in the first 70 days of the 200
market-day sample count more in the uniform system, and the negative
returns of the bond fund in the last 30 days count more in the
late-heavy system. (Figure \ref{growth5_2010})


\subsubsection{The \textit{eflag} flag}\label{the-eflag}

When the matrix equation $E = e_0 + e_F X$ holds exactly, the global variable
\textit{eflag} is 0 or \textbf{false}. This is typically the case when
there are many more periods than funds, as in the five-fund,
200-market-day examples just considered. Figure \ref{xe-plane-1}
displays this relationship graphically.

Rather than increase the number of funds to illustrate the
$\textit{eflag}=\textbf{true}$ condition let us decrease the number
of periods from $M=200$ market days to $M=3$ quarters. Now
$\rho=4$ (quarters per year), and we will use late-heavy weights,
$\bm{\omega}=[\omega_1,\omega_2,\omega_3]^T$, comparable to those
in Figure \ref{weights-figure}.

Here is the complete data.

\begin{table}[H]
  \centering
  \caption{\label{adjclose5_2010L3Q}%
    Quarter-ending adjusted closing prices
  }
  { \small
\begin{tabular}{|l|rrrrr|}
  \hline\rule{0mm}{4mm}
  date
  & \ticker{IEF}\hspace*{1.0ex} & \ticker{IWB}\hspace*{1.0ex}
  & \ticker{IWM}\hspace*{1.0ex} & \ticker{EFA}\hspace*{1.0ex}
  & \ticker{EEM}\hspace*{1.0ex} \\
  \hline\rule{0mm}{4mm}%
  2010-03-31 &  92.925 &  91.157 &  85.766 &  93.619 &  87.102 \\
  2010-06-30 & 100.196 &  80.605 &  77.316 &  79.141 &  77.686 \\
  2010-09-30 & 104.498 &  89.958 &  85.884 &  93.452 &  93.194 \\
  2010-12-31 & 100.000 & 100.000 & 100.000 & 100.000 & 100.000 \\
  \hline
\end{tabular}
  }
\end{table}

\vspace*{-5.0ex}
\begin{center}
\begin{minipage}{5.0in}
\begin{table}[H]
  \centering
  \caption{\label{rtn5_2010L3Q}%
    Quarterly returns and late-heavy weights
  }
    \newcolumntype{.}{D{.}{.}{3}}
  { \small
  $
\begin{array}{|c|.....|c|}
  \hline\rule{0mm}{4mm}
  \text{quarter}
  & \multicolumn{1}{r}{\ticker{IEF}\hspace*{1.0ex}} & \multicolumn{1}{r}{\ticker{IWB}\hspace*{1.0ex}}
  & \multicolumn{1}{r}{\ticker{IWM}\hspace*{1.0ex}} & \multicolumn{1}{r}{\ticker{EFA}\hspace*{1.0ex}}
  & \multicolumn{1}{r|}{\ticker{EEM}\hspace*{1.0ex}}
  & \text{weights} \\
  \hline\rule{0mm}{4mm}
  2 &  7.271 & -10.552 &  -8.450 & -14.478 &  -9.416 & 2/9 \\
  3 &  4.302 &   9.353 &   8.568 &  14.311 &  15.508 & 3/9 \\
  4 & -4.498 &  10.042 &  14.116 &   6.548 &   6.806 & 4/9 \\
  \hline
\end{array}
  $
  }
\end{table}
\end{minipage}\\[1.0ex]
{\footnotesize quarterly returns = adjusted closing price differences}
\end{center}

\vspace*{2.0ex}
Consider the return vector matrix,
$R = [\mathbf{r}_1,\ldots,\mathbf{r}_5]$, in Table \ref{rtn5_2010L3Q}.
It is easy to see that the return-flat tangent space
\begin{equation*}
  \mathcal{T}(R)=\left\{\sum_{j=1}^5\mathbf{r}_jt_j:
    \sum_{j=1}^5t_j=0\right\}
\end{equation*}
is all of $\mathds{R}^3$; in particular, the constant return vector,
$\mathbf{1}_3$, is contained in $\mathcal{T}(R)$. This is the
$\textit{eflag}=\textbf{true}$ condition implying that
equation \eqref{ex-equation} is not exact.

The \textbf{rtndecomp} output corresponding to the data of Table
\ref{rtn5_2010L3Q} is shown in Table \ref{decomp5_2010L3}. The yellow
row shows the approximate $e$-values that result from applying equation
\eqref{ex-equation} to the $x$-values of the green, productive risk row.

\begin{table}[H]
  \centering
  \caption{\label{decomp5_2010L3}%
    Decomposition of return data -- 5 large ETF universe\\
    last three quarters of 2010 -- late-heavy weights
  }
{ \small
  \newcolumntype{.}{D{.}{.}{2}}
  \newcolumntype{W}{|>{\columncolor{white}}c|}
$
\begin{array}{|c|.....||r.|}
\cline{1-6}\rule{0mm}{4mm}
  \text{fund}
  & \multicolumn{1}{r}{\parbox{7.0ex}{\hfill\ticker{IEF}\hspace*{0.2ex}}}
  & \multicolumn{1}{r}{\parbox{7.0ex}{\hfill\ticker{IWB}\hspace*{0.2ex}}}
  & \multicolumn{1}{r}{\parbox{7.0ex}{\hfill\ticker{IWM}\hspace*{0.2ex}}}
  & \multicolumn{1}{r}{\parbox{7.0ex}{\hfill\ticker{EFA}\hspace*{0.2ex}}}
  & \multicolumn{1}{r|}{\parbox{7.0ex}{\hfill\ticker{EEM}\hspace*{0.8ex}}} \\
\hline\rule[-1.5mm]{0mm}{6mm}
  E &   4.20 &  20.94 &  29.01 &  17.85 & 24.41
  & \multicolumn{2}{c|}{} \\
\cline{1-6}\rowcolor{approx}
  \text{approx} & 4.84 & 24.15 & 26.64 & 21.75 & 19.03
  & \multicolumn{2}{W}{\raisebox{1.0ex}[0pt]{$\widehat{V}_\text{T}$}} \\
\hline\rowcolor{productive}
  \multicolumn{1}{W}{\rule{0mm}{4mm}}
  & -10.13 & 12.42 & 15.34 & 9.63 & 6.45 & 626 & 42.2\% \\
\cline{2-8}\rowcolor{nonproductive}\multicolumn{1}{W}{\raisebox{1.5ex}[0pt]{$F$}}
  & -0.78 & -11.44 & -8.37 & -19.06& -17.07 & 856 & 57.8\% \\
\hhline{|=|=====#==|}\rule{0mm}{4mm}
  & \multicolumn{1}{r@{\hspace*{1.1ex}}}{103}
  & \multicolumn{1}{r@{\hspace*{1.1ex}}}{285}
  & \multicolumn{1}{r@{\hspace*{1.1ex}}}{305}
  & \multicolumn{1}{r@{\hspace*{1.1ex}}}{456}
  & \multicolumn{1}{r@{\hspace*{1.1ex}}||}{333}
  & \multicolumn{1}{r}{1483}
  & \multicolumn{1}{r|}{100\%\hspace*{-0.8ex}} \\
\cline{8-8} \raisebox{1.6ex}[0pt]{$\widehat{V}_\text{T}$}
  & 3.2\% & 14.2\% & 26.6\% & 27.5\% & 28.4\%
  & \multicolumn{1}{r|}{100\%\hspace*{-0.5ex}}\\
\cline{1-7}
  \multicolumn{7}{c}{\rule{0mm}{5mm}
    \text{~with}~f_0 = 0.005,~ e_0 = 13.51,~\text{and}~e_F = 0.856
  }
\end{array}
$
}
\end{table}

\begin{minipage}{9.3cm}
Figure \ref{xe-plane-2} shows the projection, $(x,y,e)\mapsto(x,e)$, of
the Table \ref{decomp5_2010L3} data onto the $XE$-plane. Here $y$ is the
nonproductive risk variable represented by the orange row of $F$. The
blue polygon is the image of the portfolio polyhedron.\\

The line through the blue polygon is the graph of the (approximate)
expected return function $e = e_0 + e_F x$ \eqref{ex-equation}. The
\textbf{rtndecomp} algorithm guarantees that this graph passes through
the mean XE-point, $(\bar{x},\bar{e})=(6.74,19.28)$ in this example.\\

Figure \ref{xe-plane-2} is comparable to the late-heavy side of Figure
\ref{xe-plane-1}. The $e$-values have the same order in both pictures,
but the $x$-values of the middle funds, \texttt{EFA}, \texttt{IWB}, and
\texttt{EEM}, are permuted from one picture to the other.
\end{minipage}
\hfill
\begin{minipage}{5.0cm}
\centering
\captionsetup{justification=centering}

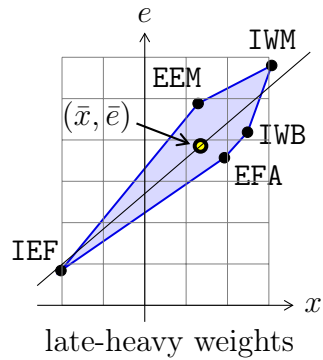
\captionof{figure}{\label{xe-plane-2}
expected return\\
as an approximate function of productive risk\\
last three quarters of 2010
}
\begin{tikzpicture}[scale=0.11, >={angle 60}]
  \fill[fill=feasible] (-10.13,4.20) -- (9.63,17.85)
  -- (12.42,20.94) -- (15.34,29.01) -- (6.45,24.41) -- cycle;
  \draw[color=boundary,line width=0.75pt] (-10.13,4.20) -- (9.63,17.85)
  -- (12.42,20.94) -- (15.34,29.01) -- (6.45,24.41) -- cycle;
  \filldraw (-10.13,4.20) circle (18pt) node[left=10,above=0]{\texttt{IEF}};
  \filldraw (12.42,20.94) circle (18pt) node[right=13,below=-7]{\texttt{IWB}};
  \filldraw (15.34,29.01) circle (18pt) node[above=2]{\texttt{IWM}};
  \filldraw (9.63,17.85) circle (18pt) node[right=13,below=-1]{\texttt{EFA}};
  \filldraw (6.45,24.41) circle (18pt) node[left=8,above=2]{\texttt{EEM}};
  \draw[help lines] (-10,0) grid[step=5] (15,30);
  \draw[->] (-13,0) -- (18,0) node[right] {$x$};
  \draw[->] (0,-2) -- (0,33) node[above] {$e$};
  \filldraw[fill=yellow,line width=1.5pt] (6.74,19.28) circle (20pt);
  \draw[->, thick] (-1,22) node[left=15,above=-8]{$(\bar{x},\bar{e})$} -- (5.5,19.5);
  \draw (-13,2.38) -- (20,30.63);
  \draw (3,0) node[below=5]{late-heavy weights};
\end{tikzpicture}
\end{minipage}


\subsubsection{The \textit{XY-} and \textit{EV-} planes}\label{xy-plane5}

Let us now return to the late-heavy weight output in Table
\ref{decomp5_2010L200}. Given a portfolio $\mathbf{p}$ let
\begin{alignat*}{4}
  x &= x_\textbf{p} &\,= X\mathbf{p} &= F(1, :)\mathbf{p},\\
  y &= y_\textbf{p} &\,= Y\mathbf{p} &= F(2, :)\mathbf{p},
\end{alignat*}
denote the productive and major nonproductive risk coordinates of
$\mathbf{p}$ (or really of $\mathbf{z}_\textbf{p} = Z\mathbf{p}$),
respectively, and let
\[
  \mathbf{y}=\mathbf{y}_\textbf{p}=F(2:m, :)\mathbf{p},\quad
  m=4,
\]
denote
the full vector of nonproductive risk corresponding to $\mathbf{p}$.
As noted in Section \ref{mv-analysis},
\begin{alignat}{4}
  \tag{\ref{e_of_x}}
  e &= e_\textbf{p} &&\,= e_0 + e_F x_\textbf{p},\\
  \tag{\ref{v_of_f}}
  v &= v_\textbf{p} &&\,=
    f_0^2 + x_\textbf{p}^2 + \|\mathbf{y}_\textbf{p}\|^2,
\end{alignat}

Figure \ref{rd_tikz5} shows the images, $(X,Y)(\bm{\Delta})$ and
$\Phi(\bm{\Delta})$, of the portfolio simplex $\bm{\Delta}$ in the $XY$- and
$EV$-planes, respectively. The $x$ and $y$ grid lines are 5 units apart.
Since the $(x,\mathbf{y})$-tuples are coordinate vectors with respect to
an orthonormal basis, the $XY$-image is the perpendicular projection of
the four-dimensional polytope, $F(\bm{\Delta})$, onto the $XY$-plane.

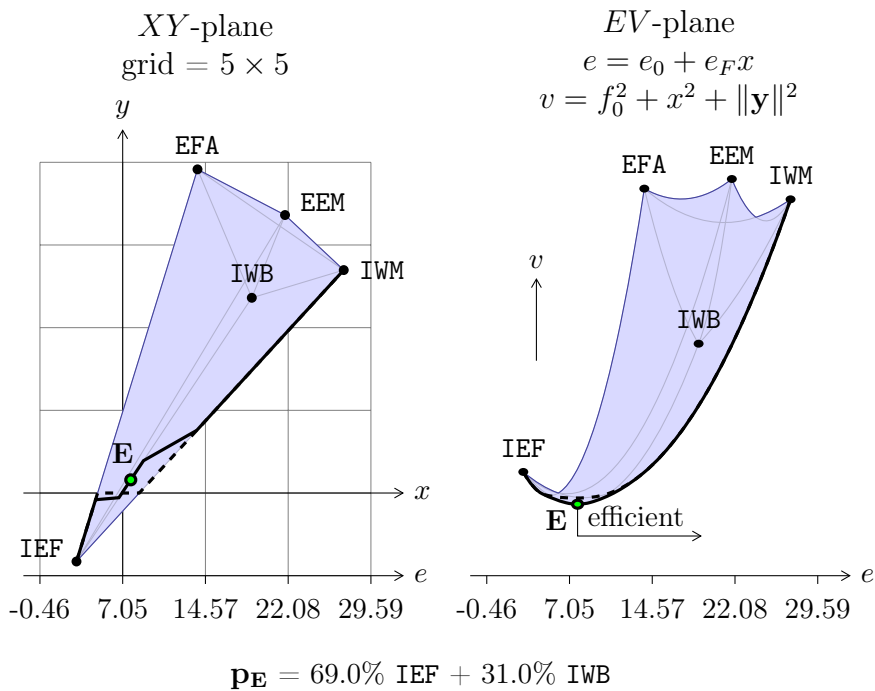
\begin{figure}[ht]
  \centering
  \caption{\label{rd_tikz5}
  Planar representations of return data -- 5 large ETFs\\ 
  last 200 market-days of 2010 -- late-heavy weights}

\begin{tikzpicture}[scale=0.22,>={angle 60}]

  \node at (-15,-2) {\parbox{5cm}{\centering%
    $XY$-plane\\ 
    grid = $5\times5$\\
    ~
  }};
  \node at (13,-2) {\parbox{5cm}{\centering%
    $EV$-plane\\
    $e = e_0 + e_F x$\\
    $v = f_0^2 + x^2 + \|\mathbf{y}\|^2$ 
  }};
\small
  \node at (-2,-39) {{\normalsize$\mathbf{p}_\textbf{E}$}
    = 69.0\% \ticker{IEF} + 31.0\% \ticker{IWB}};
  \draw[->] (5,-20) -- (5,-15) node[above] {$v$};

\begin{scope}[xshift=-20cm,yshift=-28cm]
  \draw[help lines] (-5,-5) grid[step=5] (15,20);
  \draw[->] (-6,0) -- (17,0) node[right] {$x$};
  \draw[->] (-6,-5) -- (17,-5) node[right] {$e$};
  \draw[->] (0,-5) -- (0,22) node[above] {$y$};
  \draw (-5,-5) -- (-5,-5.5) node[below=2] {-0.46};
  \draw (0,-5) -- (0,-5.5) node[below=2] { 7.05};
  \draw (5,-5) -- (5,-5.5) node[below=2] {14.57};
  \draw (10,-5) -- (10,-5.5) node[below=2] {22.08};
  \draw (15,-5) -- (15,-5.5) node[below=2] {29.59};
  \fill[obtainable1]
  (-2.7922,-4.1412) --  
  (4.5205,19.5736) --  
  (9.8099,16.8115) --  
  (13.3578,13.4828) --  
  cycle;
  \draw[connection1] (-2.7922,-4.1412) -- (7.8018,11.8061);  
  \draw[connection1] (-2.7922,-4.1412) -- (13.3578,13.4828);  
  \draw[connection1] (-2.7922,-4.1412) -- (4.5205,19.5736);  
  \draw[connection1] (-2.7922,-4.1412) -- (9.8099,16.8115);  
  \draw[connection1] (7.8018,11.8061) -- (13.3578,13.4828);  
  \draw[connection1] (7.8018,11.8061) -- (4.5205,19.5736);  
  \draw[connection1] (7.8018,11.8061) -- (9.8099,16.8115);  
  \draw[connection1] (13.3578,13.4828) -- (4.5205,19.5736);  
  \draw[connection1] (13.3578,13.4828) -- (9.8099,16.8115);  
  \draw[connection1] (4.5205,19.5736) -- (9.8099,16.8115);  
  \draw[boundary1] (1.0026,0)
  -- (-2.7922,-4.1412)
  -- (4.5205,19.5736)
  -- (9.8099,16.8115)
  -- (13.3578,13.4828);
  \draw[very thick,dashed] (-2.7922,-4.1412)
  -- (-1.5153,0.0000)
  -- (1.0026,-0.0000)
  -- (13.3578,13.4828);
  \draw[very thick] (-2.7922,-4.1412)
  -- (-1.6403,-0.4054)
  -- (-0.2300,-0.2842)
  -- (1.2579,1.9555)
  -- (4.4744,3.7886)
  -- (13.3578,13.4828);
  \fill (-2.7922,-4.1412) circle (0.30) node[above=5,left=1] {\ticker{IEF}};  
  \fill (7.8018,11.8061) circle (0.30) node[above=2] {\ticker{IWB}};  
  \fill (13.3578,13.4828) circle (0.30) node[right=2] {\ticker{IWM}};  
  \fill (4.5205,19.5736) circle (0.30) node[above=2] {\ticker{EFA}};  
  \fill (9.8099,16.8115) circle (0.30) node[above=5,right=2] {\ticker{EEM}};  
  \filldraw[fill=green,very thick] (0.4882,0.7969) circle (0.30) node[left=3,above=2] {\textbf{E}};
\end{scope}  

\begin{scope}[xshift=7cm,yshift=-28cm]
  \draw[->] (-6,-5) -- (17,-5) node[right] {$e$};
  \draw (-5,-5) -- (-5,-5.5) node[below=2] {-0.46};
  \draw (0,-5) -- (0,-5.5) node[below=2] { 7.05};
  \draw (5,-5) -- (5,-5.5) node[below=2] {14.57};
  \draw (10,-5) -- (10,-5.5) node[below=2] {22.08};
  \draw (15,-5) -- (15,-5.5) node[below=2] {29.59};
\end{scope}

\begin{scope}[xshift=7cm,yshift=-29cm,yscale=0.8]
  \fill[obtainable1]
  (-2.7922,2.8443) .. controls (-2.4082,1.9786) and (-2.0242,1.4264) ..
  (-1.6403,1.1878) .. controls (-1.1702,0.8957) and (-0.7001,0.6730) ..
  (-0.2300,0.5198) .. controls (0.2660,0.3582) and (0.7619,0.3640) ..
  (1.2579,0.5372) .. controls (2.3301,0.9116) and (3.4022,1.6575) ..
  (4.4744,2.7748) .. controls (7.4356,5.8605) and (10.3967,12.7531) ..
  (13.3578,23.4526) .. controls (12.6611,22.8410) and (11.9643,22.3892) ..
  (11.2675,22.0972) .. controls (10.7816,22.5356) and (10.2958,23.4995) ..
  (9.8099,24.9888) .. controls (8.0468,23.2542) and (6.2836,23.0116) ..
  (4.5205,24.2610) .. controls (2.7881,10.2079) and (1.0557,2.5366) ..
  (-0.6768,1.2470) .. controls (-1.3819,1.5635) and (-2.0871,2.0960) ..
  (-2.7922,2.8443);
  \draw[connection1] (-2.7922,2.8443) .. controls (0.7391,-2.4122) and (4.2704,0.8191) .. (7.8018,12.5383);  
  \draw[connection1] (-2.7922,2.8443) .. controls (2.5911,-2.8685) and (7.9745,4.0009) .. (13.3578,23.4526);  
  \draw[connection1] (-2.7922,2.8443) .. controls (-0.3546,-2.6511) and (2.0829,4.4878) .. (4.5205,24.2610);  
  \draw[connection1] (-2.7922,2.8443) .. controls (1.4085,-2.4356) and (5.6092,4.9459) .. (9.8099,24.9888);  
  \draw[connection1] (7.8018,12.5383) .. controls (9.6538,14.9674) and (11.5058,18.6055) .. (13.3578,23.4526);  
  \draw[connection1] (7.8018,12.5383) .. controls (6.7080,14.7610) and (5.6143,18.6685) .. (4.5205,24.2610);  
  \draw[connection1] (7.8018,12.5383) .. controls (8.4712,14.6859) and (9.1405,18.8361) .. (9.8099,24.9888);  
  \draw[connection1] (13.3578,23.4526) .. controls (10.4121,20.8672) and (7.4663,21.1367) .. (4.5205,24.2610);  
  \draw[connection1] (13.3578,23.4526) .. controls (12.1752,20.8516) and (10.9925,21.3637) .. (9.8099,24.9888);  
  \draw[connection1] (4.5205,24.2610) .. controls (6.2836,23.0116) and (8.0468,23.2542) .. (9.8099,24.9888);  
  \draw[boundary1]
  (13.3578,23.4526) .. controls (12.6611,22.8410) and (11.9643,22.3892) ..
  (11.2675,22.0972) .. controls (10.7816,22.5356) and (10.2958,23.4995) ..
  (9.8099,24.9888) .. controls (8.0468,23.2542) and (6.2836,23.0116) ..
  (4.5205,24.2610) .. controls (2.7881,10.2079) and (1.0557,2.5366) ..
  (-0.6768,1.2470) .. controls (-1.3819,1.5635) and (-2.0871,2.0960) ..
  (-2.7922,2.8443);
  \draw[very thick,dashed]
  (-2.7922,2.8443) .. controls (-2.3666,1.8846) and (-1.9409,1.3103) ..
  (-1.5153,1.1212) .. controls (-0.6760,0.8855) and (0.1633,0.8122) ..
  (1.0026,0.9013) .. controls (5.1210,1.0545) and (9.2394,8.5716) ..
  (13.3578,23.4526);
  \draw[very thick]
  (-2.7922,2.8443) .. controls (-2.4082,1.9786) and (-2.0242,1.4264) ..
  (-1.6403,1.1878) .. controls (-1.1702,0.8957) and (-0.7001,0.6730) ..
  (-0.2300,0.5198) .. controls (0.2660,0.3582) and (0.7619,0.3640) ..
  (1.2579,0.5372) .. controls (2.3301,0.9116) and (3.4022,1.6575) ..
  (4.4744,2.7748) .. controls (7.4356,5.8605) and (10.3967,12.7531) ..
  (13.3578,23.4526);
  \fill (-2.7922,2.8443) circle (0.30) node[above=2] {\ticker{IEF}};  
  \fill (7.8018,12.5383) circle (0.30) node[above=2] {\ticker{IWB}};  
  \fill (13.3578,23.4526) circle (0.30) node[above=2] {\ticker{IWM}};  
  \fill (4.5205,24.2610) circle (0.30) node[above=2] {\ticker{EFA}};  
  \fill (9.8099,24.9888) circle (0.30) node[above=2] {\ticker{EEM}};  
  \draw[->] (0.4882,0.4028) -- (0.4882,-2) -- (8,-2);
  \node[right] at (0.4882,-0.4831) {efficient};
  \filldraw[fill=green,very thick] (0.4882,0.4028) circle (0.30) node[left=8,below=-2] {\textbf{E}};
\end{scope} 

\end{tikzpicture}

\end{figure}

In Figure \ref{rd_tikz5} the $y$-coordinates of the securities are
actually the negatives of the $Y$ coefficients in Table
\ref{decomp5_2010L200}. This sign change makes the comparison of the
$XY$ and $EV$ images more natural, but it has no effect on our
analysis---a principal direction of nonproductive risk is, at most,
determined up to a reflection through the origin. On the other hand, it
is important to note that the stock funds are all in the first quadrant
and the bond fund, \ticker{IEF}, is in the third quadrant of $XY$-side
of Figure \ref{rd_tikz5}. This corresponds to the fact that the stock
funds are positively correlated with each other and negatively
correlated with the bond fund. This is also why the bond fund,
\ticker{IEF}, is a component of every minimum-variance portfolio other
than single security portfolio of maximum expected return, \ticker{IWM}.

The solid black path through either image corresponds to the set of
minimum-variance portfolios. 
As noted at the end of Section \ref{mv-analysis}, a minimum-variance
portfolio at a particular x = x* must minimize the value
$\|\mathbf{y}\|$ on the polytope $F(\bm{\Delta})\cap\{x = x*\}$. Points
on the dotted path in either image approximate this criterion. They
correspond to portfolios that minimize $|y|$ rather that
$\|\mathbf{y}\|$.

The point \textbf{E} in either figure is the image of the portfolio
$\mathbf{p}_\textbf{E}$ of absolute minimum variance.
$\mathbf{p}_\textbf{E}$ is an efficient portfolio. The solid black path
to the right of \textbf{E} is the image of the other efficient
portfolios. These efficient portfolios, in total, make up the piecewise
linear path in $\bm{\Delta}$ that goes from $\mathbf{p}_\textbf{E}$
through the ``corner portfolios'' of Table \ref{efficient5_2010L200} to
\ticker{IWM}. The corner portfolios show up as the corners above the
$x$-axis in the $XY$-image of the minimum-variance path.

\begin{table}[H]
  \centering
  \caption{\label{efficient5_2010L200}%
    Optimal portfolio paths -- 5 large ETF universe \\
    last 200 days of 2010 -- late-heavy weights
  }
{ \small
  \begin{tabular}{|c|rrrr|rrr|}
  \cline{2-8}
  \multicolumn{1}{l|}{\rule{0mm}{4mm}}
    & \multicolumn{4}{c|}{minimum-$\|\mathbf{y}\|$ (efficient) path}
    & \multicolumn{3}{c|}{minimum-$|y|$ path} \\
  \cline{2-8}
  \multicolumn{1}{l|}{\rule{0mm}{2mm}}
    & & \multicolumn{2}{c}{corner} & & \multicolumn{1}{c}{$x_\textbf{E}$} & \multicolumn{1}{c}{corner} & \\
  \multicolumn{1}{l|}{}
    & \multicolumn{1}{c}{ \raisebox{1.2ex}[0pt]{\normalsize$\mathbf{p}_\textbf{E}$}}
    & \multicolumn{2}{c}{portfolios}
    & \multicolumn{1}{c|}{ \raisebox{1.2ex}[0pt]{\ticker{IWM}}}
    & \multicolumn{1}{c}{portf}
    & \multicolumn{1}{c}{portf}
    & \multicolumn{1}{c|}{ \raisebox{1.2ex}[0pt]{\ticker{IWM}}} \\
  \hline\rule[-1.5mm]{0mm}{6mm}%
    \ticker{IEF} & 0.690 & 0.618 &  0.550 & \multicolumn{1}{c|}{0}
      & 0.777 & 0.765 & \multicolumn{1}{c|}{0} \\
    \ticker{IWB} & 0.310 & 0.382 & \multicolumn{1}{c}{0} & \multicolumn{1}{c|}{0}
      & \multicolumn{1}{c}{0} & \multicolumn{1}{c}{0} & \multicolumn{1}{c|}{0} \\
    \ticker{IWM} & \multicolumn{1}{c}{0} & \multicolumn{1}{c}{0} & 0.450 & 1.000
      & 0.187 & 0.235 & 1.000 \\
    \ticker{EFA} & \multicolumn{1}{c}{0} & \multicolumn{1}{c}{0} & \multicolumn{1}{c}{0} & \multicolumn{1}{c|}{0}
      & 0.036 & \multicolumn{1}{c}{0} & \multicolumn{1}{c|}{0} \\
    \ticker{EEM} & \multicolumn{1}{c}{0} & \multicolumn{1}{c}{0} & \multicolumn{1}{c}{0} & \multicolumn{1}{c|}{0}
      & \multicolumn{1}{c}{0} & \multicolumn{1}{c}{0} & \multicolumn{1}{c|}{0} \\
  \hline\rule[-1.5mm]{0mm}{6mm}%
    $x$ & 0.49 & 1.26 & 4.47 & 13.36 & 0.49 & 1.00 & 13.36 \\
    $e$ & 7.79 & 8.94 & 13.78 & 27.12 & 7.79 & 8.56 & 27.12 \\
    $\sigma$ & 5.69 & 5.88 & 8.48 & 20.41 & 6.33 & 6.38 & 20.41 \\
  \hline\rule[-1.5mm]{0mm}{6mm}%
    avg $e$ & \multicolumn{4}{c|}{17.45} & \multicolumn{3}{c|}{17.45} \\
    rms $\sigma$ & \multicolumn{4}{c|}{12.70} & \multicolumn{3}{c|}{12.74} \\
  \hline
  \end{tabular}
}
\end{table}

Table \ref{efficient5_2010L200} also shows the $x_\textbf{E}$- and
corner porfolios of the minimum-$|y|$ path over the efficient $x$-range
from $x_\textbf{E} = 0.49$ to $x_\textrm{max} = 13.36$. The
minimum-$|y|$ path and the efficient path are exactly the same from
$x=4.47$ to $x_\textrm{max}$, but the paths differ between
$x_\textbf{E}$ and $x=4.47$, the most substantial
$\sigma$-differences occurring near $x_\textbf{E}$.

The average value of $e$ over the two portfolio paths in Table
\ref{efficient5_2010L200} is just the average of the end values, 7.79 and
27.12. On the other hand, the average variance,
\[
  \avg v = \frac{1}{x_\textrm{max} - x_\textbf{E}}
  \int_{x_\textbf{E}}^{x_\textrm{max}} v(x)\, dx,
\]
and the root-mean-square risk,
$\rms \sigma = \sqrt{\avg v}$, depend on
the whole path. 

\vspace*{2.0ex}
\begin{rmk*}
Throughout this paper we use Markowitz's Critical Line Algorithm
as described in \cite{Niedermayer:2006uq} to compute minimum-variance
paths through portfolio simplices.
\end{rmk*}


\subsection{Eighteen emerging markets ETFs -- 2010}
\label{18emfunds}

Now let us consider a larger universe of securities---the 18 iShares
emerging markets ETFs that existed throughout 2010
\begin{compactenum}\label{emfunds}
  \item \ticker{BKF } -- iShares MSCI BRIC Index Fund
  \item \ticker{ECH } -- iShares MSCI Chile Investable Market Index Fund
  \item \ticker{EEM } -- iShares MSCI Emerging Markets Index Fund
  \item \ticker{EMIF} -- iShares S\&P Emerging Markets Infrastructure Index Fund
  \item \ticker{EPU } -- iShares MSCI All Peru Capped Index Fund
  \item \ticker{ESR } -- iShares MSCI Emerging Markets Eastern Europe Index Fund
  \item \ticker{EWM } -- iShares MSCI Malaysia Index Fund
  \item \ticker{EWT } -- iShares MSCI Taiwan Index Fund
  \item \ticker{EWW } -- iShares MSCI Mexico Investable Market Index Fund
  \item \ticker{EWY } -- iShares MSCI South Korea Index Fund
  \item \ticker{EWZ } -- iShares MSCI Brazil Index Fund
  \item \ticker{EZA } -- iShares MSCI South Africa Index Fund
  \item \ticker{FCHI} -- iShares FTSE China (HK Listed) Index Fund
  \item \ticker{FXI } -- iShares FTSE China 25 Index Fund
  \item \ticker{ILF } -- iShares S\&P Latin America 40 Index Fund
  \item \ticker{INDY} -- iShares S\&P India Nifty 50 Index Fund
  \item \ticker{THD } -- iShares MSCI Thailand Investable Market Index Fund
  \item \ticker{TUR } -- iShares MSCI Turkey Investable Market Index Fund
\end{compactenum}
Our input to \textbf{rtndecomp} will be the $200\times18$ matrix
$R=[r_{ij}]$ of normalized, linear, daily returns for the 18 emerging
markets funds listed above, over the last 200 market days of 2010.
The returns are normalized on 2010-12-31---they are daily
adjusted-closing-price differences divided by 2010-12-31 adjusted
closing prices. We will stick to the late-heavy weights $\bm{\omega}$ of
Figure \ref{weights-figure} and use $\rho=252$ market-days per year.

The $17\times18$ risk matrix $F$ corresponding to this example is not
displayed, but Table \ref{totalvar18_2010} summarizes how the total
variance of return is decomposed by $f_0$ and $F$, and Figure
\ref{rd_tikz18} shows the $XY$ and $EV$ planar representations of the
\textbf{rtndecomp} output.

\begin{table}[H]
  \centering
  \caption{\label{totalvar18_2010}%
    Decomposition of total variance -- 18 emerging market ETFs\\
    last 200 market-days of 2010 -- late-heavy weights
  }
  { \small
\begin{tabular}{|l|rr|}
\hline\rule{0mm}{4mm}%
  systemic variance ($n f_0^2$) & 1131 & 13.9\% \\
  productive variance ($\sum x^2$) & 1090 & 13.4\% \\
  major nonproductive variance ($\sum y^2$) & 4427 & 54.3\% \\
  other nonproductive variance & 1512 & 18.5\% \\
\hline\rule{0mm}{4mm}%
  total variance & 8160 & 100.0\% \\
\hline
\end{tabular}
  }
\end{table}

\begin{figure}[ht]
  \centering
  \caption{\label{rd_tikz18}
  Planar representations of return data -- 18 emerging markets ETFs\\ 
  last 200 market-days of 2010 -- late-heavy weights}

\begin{tikzpicture}[scale=0.34,>={angle 60}]

  \node at (-14.5,-2) {\parbox{5cm}{\centering%
    $XY$-plane\\ 
    grid = $5\times5$\\
    ~
  }};
  \node at (6,-2) {\parbox{5cm}{\centering%
    $EV$-plane\\
    $e = e_0 + e_F x$\\
    $v = f_0^2 + x^2 + \|\mathbf{y}\|^2$ 
  }};
\small
  \node at (-5,-30.5) {{\normalsize$\mathbf{p}_\textbf{E}$}
  = 39.5\% \ticker{ECH} + 13.0\% \ticker{EPU} + 41.6\% \ticker{EWM} + 5.9\% \ticker{EWT}};
  \draw[->] (10,-13) -- (10,-8) node[above] {$v$} ;

\begin{scope}[xshift=-7cm,yshift=-31cm]
  \draw[help lines] (-15,5) grid[step=5] (0,25);
  \draw[->] (0,5) -- (0,26.4) node[above] {$y$};
  \draw[->] (-16,5) -- (1.4,5) node[right] {$e$};
  \draw (-15,5) -- (-15,4.6) node[below=2] {-2.71};
  \draw (-10,5) -- (-10,4.6) node[below=2] {15.57};
  \draw (-5,5) -- (-5,4.6) node[below=2] {33.86};
  \draw (0,5) -- (0,4.6) node[below=2] {52.14};
  \node[above, rotate = 90] at (-15,10) {$x=-15$};
  \node[above] at (-12.5,5) {$y=5$};
  \node[above] at (-12.5,25) {$y=25$};
  \fill[obtainable1]
  (-12.4886,15.1321) --  
  (-5.8599,7.7519) --  
  (-3.3288,6.0908) --  
  (-0.0187,12.2987) --  
  (-8.9946,23.8857) --  
  cycle;
  \draw[connection1] (-8.2990,17.3696) -- (-10.5438,17.1920);  
  \draw[connection1] (-8.2990,17.3696) -- (-3.3288,6.0908);  
  \draw[connection1] (-8.2990,17.3696) -- (-9.2684,12.4610);  
  \draw[connection1] (-8.2990,17.3696) -- (-0.0187,12.2987);  
  \draw[connection1] (-8.2990,17.3696) -- (-7.9379,19.7779);  
  \draw[connection1] (-8.2990,17.3696) -- (-5.8599,7.7519);  
  \draw[connection1] (-8.2990,17.3696) -- (-4.3589,10.1846);  
  \draw[connection1] (-8.2990,17.3696) -- (-6.4031,13.5327);  
  \draw[connection1] (-8.2990,17.3696) -- (-5.1895,17.7841);  
  \draw[connection1] (-8.2990,17.3696) -- (-9.4531,18.6390);  
  \draw[connection1] (-8.2990,17.3696) -- (-4.8883,17.1922);  
  \draw[connection1] (-8.2990,17.3696) -- (-11.5798,14.7002);  
  \draw[connection1] (-8.2990,17.3696) -- (-12.4886,15.1321);  
  \draw[connection1] (-8.2990,17.3696) -- (-8.0982,16.5838);  
  \draw[connection1] (-8.2990,17.3696) -- (-8.2245,17.7093);  
  \draw[connection1] (-8.2990,17.3696) -- (-3.4212,13.5828);  
  \draw[connection1] (-8.2990,17.3696) -- (-8.9946,23.8857);  
  \draw[connection1] (-10.5438,17.1920) -- (-3.3288,6.0908);  
  \draw[connection1] (-10.5438,17.1920) -- (-9.2684,12.4610);  
  \draw[connection1] (-10.5438,17.1920) -- (-0.0187,12.2987);  
  \draw[connection1] (-10.5438,17.1920) -- (-7.9379,19.7779);  
  \draw[connection1] (-10.5438,17.1920) -- (-5.8599,7.7519);  
  \draw[connection1] (-10.5438,17.1920) -- (-4.3589,10.1846);  
  \draw[connection1] (-10.5438,17.1920) -- (-6.4031,13.5327);  
  \draw[connection1] (-10.5438,17.1920) -- (-5.1895,17.7841);  
  \draw[connection1] (-10.5438,17.1920) -- (-9.4531,18.6390);  
  \draw[connection1] (-10.5438,17.1920) -- (-4.8883,17.1922);  
  \draw[connection1] (-10.5438,17.1920) -- (-11.5798,14.7002);  
  \draw[connection1] (-10.5438,17.1920) -- (-12.4886,15.1321);  
  \draw[connection1] (-10.5438,17.1920) -- (-8.0982,16.5838);  
  \draw[connection1] (-10.5438,17.1920) -- (-8.2245,17.7093);  
  \draw[connection1] (-10.5438,17.1920) -- (-3.4212,13.5828);  
  \draw[connection1] (-10.5438,17.1920) -- (-8.9946,23.8857);  
  \draw[connection1] (-3.3288,6.0908) -- (-9.2684,12.4610);  
  \draw[connection1] (-3.3288,6.0908) -- (-0.0187,12.2987);  
  \draw[connection1] (-3.3288,6.0908) -- (-7.9379,19.7779);  
  \draw[connection1] (-3.3288,6.0908) -- (-5.8599,7.7519);  
  \draw[connection1] (-3.3288,6.0908) -- (-4.3589,10.1846);  
  \draw[connection1] (-3.3288,6.0908) -- (-6.4031,13.5327);  
  \draw[connection1] (-3.3288,6.0908) -- (-5.1895,17.7841);  
  \draw[connection1] (-3.3288,6.0908) -- (-9.4531,18.6390);  
  \draw[connection1] (-3.3288,6.0908) -- (-4.8883,17.1922);  
  \draw[connection1] (-3.3288,6.0908) -- (-11.5798,14.7002);  
  \draw[connection1] (-3.3288,6.0908) -- (-12.4886,15.1321);  
  \draw[connection1] (-3.3288,6.0908) -- (-8.0982,16.5838);  
  \draw[connection1] (-3.3288,6.0908) -- (-8.2245,17.7093);  
  \draw[connection1] (-3.3288,6.0908) -- (-3.4212,13.5828);  
  \draw[connection1] (-3.3288,6.0908) -- (-8.9946,23.8857);  
  \draw[connection1] (-9.2684,12.4610) -- (-0.0187,12.2987);  
  \draw[connection1] (-9.2684,12.4610) -- (-7.9379,19.7779);  
  \draw[connection1] (-9.2684,12.4610) -- (-5.8599,7.7519);  
  \draw[connection1] (-9.2684,12.4610) -- (-4.3589,10.1846);  
  \draw[connection1] (-9.2684,12.4610) -- (-6.4031,13.5327);  
  \draw[connection1] (-9.2684,12.4610) -- (-5.1895,17.7841);  
  \draw[connection1] (-9.2684,12.4610) -- (-9.4531,18.6390);  
  \draw[connection1] (-9.2684,12.4610) -- (-4.8883,17.1922);  
  \draw[connection1] (-9.2684,12.4610) -- (-11.5798,14.7002);  
  \draw[connection1] (-9.2684,12.4610) -- (-12.4886,15.1321);  
  \draw[connection1] (-9.2684,12.4610) -- (-8.0982,16.5838);  
  \draw[connection1] (-9.2684,12.4610) -- (-8.2245,17.7093);  
  \draw[connection1] (-9.2684,12.4610) -- (-3.4212,13.5828);  
  \draw[connection1] (-9.2684,12.4610) -- (-8.9946,23.8857);  
  \draw[connection1] (-0.0187,12.2987) -- (-7.9379,19.7779);  
  \draw[connection1] (-0.0187,12.2987) -- (-5.8599,7.7519);  
  \draw[connection1] (-0.0187,12.2987) -- (-4.3589,10.1846);  
  \draw[connection1] (-0.0187,12.2987) -- (-6.4031,13.5327);  
  \draw[connection1] (-0.0187,12.2987) -- (-5.1895,17.7841);  
  \draw[connection1] (-0.0187,12.2987) -- (-9.4531,18.6390);  
  \draw[connection1] (-0.0187,12.2987) -- (-4.8883,17.1922);  
  \draw[connection1] (-0.0187,12.2987) -- (-11.5798,14.7002);  
  \draw[connection1] (-0.0187,12.2987) -- (-12.4886,15.1321);  
  \draw[connection1] (-0.0187,12.2987) -- (-8.0982,16.5838);  
  \draw[connection1] (-0.0187,12.2987) -- (-8.2245,17.7093);  
  \draw[connection1] (-0.0187,12.2987) -- (-3.4212,13.5828);  
  \draw[connection1] (-0.0187,12.2987) -- (-8.9946,23.8857);  
  \draw[connection1] (-7.9379,19.7779) -- (-5.8599,7.7519);  
  \draw[connection1] (-7.9379,19.7779) -- (-4.3589,10.1846);  
  \draw[connection1] (-7.9379,19.7779) -- (-6.4031,13.5327);  
  \draw[connection1] (-7.9379,19.7779) -- (-5.1895,17.7841);  
  \draw[connection1] (-7.9379,19.7779) -- (-9.4531,18.6390);  
  \draw[connection1] (-7.9379,19.7779) -- (-4.8883,17.1922);  
  \draw[connection1] (-7.9379,19.7779) -- (-11.5798,14.7002);  
  \draw[connection1] (-7.9379,19.7779) -- (-12.4886,15.1321);  
  \draw[connection1] (-7.9379,19.7779) -- (-8.0982,16.5838);  
  \draw[connection1] (-7.9379,19.7779) -- (-8.2245,17.7093);  
  \draw[connection1] (-7.9379,19.7779) -- (-3.4212,13.5828);  
  \draw[connection1] (-7.9379,19.7779) -- (-8.9946,23.8857);  
  \draw[connection1] (-5.8599,7.7519) -- (-4.3589,10.1846);  
  \draw[connection1] (-5.8599,7.7519) -- (-6.4031,13.5327);  
  \draw[connection1] (-5.8599,7.7519) -- (-5.1895,17.7841);  
  \draw[connection1] (-5.8599,7.7519) -- (-9.4531,18.6390);  
  \draw[connection1] (-5.8599,7.7519) -- (-4.8883,17.1922);  
  \draw[connection1] (-5.8599,7.7519) -- (-11.5798,14.7002);  
  \draw[connection1] (-5.8599,7.7519) -- (-12.4886,15.1321);  
  \draw[connection1] (-5.8599,7.7519) -- (-8.0982,16.5838);  
  \draw[connection1] (-5.8599,7.7519) -- (-8.2245,17.7093);  
  \draw[connection1] (-5.8599,7.7519) -- (-3.4212,13.5828);  
  \draw[connection1] (-5.8599,7.7519) -- (-8.9946,23.8857);  
  \draw[connection1] (-4.3589,10.1846) -- (-6.4031,13.5327);  
  \draw[connection1] (-4.3589,10.1846) -- (-5.1895,17.7841);  
  \draw[connection1] (-4.3589,10.1846) -- (-9.4531,18.6390);  
  \draw[connection1] (-4.3589,10.1846) -- (-4.8883,17.1922);  
  \draw[connection1] (-4.3589,10.1846) -- (-11.5798,14.7002);  
  \draw[connection1] (-4.3589,10.1846) -- (-12.4886,15.1321);  
  \draw[connection1] (-4.3589,10.1846) -- (-8.0982,16.5838);  
  \draw[connection1] (-4.3589,10.1846) -- (-8.2245,17.7093);  
  \draw[connection1] (-4.3589,10.1846) -- (-3.4212,13.5828);  
  \draw[connection1] (-4.3589,10.1846) -- (-8.9946,23.8857);  
  \draw[connection1] (-6.4031,13.5327) -- (-5.1895,17.7841);  
  \draw[connection1] (-6.4031,13.5327) -- (-9.4531,18.6390);  
  \draw[connection1] (-6.4031,13.5327) -- (-4.8883,17.1922);  
  \draw[connection1] (-6.4031,13.5327) -- (-11.5798,14.7002);  
  \draw[connection1] (-6.4031,13.5327) -- (-12.4886,15.1321);  
  \draw[connection1] (-6.4031,13.5327) -- (-8.0982,16.5838);  
  \draw[connection1] (-6.4031,13.5327) -- (-8.2245,17.7093);  
  \draw[connection1] (-6.4031,13.5327) -- (-3.4212,13.5828);  
  \draw[connection1] (-6.4031,13.5327) -- (-8.9946,23.8857);  
  \draw[connection1] (-5.1895,17.7841) -- (-9.4531,18.6390);  
  \draw[connection1] (-5.1895,17.7841) -- (-4.8883,17.1922);  
  \draw[connection1] (-5.1895,17.7841) -- (-11.5798,14.7002);  
  \draw[connection1] (-5.1895,17.7841) -- (-12.4886,15.1321);  
  \draw[connection1] (-5.1895,17.7841) -- (-8.0982,16.5838);  
  \draw[connection1] (-5.1895,17.7841) -- (-8.2245,17.7093);  
  \draw[connection1] (-5.1895,17.7841) -- (-3.4212,13.5828);  
  \draw[connection1] (-5.1895,17.7841) -- (-8.9946,23.8857);  
  \draw[connection1] (-9.4531,18.6390) -- (-4.8883,17.1922);  
  \draw[connection1] (-9.4531,18.6390) -- (-11.5798,14.7002);  
  \draw[connection1] (-9.4531,18.6390) -- (-12.4886,15.1321);  
  \draw[connection1] (-9.4531,18.6390) -- (-8.0982,16.5838);  
  \draw[connection1] (-9.4531,18.6390) -- (-8.2245,17.7093);  
  \draw[connection1] (-9.4531,18.6390) -- (-3.4212,13.5828);  
  \draw[connection1] (-9.4531,18.6390) -- (-8.9946,23.8857);  
  \draw[connection1] (-4.8883,17.1922) -- (-11.5798,14.7002);  
  \draw[connection1] (-4.8883,17.1922) -- (-12.4886,15.1321);  
  \draw[connection1] (-4.8883,17.1922) -- (-8.0982,16.5838);  
  \draw[connection1] (-4.8883,17.1922) -- (-8.2245,17.7093);  
  \draw[connection1] (-4.8883,17.1922) -- (-3.4212,13.5828);  
  \draw[connection1] (-4.8883,17.1922) -- (-8.9946,23.8857);  
  \draw[connection1] (-11.5798,14.7002) -- (-12.4886,15.1321);  
  \draw[connection1] (-11.5798,14.7002) -- (-8.0982,16.5838);  
  \draw[connection1] (-11.5798,14.7002) -- (-8.2245,17.7093);  
  \draw[connection1] (-11.5798,14.7002) -- (-3.4212,13.5828);  
  \draw[connection1] (-11.5798,14.7002) -- (-8.9946,23.8857);  
  \draw[connection1] (-12.4886,15.1321) -- (-8.0982,16.5838);  
  \draw[connection1] (-12.4886,15.1321) -- (-8.2245,17.7093);  
  \draw[connection1] (-12.4886,15.1321) -- (-3.4212,13.5828);  
  \draw[connection1] (-12.4886,15.1321) -- (-8.9946,23.8857);  
  \draw[connection1] (-8.0982,16.5838) -- (-8.2245,17.7093);  
  \draw[connection1] (-8.0982,16.5838) -- (-3.4212,13.5828);  
  \draw[connection1] (-8.0982,16.5838) -- (-8.9946,23.8857);  
  \draw[connection1] (-8.2245,17.7093) -- (-3.4212,13.5828);  
  \draw[connection1] (-8.2245,17.7093) -- (-8.9946,23.8857);  
  \draw[connection1] (-3.4212,13.5828) -- (-8.9946,23.8857);  
  \draw[boundary1] (-0.0187,12.2987) -- (-8.9946,23.8857) -- (-12.4886,15.1321);
  \draw[very thick,dashed]
  (-12.4886,15.1321) -- (-5.8599,7.7519) -- (-3.3288,6.0908) -- (-0.0187,12.2987);
  \draw[very thick] (-12.4886,15.1321)
  -- (-11.2053,14.0676)
  -- (-9.1084,11.6453)
  -- (-8.5418,11.0351)
  -- (-6.2576,8.6110)
  -- (-5.3918,7.6324)
  -- (-5.0545,7.3707)
  -- (-4.9849,7.3172)
  -- (-4.8192,7.3014)
  -- (-1.6963,9.3508)
  -- (-1.5235,9.4765)
  -- (-0.0187,12.2987);
  \fill (-8.2990,17.3696) circle (0.22) node[right=2] {\ticker{EEM}};  
  \fill (-10.5438,17.1920) circle (0.22) node[left=3,above=1.5] {\ticker{BKF}};  
  \fill (-3.3288,6.0908) circle (0.22) node[below=1.5] {\ticker{ECH}};  
  \fill (-9.2684,12.4610) circle (0.22) node[right=1.5] {\ticker{EMIF}};  
  \fill (-0.0187,12.2987) circle (0.22) node[right=1.5] {\ticker{EPU}};  
  \fill (-7.9379,19.7779) circle (0.22) node[above=1.5] {\ticker{ESR}};  
  \fill (-5.8599,7.7519) circle (0.22) node[left=3,below=1.5] {\ticker{EWM}};  
  \fill (-4.3589,10.1846) circle (0.22) node[above=1.5] {\ticker{EWT}};  
  \fill (-6.4031,13.5327) circle (0.22) node[above=1.5] {\ticker{EWW}};  
  \fill (-5.1895,17.7841) circle (0.22) node[right=2,above=1.5] {\ticker{EWY}};  
  \fill (-9.4531,18.6390) circle (0.22) node[left=1,above=1.5] {\ticker{EWZ}};  
  \fill (-4.8883,17.1922) circle (0.22) node[right=1,below=1.5] {\ticker{EZA}};  
  \fill (-11.5798,14.7002) circle (0.22) node[right=1.5] {\ticker{FCHI}};  
  \fill (-12.4886,15.1321) circle (0.22) node[left=1.5] {\ticker{FXI}};  
  \fill (-8.0982,16.5838) circle (0.22) node[below=1.5] {\ticker{ILF}};  
  \fill (-8.2245,17.7093) circle (0.22) node[right=6,above=1.5] {\ticker{INDY}};  
  \fill (-3.4212,13.5828) circle (0.22) node[above=1.5] {\ticker{THD}};  
  \fill (-8.9946,23.8857) circle (0.22) node[above=1.5] {\ticker{TUR}};  
  \filldraw[fill=green,very thick] (-4.0110,7.8318) circle (0.22) node[above=2] {\textbf{E}};
  \draw[red] (-2.2,9.2) circle [x radius = 6.0cm,y radius = 3.8cm, rotate = 40];
\end{scope}  

\begin{scope}[xshift=13cm,yshift=-31cm]
  \draw[->] (-16,5) -- (1.4,5) node[right] {$e$};
  \draw (-15,5) -- (-15,4.6) node[below=2] {-2.71};
  \draw (-10,5) -- (-10,4.6) node[below=2] {15.57};
  \draw (-5,5) -- (-5,4.6) node[below=2] {33.86};
  \draw (0,5) -- (0,4.6) node[below=2] {52.14};
\end{scope}
  
\begin{scope}[xshift=13cm,yshift=-26cm,yscale=0.85]
  \fill[obtainable1]
  (-12.4886,11.9572) .. controls (-12.0608,10.8102) and (-11.6330,9.8543) ..
  (-11.2053,9.0896) .. controls (-10.5063,7.8400) and (-9.8074,6.7864) ..
  (-9.1084,5.9287) .. controls (-8.9195,5.6969) and (-8.7307,5.4763) ..
  (-8.5418,5.2667) .. controls (-7.7804,4.4218) and (-7.0190,3.7547) ..
  (-6.2576,3.2654) .. controls (-5.9690,3.0800) and (-5.6804,2.9239) ..
  (-5.3918,2.7972) .. controls (-5.2794,2.7478) and (-5.1670,2.7018) ..
  (-5.0545,2.6589) .. controls (-5.0313,2.6501) and (-5.0081,2.6414) ..
  (-4.9849,2.6329) .. controls (-4.9297,2.6127) and (-4.8744,2.5973) ..
  (-4.8192,2.5868) .. controls (-3.7782,2.3882) and (-2.7372,2.5733) ..
  (-1.6963,3.1420) .. controls (-1.6387,3.1735) and (-1.5811,3.2074) ..
  (-1.5235,3.2439) .. controls (-1.0219,3.5613) and (-0.5203,4.2727) ..
  (-0.0187,5.3778) .. controls (-0.9297,5.7783) and (-1.8407,6.5849) ..
  (-2.7517,7.7978) .. controls (-2.9749,8.3783) and (-3.1981,9.0525) ..
  (-3.4212,9.8206) .. controls (-3.6353,9.7589) and (-3.8493,9.7602) ..
  (-4.0634,9.8246) .. controls (-5.7071,12.7164) and (-7.3509,16.9306) ..
  (-8.9946,22.4671) .. controls (-10.1593,14.8749) and (-11.3239,11.3716) ..
  (-12.4886,11.9572);
  \draw[connection1] (-8.2990,9.8625) .. controls (-9.0473,9.9790) and (-9.7955,10.3588) .. (-10.5438,11.0018);  
  \draw[connection1] (-8.2990,9.8625) .. controls (-6.6423,5.6878) and (-4.9856,3.6676) .. (-3.3288,3.8019);  
  \draw[connection1] (-8.2990,9.8625) .. controls (-8.6221,8.3947) and (-8.9453,7.4182) .. (-9.2684,6.9330);  
  \draw[connection1] (-8.2990,9.8625) .. controls (-5.5389,6.9478) and (-2.7788,5.4530) .. (-0.0187,5.3778);  
  \draw[connection1] (-8.2990,9.8625) .. controls (-8.1787,10.2562) and (-8.0583,11.9235) .. (-7.9379,14.8644);  
  \draw[connection1] (-8.2990,9.8625) .. controls (-7.4860,6.5539) and (-6.6729,4.3974) .. (-5.8599,3.3932);  
  \draw[connection1] (-8.2990,9.8625) .. controls (-6.9857,7.2222) and (-5.6723,5.6168) .. (-4.3589,5.0463);  
  \draw[connection1] (-8.2990,9.8625) .. controls (-7.6671,8.3727) and (-7.0351,7.4985) .. (-6.4031,7.2399);  
  \draw[connection1] (-8.2990,9.8625) .. controls (-7.2625,9.5911) and (-6.2260,10.1981) .. (-5.1895,11.6835);  
  \draw[connection1] (-8.2990,9.8625) .. controls (-8.6837,10.2932) and (-9.0684,11.5424) .. (-9.4531,13.6103);  
  \draw[connection1] (-8.2990,9.8625) .. controls (-7.1621,9.2137) and (-6.0252,9.4600) .. (-4.8883,10.6015);  
  \draw[connection1] (-8.2990,9.8625) .. controls (-9.3926,9.3441) and (-10.4862,9.7347) .. (-11.5798,11.0344);  
  \draw[connection1] (-8.2990,9.8625) .. controls (-9.6955,9.6970) and (-11.0920,10.3952) .. (-12.4886,11.9572);  
  \draw[connection1] (-8.2990,9.8625) .. controls (-8.2321,9.4927) and (-8.1651,9.5533) .. (-8.0982,10.0442);  
  \draw[connection1] (-8.2990,9.8625) .. controls (-8.2742,9.7582) and (-8.2494,11.1422) .. (-8.2245,14.0145);  
  \draw[connection1] (-8.2990,9.8625) .. controls (-6.6731,7.8898) and (-5.0472,7.8758) .. (-3.4212,9.8206);  
  \draw[connection1] (-8.2990,9.8625) .. controls (-8.5309,11.3378) and (-8.7627,15.5393) .. (-8.9946,22.4671);  
  \draw[connection1] (-10.5438,11.0018) .. controls (-8.1388,6.1042) and (-5.7338,3.7042) .. (-3.3288,3.8019);  
  \draw[connection1] (-10.5438,11.0018) .. controls (-10.1186,9.0706) and (-9.6935,7.7143) .. (-9.2684,6.9330);  
  \draw[connection1] (-10.5438,11.0018) .. controls (-7.0354,7.5605) and (-3.5270,5.6858) .. (-0.0187,5.3778);  
  \draw[connection1] (-10.5438,11.0018) .. controls (-9.6752,10.8442) and (-8.8066,12.1317) .. (-7.9379,14.8644);  
  \draw[connection1] (-10.5438,11.0018) .. controls (-8.9825,6.9343) and (-7.4212,4.3980) .. (-5.8599,3.3932);  
  \draw[connection1] (-10.5438,11.0018) .. controls (-8.4822,7.4800) and (-6.4206,5.4949) .. (-4.3589,5.0463);  
  \draw[connection1] (-10.5438,11.0018) .. controls (-9.1636,8.8951) and (-7.7833,7.6412) .. (-6.4031,7.2399);  
  \draw[connection1] (-10.5438,11.0018) .. controls (-8.7590,9.8479) and (-6.9743,10.0751) .. (-5.1895,11.6835);  
  \draw[connection1] (-10.5438,11.0018) .. controls (-10.1802,11.1867) and (-9.8167,12.0562) .. (-9.4531,13.6103);  
  \draw[connection1] (-10.5438,11.0018) .. controls (-8.6586,9.5587) and (-6.7735,9.4253) .. (-4.8883,10.6015);  
  \draw[connection1] (-10.5438,11.0018) .. controls (-10.8891,10.2161) and (-11.2345,10.2270) .. (-11.5798,11.0344);  
  \draw[connection1] (-10.5438,11.0018) .. controls (-11.1920,10.5517) and (-11.8403,10.8701) .. (-12.4886,11.9572);  
  \draw[connection1] (-10.5438,11.0018) .. controls (-9.7286,10.1882) and (-8.9134,9.8690) .. (-8.0982,10.0442);  
  \draw[connection1] (-10.5438,11.0018) .. controls (-9.7707,10.5368) and (-8.9976,11.5410) .. (-8.2245,14.0145);  
  \draw[connection1] (-10.5438,11.0018) .. controls (-8.1696,8.2460) and (-5.7954,7.8523) .. (-3.4212,9.8206);  
  \draw[connection1] (-10.5438,11.0018) .. controls (-10.0274,12.0604) and (-9.5110,15.8821) .. (-8.9946,22.4671);  
  \draw[connection1] (-3.3288,3.8019) .. controls (-5.3087,2.9833) and (-7.2885,4.0270) .. (-9.2684,6.9330);  
  \draw[connection1] (-3.3288,3.8019) .. controls (-2.2254,2.4215) and (-1.1221,2.9468) .. (-0.0187,5.3778);  
  \draw[connection1] (-3.3288,3.8019) .. controls (-4.8652,3.9185) and (-6.4016,7.6060) .. (-7.9379,14.8644);  
  \draw[connection1] (-3.3288,3.8019) .. controls (-4.1725,2.3452) and (-5.0162,2.2090) .. (-5.8599,3.3932);  
  \draw[connection1] (-3.3288,3.8019) .. controls (-3.6722,2.5170) and (-4.0156,2.9318) .. (-4.3589,5.0463);  
  \draw[connection1] (-3.3288,3.8019) .. controls (-4.3536,3.1204) and (-5.3784,4.2664) .. (-6.4031,7.2399);  
  \draw[connection1] (-3.3288,3.8019) .. controls (-3.9490,3.6802) and (-4.5693,6.3074) .. (-5.1895,11.6835);  
  \draw[connection1] (-3.3288,3.8019) .. controls (-5.3703,3.9421) and (-7.4117,7.2116) .. (-9.4531,13.6103);  
  \draw[connection1] (-3.3288,3.8019) .. controls (-3.8487,3.3275) and (-4.3685,5.5940) .. (-4.8883,10.6015);  
  \draw[connection1] (-3.3288,3.8019) .. controls (-6.0791,3.1719) and (-8.8295,5.5827) .. (-11.5798,11.0344);  
  \draw[connection1] (-3.3288,3.8019) .. controls (-6.3821,3.2941) and (-9.4353,6.0125) .. (-12.4886,11.9572);  
  \draw[connection1] (-3.3288,3.8019) .. controls (-4.9186,3.6558) and (-6.5084,5.7365) .. (-8.0982,10.0442);  
  \draw[connection1] (-3.3288,3.8019) .. controls (-4.9607,3.4530) and (-6.5926,6.8572) .. (-8.2245,14.0145);  
  \draw[connection1] (-3.3288,3.8019) .. controls (-3.3596,2.7028) and (-3.3904,4.7091) .. (-3.4212,9.8206);  
  \draw[connection1] (-3.3288,3.8019) .. controls (-5.2174,4.1004) and (-7.1060,10.3221) .. (-8.9946,22.4671);  
  \draw[connection1] (-9.2684,6.9330) .. controls (-6.1851,5.0933) and (-3.1019,4.5749) .. (-0.0187,5.3778);  
  \draw[connection1] (-9.2684,6.9330) .. controls (-8.8249,7.8916) and (-8.3814,10.5353) .. (-7.9379,14.8644);  
  \draw[connection1] (-9.2684,6.9330) .. controls (-8.1322,4.8342) and (-6.9960,3.6542) .. (-5.8599,3.3932);  
  \draw[connection1] (-9.2684,6.9330) .. controls (-7.6319,5.1821) and (-5.9954,4.5532) .. (-4.3589,5.0463);  
  \draw[connection1] (-9.2684,6.9330) .. controls (-8.3133,6.2890) and (-7.3582,6.3913) .. (-6.4031,7.2399);  
  \draw[connection1] (-9.2684,6.9330) .. controls (-7.9088,7.0301) and (-6.5491,8.6136) .. (-5.1895,11.6835);  
  \draw[connection1] (-9.2684,6.9330) .. controls (-9.3300,8.0170) and (-9.3916,10.2427) .. (-9.4531,13.6103);  
  \draw[connection1] (-9.2684,6.9330) .. controls (-7.8084,6.8324) and (-6.3484,8.0552) .. (-4.8883,10.6015);  
  \draw[connection1] (-9.2684,6.9330) .. controls (-10.0389,7.1909) and (-10.8093,8.5580) .. (-11.5798,11.0344);  
  \draw[connection1] (-9.2684,6.9330) .. controls (-10.3418,7.4043) and (-11.4152,9.0790) .. (-12.4886,11.9572);  
  \draw[connection1] (-9.2684,6.9330) .. controls (-8.8783,7.2261) and (-8.4882,8.2631) .. (-8.0982,10.0442);  
  \draw[connection1] (-9.2684,6.9330) .. controls (-8.9204,7.4399) and (-8.5725,9.8004) .. (-8.2245,14.0145);  
  \draw[connection1] (-9.2684,6.9330) .. controls (-7.3193,5.6124) and (-5.3703,6.5749) .. (-3.4212,9.8206);  
  \draw[connection1] (-9.2684,6.9330) .. controls (-9.1771,8.7716) and (-9.0859,13.9496) .. (-8.9946,22.4671);  
  \draw[connection1] (-0.0187,5.3778) .. controls (-2.6584,5.9096) and (-5.2982,9.0718) .. (-7.9379,14.8644);  
  \draw[connection1] (-0.0187,5.3778) .. controls (-1.9657,3.3377) and (-3.9128,2.6761) .. (-5.8599,3.3932);  
  \draw[connection1] (-0.0187,5.3778) .. controls (-1.4654,3.6194) and (-2.9122,3.5089) .. (-4.3589,5.0463);  
  \draw[connection1] (-0.0187,5.3778) .. controls (-2.1468,4.5140) and (-4.2750,5.1347) .. (-6.4031,7.2399);  
  \draw[connection1] (-0.0187,5.3778) .. controls (-1.7423,5.1465) and (-3.4659,7.2484) .. (-5.1895,11.6835);  
  \draw[connection1] (-0.0187,5.3778) .. controls (-3.1635,5.8417) and (-6.3083,8.5859) .. (-9.4531,13.6103);  
  \draw[connection1] (-0.0187,5.3778) .. controls (-1.6419,5.2905) and (-3.2651,7.0318) .. (-4.8883,10.6015);  
  \draw[connection1] (-0.0187,5.3778) .. controls (-3.8724,5.2635) and (-7.7261,7.1490) .. (-11.5798,11.0344);  
  \draw[connection1] (-0.0187,5.3778) .. controls (-4.1753,5.3917) and (-8.3319,7.5848) .. (-12.4886,11.9572);  
  \draw[connection1] (-0.0187,5.3778) .. controls (-2.7118,5.4925) and (-5.4050,7.0479) .. (-8.0982,10.0442);  
  \draw[connection1] (-0.0187,5.3778) .. controls (-2.7540,5.4185) and (-5.4892,8.2974) .. (-8.2245,14.0145);  
  \draw[connection1] (-0.0187,5.3778) .. controls (-1.1529,4.4361) and (-2.2870,5.9171) .. (-3.4212,9.8206);  
  \draw[connection1] (-0.0187,5.3778) .. controls (-3.0107,6.6930) and (-6.0026,12.3894) .. (-8.9946,22.4671);  
  \draw[connection1] (-7.9379,14.8644) .. controls (-7.2453,8.6425) and (-6.5526,4.8188) .. (-5.8599,3.3932);  
  \draw[connection1] (-7.9379,14.8644) .. controls (-6.7449,9.2191) and (-5.5519,5.9465) .. (-4.3589,5.0463);  
  \draw[connection1] (-7.9379,14.8644) .. controls (-7.4263,10.3878) and (-6.9147,7.8463) .. (-6.4031,7.2399);  
  \draw[connection1] (-7.9379,14.8644) .. controls (-7.0218,11.4852) and (-6.1056,10.4249) .. (-5.1895,11.6835);  
  \draw[connection1] (-7.9379,14.8644) .. controls (-8.4430,12.3066) and (-8.9481,11.8886) .. (-9.4531,13.6103);  
  \draw[connection1] (-7.9379,14.8644) .. controls (-6.9214,11.7510) and (-5.9049,10.3301) .. (-4.8883,10.6015);  
  \draw[connection1] (-7.9379,14.8644) .. controls (-9.1519,11.0528) and (-10.3658,9.7761) .. (-11.5798,11.0344);  
  \draw[connection1] (-7.9379,14.8644) .. controls (-9.4548,11.2882) and (-10.9717,10.3191) .. (-12.4886,11.9572);  
  \draw[connection1] (-7.9379,14.8644) .. controls (-7.9913,11.4618) and (-8.0448,9.8550) .. (-8.0982,10.0442);  
  \draw[connection1] (-7.9379,14.8644) .. controls (-8.0335,12.2082) and (-8.1290,11.9249) .. (-8.2245,14.0145);  
  \draw[connection1] (-7.9379,14.8644) .. controls (-6.4324,9.5812) and (-4.9268,7.8999) .. (-3.4212,9.8206);  
  \draw[connection1] (-7.9379,14.8644) .. controls (-8.2902,14.2681) and (-8.6424,16.8023) .. (-8.9946,22.4671);  
  \draw[connection1] (-5.8599,3.3932) .. controls (-5.3596,2.9184) and (-4.8593,3.4695) .. (-4.3589,5.0463);  
  \draw[connection1] (-5.8599,3.3932) .. controls (-6.0410,3.4913) and (-6.2220,4.7735) .. (-6.4031,7.2399);  
  \draw[connection1] (-5.8599,3.3932) .. controls (-5.6364,4.3421) and (-5.4130,7.1055) .. (-5.1895,11.6835);  
  \draw[connection1] (-5.8599,3.3932) .. controls (-7.0576,4.4189) and (-8.2554,7.8246) .. (-9.4531,13.6103);  
  \draw[connection1] (-5.8599,3.3932) .. controls (-5.5360,3.9641) and (-5.2122,6.3668) .. (-4.8883,10.6015);  
  \draw[connection1] (-5.8599,3.3932) .. controls (-7.7665,3.9883) and (-9.6732,6.5353) .. (-11.5798,11.0344);  
  \draw[connection1] (-5.8599,3.3932) .. controls (-8.0694,4.2380) and (-10.2790,7.0927) .. (-12.4886,11.9572);  
  \draw[connection1] (-5.8599,3.3932) .. controls (-6.6060,4.0364) and (-7.3521,6.2534) .. (-8.0982,10.0442);  
  \draw[connection1] (-5.8599,3.3932) .. controls (-6.6481,4.3181) and (-7.4363,7.8585) .. (-8.2245,14.0145);  
  \draw[connection1] (-5.8599,3.3932) .. controls (-5.0470,3.4165) and (-4.2341,5.5590) .. (-3.4212,9.8206);  
  \draw[connection1] (-5.8599,3.3932) .. controls (-6.9048,5.3405) and (-7.9497,11.6984) .. (-8.9946,22.4671);  
  \draw[connection1] (-4.3589,5.0463) .. controls (-5.0403,4.5101) and (-5.7217,5.2413) .. (-6.4031,7.2399);  
  \draw[connection1] (-4.3589,5.0463) .. controls (-4.6358,5.8461) and (-4.9126,8.0585) .. (-5.1895,11.6835);  
  \draw[connection1] (-4.3589,5.0463) .. controls (-6.0570,5.4571) and (-7.7551,8.3118) .. (-9.4531,13.6103);  
  \draw[connection1] (-4.3589,5.0463) .. controls (-4.5354,5.0673) and (-4.7119,6.9190) .. (-4.8883,10.6015);  
  \draw[connection1] (-4.3589,5.0463) .. controls (-6.7659,5.0817) and (-9.1728,7.0777) .. (-11.5798,11.0344);  
  \draw[connection1] (-4.3589,5.0463) .. controls (-7.0688,5.3644) and (-9.7787,7.6680) .. (-12.4886,11.9572);  
  \draw[connection1] (-4.3589,5.0463) .. controls (-5.6053,5.0512) and (-6.8518,6.7172) .. (-8.0982,10.0442);  
  \draw[connection1] (-4.3589,5.0463) .. controls (-5.6475,5.4049) and (-6.9360,8.3943) .. (-8.2245,14.0145);  
  \draw[connection1] (-4.3589,5.0463) .. controls (-4.0464,4.1389) and (-3.7338,5.7304) .. (-3.4212,9.8206);  
  \draw[connection1] (-4.3589,5.0463) .. controls (-5.9042,5.9642) and (-7.4494,11.7711) .. (-8.9946,22.4671);  
  \draw[connection1] (-6.4031,7.2399) .. controls (-5.9986,7.2043) and (-5.5940,8.6855) .. (-5.1895,11.6835);  
  \draw[connection1] (-6.4031,7.2399) .. controls (-7.4198,7.9460) and (-8.4365,10.0695) .. (-9.4531,13.6103);  
  \draw[connection1] (-6.4031,7.2399) .. controls (-5.8982,7.1427) and (-5.3933,8.2632) .. (-4.8883,10.6015);  
  \draw[connection1] (-6.4031,7.2399) .. controls (-8.1287,6.9769) and (-9.8542,8.2417) .. (-11.5798,11.0344);  
  \draw[connection1] (-6.4031,7.2399) .. controls (-8.4316,7.2424) and (-10.4601,8.8149) .. (-12.4886,11.9572);  
  \draw[connection1] (-6.4031,7.2399) .. controls (-6.9681,7.4409) and (-7.5332,8.3757) .. (-8.0982,10.0442);  
  \draw[connection1] (-6.4031,7.2399) .. controls (-7.0103,7.2143) and (-7.6174,9.4725) .. (-8.2245,14.0145);  
  \draw[connection1] (-6.4031,7.2399) .. controls (-5.4092,5.4527) and (-4.4152,6.3130) .. (-3.4212,9.8206);  
  \draw[connection1] (-6.4031,7.2399) .. controls (-7.2670,8.8539) and (-8.1308,13.9296) .. (-8.9946,22.4671);  
  \draw[connection1] (-5.1895,11.6835) .. controls (-6.6107,10.0386) and (-8.0319,10.6809) .. (-9.4531,13.6103);  
  \draw[connection1] (-5.1895,11.6835) .. controls (-5.0891,9.5262) and (-4.9887,9.1656) .. (-4.8883,10.6015);  
  \draw[connection1] (-5.1895,11.6835) .. controls (-7.3196,9.5372) and (-9.4497,9.3208) .. (-11.5798,11.0344);  
  \draw[connection1] (-5.1895,11.6835) .. controls (-7.6225,9.8271) and (-10.0555,9.9183) .. (-12.4886,11.9572);  
  \draw[connection1] (-5.1895,11.6835) .. controls (-6.1591,9.4979) and (-7.1286,8.9514) .. (-8.0982,10.0442);  
  \draw[connection1] (-5.1895,11.6835) .. controls (-6.2012,9.7979) and (-7.2129,10.5749) .. (-8.2245,14.0145);  
  \draw[connection1] (-5.1895,11.6835) .. controls (-4.6001,8.2390) and (-4.0107,7.6180) .. (-3.4212,9.8206);  
  \draw[connection1] (-5.1895,11.6835) .. controls (-6.4579,11.3474) and (-7.7262,14.9420) .. (-8.9946,22.4671);  
  \draw[connection1] (-9.4531,13.6103) .. controls (-7.9315,10.6477) and (-6.4099,9.6448) .. (-4.8883,10.6015);  
  \draw[connection1] (-9.4531,13.6103) .. controls (-10.1620,10.9742) and (-10.8709,10.1156) .. (-11.5798,11.0344);  
  \draw[connection1] (-9.4531,13.6103) .. controls (-10.4649,11.2580) and (-11.4767,10.7070) .. (-12.4886,11.9572);  
  \draw[connection1] (-9.4531,13.6103) .. controls (-9.0015,12.0182) and (-8.5498,10.8295) .. (-8.0982,10.0442);  
  \draw[connection1] (-9.4531,13.6103) .. controls (-9.0436,11.2305) and (-8.6341,11.3653) .. (-8.2245,14.0145);  
  \draw[connection1] (-9.4531,13.6103) .. controls (-7.4425,9.1422) and (-5.4319,7.8790) .. (-3.4212,9.8206);  
  \draw[connection1] (-9.4531,13.6103) .. controls (-9.3003,13.2656) and (-9.1474,16.2179) .. (-8.9946,22.4671);  
  \draw[connection1] (-4.8883,10.6015) .. controls (-7.1188,8.7567) and (-9.3493,8.9011) .. (-11.5798,11.0344);  
  \draw[connection1] (-4.8883,10.6015) .. controls (-7.4217,8.9800) and (-9.9552,9.4319) .. (-12.4886,11.9572);  
  \draw[connection1] (-4.8883,10.6015) .. controls (-5.9583,9.1967) and (-7.0282,9.0109) .. (-8.0982,10.0442);  
  \draw[connection1] (-4.8883,10.6015) .. controls (-6.0004,8.9165) and (-7.1125,10.0542) .. (-8.2245,14.0145);  
  \draw[connection1] (-4.8883,10.6015) .. controls (-4.3993,7.5062) and (-3.9103,7.2459) .. (-3.4212,9.8206);  
  \draw[connection1] (-4.8883,10.6015) .. controls (-6.2571,11.2299) and (-7.6259,15.1852) .. (-8.9946,22.4671);  
  \draw[connection1] (-11.5798,11.0344) .. controls (-11.8827,11.0403) and (-12.1856,11.3479) .. (-12.4886,11.9572);  
  \draw[connection1] (-11.5798,11.0344) .. controls (-10.4193,9.2793) and (-9.2587,8.9492) .. (-8.0982,10.0442);  
  \draw[connection1] (-11.5798,11.0344) .. controls (-10.4614,9.7918) and (-9.3430,10.7852) .. (-8.2245,14.0145);  
  \draw[connection1] (-11.5798,11.0344) .. controls (-8.8603,8.1056) and (-6.1408,7.7011) .. (-3.4212,9.8206);  
  \draw[connection1] (-11.5798,11.0344) .. controls (-10.7181,11.0494) and (-9.8563,14.8603) .. (-8.9946,22.4671);  
  \draw[connection1] (-12.4886,11.9572) .. controls (-11.0251,9.8500) and (-9.5616,9.2123) .. (-8.0982,10.0442);  
  \draw[connection1] (-12.4886,11.9572) .. controls (-11.0672,10.3485) and (-9.6459,11.0343) .. (-8.2245,14.0145);  
  \draw[connection1] (-12.4886,11.9572) .. controls (-9.4661,8.6153) and (-6.4437,7.9031) .. (-3.4212,9.8206);  
  \draw[connection1] (-12.4886,11.9572) .. controls (-11.3239,11.3716) and (-10.1593,14.8749) .. (-8.9946,22.4671);  
  \draw[connection1] (-8.0982,10.0442) .. controls (-8.1403,9.1734) and (-8.1824,10.4969) .. (-8.2245,14.0145);  
  \draw[connection1] (-8.0982,10.0442) .. controls (-6.5392,7.4562) and (-4.9802,7.3817) .. (-3.4212,9.8206);  
  \draw[connection1] (-8.0982,10.0442) .. controls (-8.3970,11.2071) and (-8.6958,15.3481) .. (-8.9946,22.4671);  
  \draw[connection1] (-8.2245,14.0145) .. controls (-6.6234,9.1958) and (-5.0223,7.7978) .. (-3.4212,9.8206);  
  \draw[connection1] (-8.2245,14.0145) .. controls (-8.4812,12.8579) and (-8.7379,15.6754) .. (-8.9946,22.4671);  
  \draw[connection1] (-3.4212,9.8206) .. controls (-5.2790,9.2848) and (-7.1368,13.5003) .. (-8.9946,22.4671);  
  \draw[boundary1]
  (-0.0187,5.3778) .. controls (-0.9297,5.7783) and (-1.8407,6.5849) ..
  (-2.7517,7.7978) .. controls (-2.9749,8.3783) and (-3.1981,9.0525) ..
  (-3.4212,9.8206) .. controls (-3.6353,9.7589) and (-3.8493,9.7602) ..
  (-4.0634,9.8246) .. controls (-5.7071,12.7164) and (-7.3509,16.9306) ..
  (-8.9946,22.4671) .. controls (-10.1593,14.8749) and (-11.3239,11.3716) ..
  (-12.4886,11.9572);
  \draw[very thick,dashed]
  (-12.4886,11.9572) .. controls (-10.2790,7.0927) and (-8.0694,4.2380) ..
  (-5.8599,3.3932) .. controls (-5.0162,2.2090) and (-4.1725,2.3452) ..
  (-3.3288,3.8019) .. controls (-2.2254,2.4215) and (-1.1221,2.9468) ..
  (-0.0187,5.3778);
  \draw[very thick]
  (-12.4886,11.9572) .. controls (-12.0608,10.8102) and (-11.6330,9.8543) ..
  (-11.2053,9.0896) .. controls (-10.5063,7.8400) and (-9.8074,6.7864) ..
  (-9.1084,5.9287) .. controls (-8.9195,5.6969) and (-8.7307,5.4763) ..
  (-8.5418,5.2667) .. controls (-7.7804,4.4218) and (-7.0190,3.7547) ..
  (-6.2576,3.2654) .. controls (-5.9690,3.0800) and (-5.6804,2.9239) ..
  (-5.3918,2.7972) .. controls (-5.2794,2.7478) and (-5.1670,2.7018) ..
  (-5.0545,2.6589) .. controls (-5.0313,2.6501) and (-5.0081,2.6414) ..
  (-4.9849,2.6329) .. controls (-4.9297,2.6127) and (-4.8744,2.5973) ..
  (-4.8192,2.5868) .. controls (-3.7782,2.3882) and (-2.7372,2.5733) ..
  (-1.6963,3.1420) .. controls (-1.6387,3.1735) and (-1.5811,3.2074) ..
  (-1.5235,3.2439) .. controls (-1.0219,3.5613) and (-0.5203,4.2727) ..
  (-0.0187,5.3778);
  \fill (-8.2990,9.8625) circle (0.22) node[left=1.5] {\ticker{EEM}};  
  \fill (-10.5438,11.0018) circle (0.22) node[right=3,above=1.5] {\ticker{BKF}};  
  \fill (-3.3288,3.8019) circle (0.22) node[right=3,above=1.5] {\ticker{ECH}};  
  \fill (-9.2684,6.9330) circle (0.22) node[right=3,above=1.5] {\ticker{EMIF}};  
  \fill (-0.0187,5.3778) circle (0.22) node[right=2,above=1.5] {\ticker{EPU}};  
  \fill (-7.9379,14.8644) circle (0.22) node[above=1.5] {\ticker{ESR}};  
  \fill (-5.8599,3.3932) circle (0.22) node[above=1.5] {\ticker{EWM}};  
  \fill (-4.3589,5.0463) circle (0.22) node[above=1.5] {\ticker{EWT}};  
  \fill (-6.4031,7.2399) circle (0.22) node[above=1.5] {\ticker{EWW}};  
  \fill (-5.1895,11.6835) circle (0.22) node[above=1.5] {\ticker{EWY}};  
  \fill (-9.4531,13.6103) circle (0.22) node[left=2,above=1.5] {\ticker{EWZ}};  
  \fill (-4.8883,10.6015) circle (0.22) node[below=1.5] {\ticker{EZA}};  
  \fill (-11.5798,11.0344) circle (0.22) node[below=2,left=3] {\ticker{FCHI}};  
  \fill (-12.4886,11.9572) circle (0.22) node[left=3,above=1.5] {\ticker{FXI}};  
  \fill (-8.0982,10.0442) circle (0.22) node[right=3,above=1.5] {\ticker{ILF}};  
  \fill (-8.2245,14.0145) circle (0.22) node[right=4,below=1.5] {\ticker{INDY}};  
  \fill (-3.4212,9.8206) circle (0.22) node[right=3,above=1.5] {\ticker{THD}};  
  \fill (-8.9946,22.4671) circle (0.22) node[above=1.5] {\ticker{TUR}};  
  \draw[->] (-4.0110,2.5097) -- (-4.0110,0.8) -- (-0.0110,0.8);
  \node at (-1.8110,1.8) {efficient};
  \filldraw[fill=green,very thick] (-4.0110,2.5097) circle (0.22) node[left=8,below=-1.5] {\textbf{E}};

\end{scope} 

\end{tikzpicture}

\end{figure}
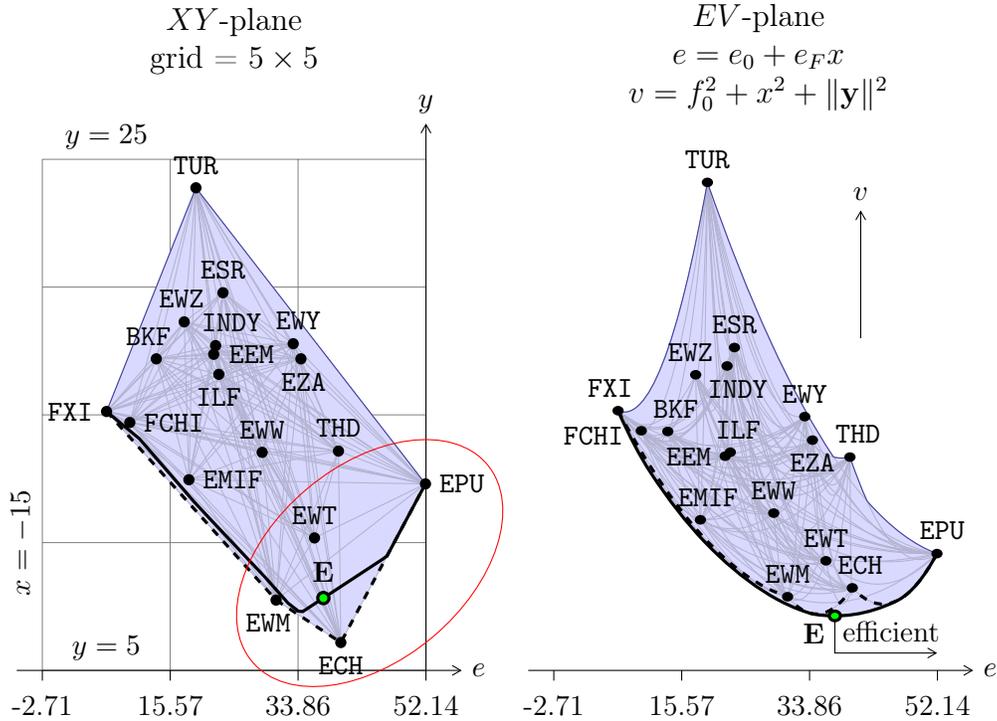

As in Section \ref{xy-plane5}, \textbf{E} is the image of the efficient
portfolio, $\mathbf{p}_\textbf{E}$, of absolute minimum variance. All
efficient portfolios from this 18 ETF universe are made up of the four
funds circled in red. These four funds have the least risk ($\sigma$) of
the eighteeen, and their expected returns are among the highest. This is
an unusual situation---where risk and return seem to be inversely
related.

As in Section \ref{xy-plane5}, the solid black path in either picture
corresponds to the set of minimum-variance portfolios. The path of
minimum-$|y|$ portfolios is dashed. The faint interior lines are
two-security-portfolio paths.

The efficient, minimum-$\|\mathbf{y}\|$ portfolio at
$x = x_\texttt{ECH} = -3.33$,
\[
\mathbf{p} =
  41.5\%~\texttt{ECH} + 24.0\%~\texttt{EPU} +
  29.3\%~\texttt{EWM} + 5.1\%~\texttt{EWT},
\]
has expected return $e = 39.97$ and risk $\sigma = 12.71$. On the other
hand, the single security \texttt{ECH} is the minimum-$|y|$ portfolio at
this value of $x$ (and $e$), but the risk of \texttt{ECH} is
$\sigma_\texttt{ECH}$ = 14.46, or 13.79\% more than the efficient value.
Apparently, in this case, the minimum-$|y|$ portfolio is not a good
approximation of the minimum-$\|\mathbf{y}\|$ portfolio. This
is apparent in Figure \ref{rd_tikz18}.


\subsubsection{The efficient four}\label{4emfunds}

Let us restrict our attention to the four emerging market funds,
\texttt{ECH}, \texttt{EPU}, \texttt{EPU}, \texttt{EPU}, that make up
the efficient portfolios of Figure \ref{rd_tikz18}.
Figure \ref{growth4_2010} shows how these funds grew in 2010, and Table
\ref{decomp4_2010L200} shows the output of \textbf{rtndecomp}
restricted to their returns over the last 200 market days of 2010.

\begin{figure}[ht]
  \centering
  \caption{\label{growth4_2010}%
    2010 adjusted closing prices of four emerging market ETFs\\
    prices normalized at 100 on 2010-12-31}
  \includegraphics{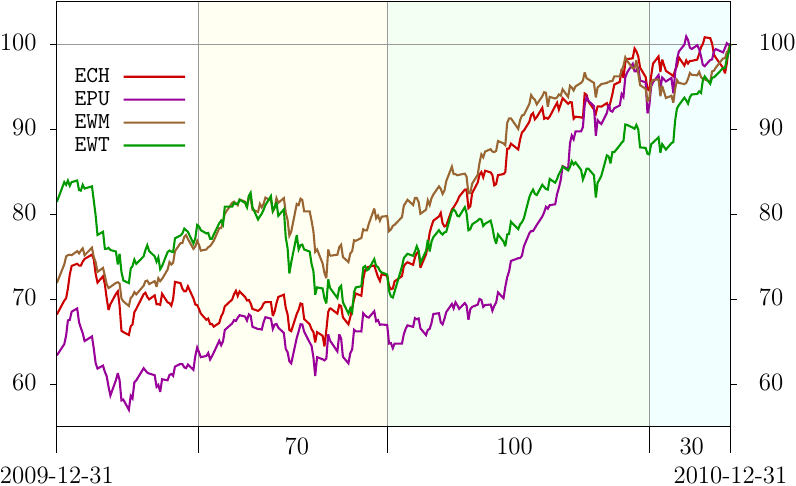}
\end{figure}

\begin{table}[H]
  \centering
  \caption{\label{decomp4_2010L200}%
    Decomposition of return data -- 4 emerging market ETFs\\
    last 200 market-days of 2010 -- late-heavy weights
  }
{ \small
  \newcolumntype{.}{D{.}{.}{2}}
  \newcolumntype{W}{|>{\columncolor{white}}c|}
$
\begin{array}{|c|....||r.|}
\cline{1-5}\rule{0mm}{4mm}
  \text{fund}
  & \multicolumn{1}{r}{\parbox{6.5ex}{\hfill\ticker{ECH}\hspace*{0.2ex}}}
  & \multicolumn{1}{r}{\parbox{7.0ex}{\hfill\ticker{EPU}\hspace*{0.2ex}}}
  & \multicolumn{1}{r}{\parbox{7.0ex}{\hfill\ticker{EWM}\hspace*{0.2ex}}}
  & \multicolumn{1}{r|}{\parbox{7.0ex}{\hfill\ticker{EWT}\hspace*{0.2ex}}} \\
\hline\rule[-1.5mm]{0mm}{6mm}
  E &   39.97 &  52.07 &  30.71 &  36.20
  & \multicolumn{2}{c|}{\widehat{V}_\text{T}} \\
\hline\rowcolor{productive}
  \multicolumn{1}{W}{\rule{0mm}{4mm}}
  & 1.45 & 8.51 & -3.94 & -0.74 & 91 & 31.0\% \\
\cline{2-7}
  \rowcolor{nonproductive}\multicolumn{1}{W}{\rule{0mm}{4mm}$F$}
  & -5.98 & 5.01 & 3.04 & 7.55 & 127 & 43.6\% \\
  \rowcolor{nonproductive}\multicolumn{1}{W}{}
  & -3.44 & 3.58 & 3.04 & -6.33 & 74 & 25.4\% \\
\hhline{|=|====#==|}\rule{0mm}{4mm}
  & \multicolumn{1}{r@{\hspace*{1.1ex}}}{50}
  & \multicolumn{1}{r@{\hspace*{1.1ex}}}{110}
  & \multicolumn{1}{r@{\hspace*{1.1ex}}}{34}
  & \multicolumn{1}{r@{\hspace*{1.1ex}}||}{98}
  & \multicolumn{1}{r}{292}
  & \multicolumn{1}{r|}{100\%\hspace*{-0.8ex}} \\
\cline{7-7} \raisebox{1.6ex}[0pt]{$\widehat{V}_\text{T}$}
  & 17.0\% & 37.8\% & 11.7\% & 33.5\%
  & \multicolumn{1}{r|}{100\%\hspace*{-0.5ex}}\\
\cline{1-6}
  \multicolumn{6}{c}{\rule{0mm}{5mm}
    \text{with}~f_0 = 12.62,~ e_0 = 37.47,~\text{and}~e_F = 1.717
  }
\end{array}
$
}
\end{table}

The $XY$ and $EV$ planar representations of Table \ref{decomp4_2010L200}
are shown in Figure \ref{rd_tikz4}.

\begin{figure}[ht]
  \centering
  \caption{\label{rd_tikz4}
  Planar representations of return data -- 4 emerging markets ETFs\\ 
  last 200 market-days of 2010 -- late-heavy weights}

\begin{tikzpicture}[scale=0.27,>={angle 60}]

  \node at (-14.5,-2) {\parbox{5cm}{\centering%
    $XY$-plane\\ 
    grid = $5\times5$\\
    ~
  }};
  \node at (8,-2) {\parbox{5cm}{\centering%
    $EV$-plane\\
    $e = e_0 + e_F x$\\
    $v = f_0^2 + x^2 + \|\mathbf{y}\|^2$ 
  }};
\small
  \node at (-3,-33) {{\normalsize$\mathbf{p}_\textbf{E}$}
  = 39.5\% \ticker{ECH} + 13.0\% \ticker{EPU} + 41.6\% \ticker{EWM} + 5.9\% \ticker{EWT}};
  \draw[->] (16,-22) -- (16,-17) node[above] {$v$} ;

\begin{scope}[xshift=-17.5cm,yshift=-18cm]
  \draw[help lines] (-5,-10) grid[step=5] (10,10);
  \draw[->] (-6,0) -- (12,0) node[right] {$x$};
  \draw[->] (-6,-10) -- (12,-10) node[right] {$e$};
  \draw[->] (0,-10) -- (0,12) node[above] {$y$};
  \draw (-5,-10) -- (-5,-10.5) node[below=2] {28.89};
  \draw (0,-10) -- (0,-10.5) node[below=2] {37.47};
  \draw (5,-10) -- (5,-10.5) node[below=2] {46.05};
  \draw (10,-10) -- (10,-10.5) node[below=2] {54.64};
  \fill[obtainable1]
  (1.4534,-5.9791) --  
  (8.5055,5.0135) --  
  (-0.7412,7.5475) --  
  (-3.9389,3.0360) --  
  cycle;
  \draw[connection1] (1.4534,-5.9791) -- (8.5055,5.0135);  
  \draw[connection1] (1.4534,-5.9791) -- (-3.9389,3.0360);  
  \draw[connection1] (1.4534,-5.9791) -- (-0.7412,7.5475);  
  \draw[connection1] (8.5055,5.0135) -- (-3.9389,3.0360);  
  \draw[connection1] (8.5055,5.0135) -- (-0.7412,7.5475);  
  \draw[connection1] (-3.9389,3.0360) -- (-0.7412,7.5475);  
  \draw[boundary1] (-2.1230,0)
  -- (1.4534,-5.9791)
  -- (8.5055,5.0135)
  -- (-0.7412,7.5475)
  -- (-3.9389,3.0360);
  \draw[very thick,dashed] (-3.9389,3.0360)
  -- (-2.1230,0.0000)
  -- (5.2892,-0.0000)
  -- (8.5055,5.0135);
  \draw[very thick] (-3.9389,3.0360)
  -- (-2.6049,0.8057)
  -- (-1.7218,-0.0000)
  -- (4.9315,0.0000)
  -- (5.2995,0.0161)
  -- (8.5055,5.0135);
  \fill (1.4534,-5.9791) circle (0.30) node[right=2,below=1.5] {\ticker{ECH}};  
  \fill (8.5055,5.0135) circle (0.30) node[above=1.5] {\ticker{EPU}};  
  \fill (-3.9389,3.0360) circle (0.30) node[left=1.5] {\ticker{EWM}};  
  \fill (-0.7412,7.5475) circle (0.30) node[left=5,above=1.5] {\ticker{EWT}};  
  \filldraw[fill=green] (0.0000,-0.0000) circle (0.30) node[above=1.5] {\textbf{E}};
\end{scope}  

\begin{scope}[xshift=6cm,yshift=-18cm]
  \draw[->] (-6,-10) -- (12,-10) node[right] {$e$};
  \draw (-5,-10) -- (-5,-10.5) node[below=2] {28.89};
  \draw (0,-10) -- (0,-10.5) node[below=2] {37.47};
  \draw (5,-10) -- (5,-10.5) node[below=2] {46.05};
  \draw (10,-10) -- (10,-10.5) node[below=2] {54.64};
\end{scope}

\begin{scope}[xshift=6cm,yshift=-25cm]
  \fill[obtainable1]
  (-3.9389,5.0968) .. controls (-3.4943,3.4066) and (-3.0496,2.1827) ..
  (-2.6049,1.4250) .. controls (-2.3105,0.9235) and (-2.0162,0.5968) ..
  (-1.7218,0.4447) .. controls (0.4959,-0.7009) and (2.7137,0.3669) ..
  (4.9315,3.6480) .. controls (5.0542,3.8295) and (5.1769,4.0253) ..
  (5.2995,4.2356) .. controls (6.3682,6.0670) and (7.4369,10.1708) ..
  (8.5055,16.5468) .. controls (7.6626,14.1552) and (6.8198,12.0920) ..
  (5.9769,10.3571) .. controls (3.7375,6.7644) and (1.4981,8.1901) ..
  (-0.7412,14.6342) .. controls (-1.4759,8.3641) and (-2.2106,4.9053) ..
  (-2.9453,4.2579) .. controls (-3.2765,4.4868) and (-3.6077,4.7664) ..
  (-3.9389,5.0968);
  \draw[connection1] (1.4534,7.4549) .. controls (3.8041,-0.5090) and (6.1548,2.5216) .. (8.5055,16.5468);  
  \draw[connection1] (1.4534,7.4549) .. controls (-0.3440,-0.9489) and (-2.1415,-1.7350) .. (-3.9389,5.0968);  
  \draw[connection1] (1.4534,7.4549) .. controls (0.7219,0.0418) and (-0.0097,2.4349) .. (-0.7412,14.6342);  
  \draw[connection1] (8.5055,16.5468) .. controls (4.3574,4.7767) and (0.2092,0.9600) .. (-3.9389,5.0968);  
  \draw[connection1] (8.5055,16.5468) .. controls (5.4233,6.4021) and (2.3410,5.7646) .. (-0.7412,14.6342);  
  \draw[connection1] (-3.9389,5.0968) .. controls (-2.8730,2.3580) and (-1.8071,5.5371) .. (-0.7412,14.6342);  
  \draw[boundary1]
  (8.5055,16.5468) .. controls (7.6626,14.1552) and (6.8198,12.0920) ..
  (5.9769,10.3571) .. controls (3.7375,6.7644) and (1.4981,8.1901) ..
  (-0.7412,14.6342) .. controls (-1.4759,8.3641) and (-2.2106,4.9053) ..
  (-2.9453,4.2579) .. controls (-3.2765,4.4868) and (-3.6077,4.7664) ..
  (-3.9389,5.0968);
  \draw[very thick,dashed]
  (-3.9389,5.0968) .. controls (-3.3336,2.7961) and (-2.7283,1.3593) ..
  (-2.1230,0.7865) .. controls (0.3477,-0.8282) and (2.8185,0.3156) ..
  (5.2892,4.2179) .. controls (6.3613,6.0405) and (7.4334,10.1502) ..
  (8.5055,16.5468);
  \draw[very thick]
  (-3.9389,5.0968) .. controls (-3.4943,3.4066) and (-3.0496,2.1827) ..
  (-2.6049,1.4250) .. controls (-2.3105,0.9235) and (-2.0162,0.5968) ..
  (-1.7218,0.4447) .. controls (0.4959,-0.7009) and (2.7137,0.3669) ..
  (4.9315,3.6480) .. controls (5.0542,3.8295) and (5.1769,4.0253) ..
  (5.2995,4.2356) .. controls (6.3682,6.0670) and (7.4369,10.1708) ..
  (8.5055,16.5468);
  \fill (1.4534,7.4549) circle (0.30) node[above=1.5] {\ticker{ECH}};  
  \fill (8.5055,16.5468) circle (0.30) node[above=1.5] {\ticker{EPU}};  
  \fill (-3.9389,5.0968) circle (0.30) node[above=1.5] {\ticker{EWM}};  
  \fill (-0.7412,14.6342) circle (0.30) node[above=1.5] {\ticker{EWT}};  
  \draw[->] (0.0000,0.0000) -- (0.0000,-1.8000) -- (6.0000,-1.8000);
  \node at (3.0000,-0.7200) {efficient};
  \filldraw[fill=green] (0.0000,0.0000) circle (0.30) node[left=8,below=-1.5] {\textbf{E}};
\end{scope} 

\end{tikzpicture}

\end{figure}
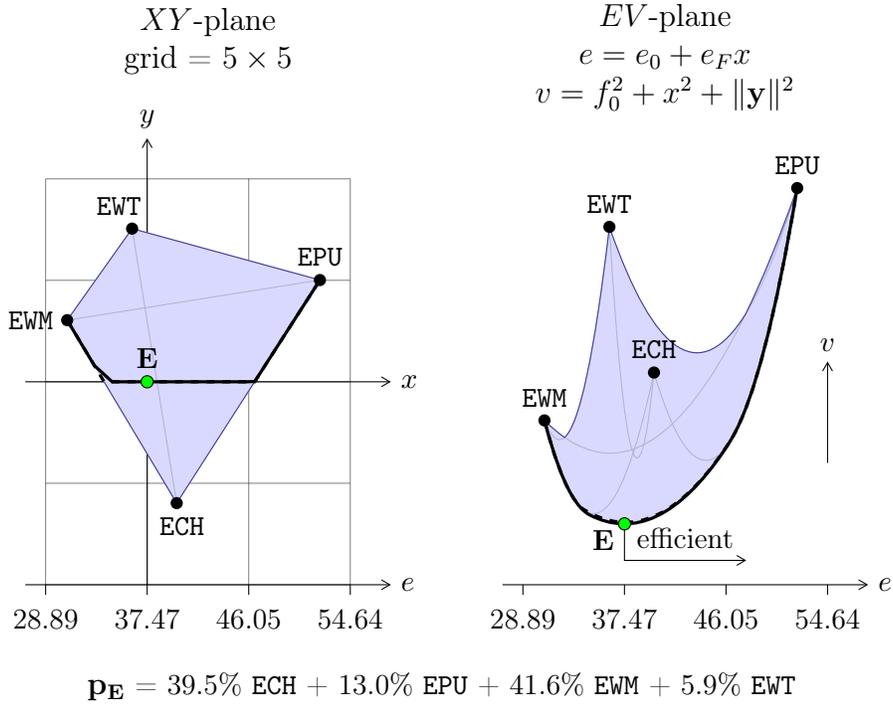

The minimum-$\|\mathbf{y}\|$ portfolio path in Table \ref{efficient4_2010L200} is
efficient in either the current four-fund universe or in the original
eighteen-emerging-market-fund universe. On the other hand the minimum
$|y|$-path is based on the four-fund $XY$ representation in Figure
\ref{rd_tikz4}. This minimum-$|y|$ path is an extremely close
approximation of the efficient path, with the maximum
$\sigma$-difference of less than 0.2\% occurring at $x = x_E = 0$.

\begin{table}[H]
  \centering
  \caption{\label{efficient4_2010L200}%
    Optimal portfolio paths -- 4 emerging market ETF universe \\
    last 200 days of 2010 -- late-heavy weights
  }
{ \small
  \begin{tabular}{|c|rrrr|rrr|}
  \cline{2-8}
  \multicolumn{1}{l|}{\rule{0mm}{4mm}}
    & \multicolumn{4}{c|}{minimum-$\|\mathbf{y}\|$ (efficient) path}
    & \multicolumn{3}{c|}{minimum-$|y|$ path} \\
  \cline{2-8}
  \multicolumn{1}{l|}{\rule{0mm}{2mm}}
    & & \multicolumn{2}{c}{corner} & & \multicolumn{1}{c}{$x_\textbf{E}$} & \multicolumn{1}{c}{corner} & \\
  \multicolumn{1}{l|}{}
    & \multicolumn{1}{c}{ \raisebox{1.2ex}[0pt]{\normalsize$\mathbf{p}_\textbf{E}$}}
    & \multicolumn{2}{c}{portfolios}
    & \multicolumn{1}{c|}{ \raisebox{1.2ex}[0pt]{\ticker{EPU}}}
    & \multicolumn{1}{c}{portf}
    & \multicolumn{1}{c}{portf}
    & \multicolumn{1}{c|}{ \raisebox{1.2ex}[0pt]{\ticker{EPU}}} \\
  \hline\rule[-1.5mm]{0mm}{6mm}%
    \ticker{ECH} & 0.395 & 0.464 &  0.455 & \multicolumn{1}{c|}{0}
      & 0.371 & 0.456 & \multicolumn{1}{c|}{0} \\
    \ticker{EPU} & 0.130 & 0.503 & 0.545 & 1.000
      & 0.156 & 0.544 & 1.000 \\
    \ticker{EWM} & 0.416 & \multicolumn{1}{c}{0} & \multicolumn{1}{c}{0} & \multicolumn{1}{c|}{0}
      & 0.473 & \multicolumn{1}{c}{0} & \multicolumn{1}{c|}{0} \\
    \ticker{EWT} & 0.059 & 0.033 & \multicolumn{1}{c}{0} & \multicolumn{1}{c|}{0}
      & \multicolumn{1}{c}{0} & \multicolumn{1}{c}{0} & \multicolumn{1}{c|}{0} \\
   \hline\rule[-1.5mm]{0mm}{6mm}%
    $x$ & 0.00 & 4.93 & 5.30 & 8.51 & 0.00 & 5.29 & 8.51 \\
    $e$ & 37.47 & 45.94 & 46.57 & 52.07 & 37.47 & 46.55 & 52.07 \\
    $\sigma$ & 12.62 & 13.55 & 13.70 & 16.42 & 12.64 & 13.69 & 16.42 \\
  \hline\rule[-1.5mm]{0mm}{6mm}%
    avg $e$ & \multicolumn{4}{.|}{44.77} & \multicolumn{3}{.|}{44.77} \\
    rms $\sigma$ & \multicolumn{4}{.|}{13.73-} & \multicolumn{3}{.|}{13.73+} \\
  \hline
  \end{tabular}
}
\end{table}


\subsection{Relative risk decomposition}
\label{relative-decomp}

The risk, $\sigma$, of an individual fund or portfolio depends only on
its periodic returns. However, the systemic, productive, and
nonproductive components of this risk depend on the universe of funds in
which the fund or portfolio resides. Table \ref{fund_universe_2010}
illustrates this dependence with a fund and a portfolio from the
universes we have considered. In this table each total risk, $\sigma$,
is the square root of the sum of the squares of its four component
risks.

\begin{table}[H]
  \centering
  \caption{\label{fund_universe_2010}%
    Decomposition of risk relative to fund universe\\
    last 200 days of 2010 -- late-heavy weights
  }
  { \small
    \newcolumntype{.}{D{.}{.}{2}}
\begin{tabular}{|l|.|.||.|.|}
\cline{2-5}
  \multicolumn{1}{l|}{\rule{0mm}{4mm}}
  & \multicolumn{2}{c||}{decomposition} 
  & \multicolumn{2}{c|}{decomposition} \\
  \multicolumn{1}{l|}{}
  & \multicolumn{2}{c||}{of \ticker{EEM}} 
  & \multicolumn{2}{c|}{of $\normalsize\mathbf{p}_\textbf{E}$} \\
\cline{2-5}
  \multicolumn{1}{l|}{\rule{0mm}{4mm}}
  & \multicolumn{1}{c|}{18 fund}
  & \multicolumn{1}{c||}{5 fund}
  & \multicolumn{1}{c|}{18 fund}
  & \multicolumn{1}{c|}{4 fund} \\
  \multicolumn{1}{l|}{}
  & \multicolumn{1}{c|}{universe}
  & \multicolumn{1}{c||}{universe}
  & \multicolumn{1}{c|}{universe}
  & \multicolumn{1}{c|}{universe} \\
\hline\rule{0mm}{4mm}%
  expected return ($e$) & 21.79 & 21.79 & 37.47 & 37.47 \\
\hline\rule{0mm}{4mm}%
  systemic risk ($f_0$) &  7.93 &  5.07 & 7.93 & 12.62 \\
  productive risk ($|x|$) & 8.30 & 9.81 & 4.01 & 0.00 \\
  major nonproductive risk ($|y|$) & 17.37 & 16.81 & 7.83 & 0.00\\
  other nonproductive risk & 2.96 &  6.13 & 4.37 & 0.00 \\
\hline\rule{0mm}{4mm}%
  total risk ($\sigma$) & 21.03 & 21.03 & 12.62 & 12.62 \\
\hline
\end{tabular}
  }
\end{table}


\subsection{2011 results}
\label{2011results}

So far we have restricted our examples to the last 200 market days of
2010 with adjusted closing prices and returns normalized at the closing
prices of that year. Late-heavy weights have been emphasized with the
idea that a strong performance in the latter part of 2010 should carry
over into 2011.

This did not turn out to be the case. Figures \ref{growth5_2011} and
\ref{growth4_2011} show how the five large ETFs of Section
\ref{5largeETFs} and the four emerging market ETFs of Section
\ref{4emfunds} performed over 2011. These are graphs of
daily adjusted closing prices. Again the prices have been normalized at 100
on 2010-12-31 so that notional portfolio proportions correspond to
investment portfolio proportions at 2010 closing prices.

\begin{figure}[H]
  \centering
  \caption{\label{growth5_2011}%
    2011 adjusted closing prices of five large ETFs\\
    prices normalized at 100 on 2010-12-31}
  \includegraphics{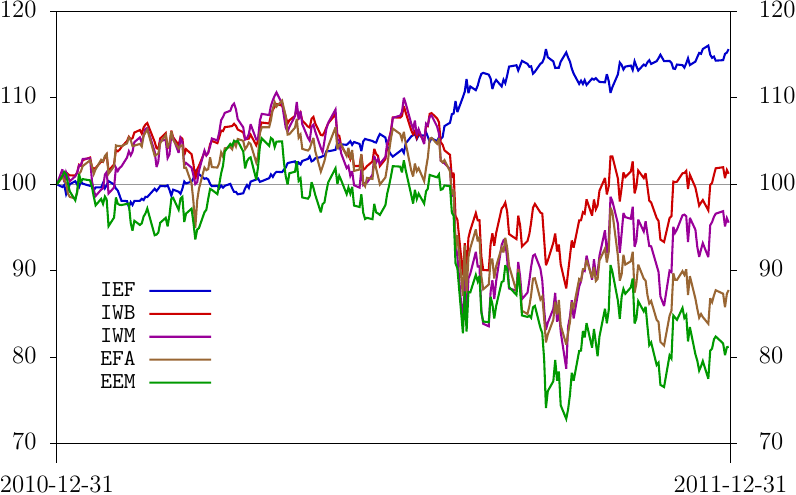}
\end{figure}

\begin{figure}[H]
  \centering
  \caption{\label{growth4_2011}%
    2011 adjusted closing prices of four emerging market ETFs\\
    prices normalized at 100 on 2010-12-31}
  \includegraphics{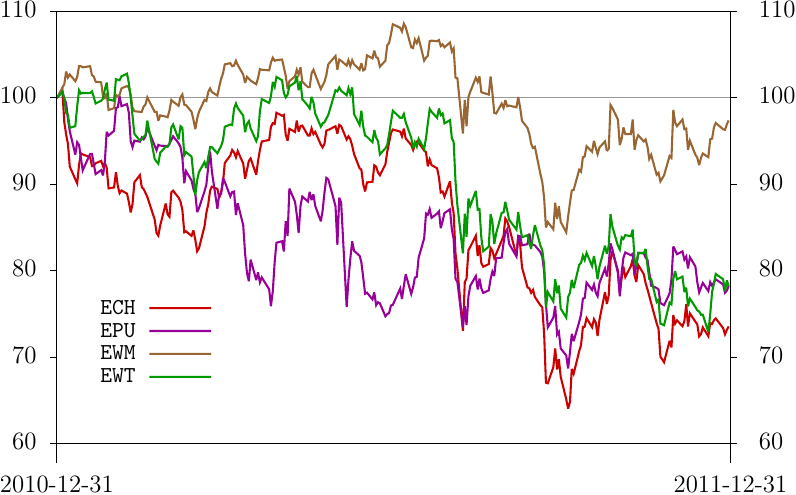}
\end{figure}

Tables \ref{decomp5_2011U252} and \ref{decomp4_2011U252} show the output
of \textbf{rtndecomp} as applied to this data. Here we have used the
full 252 markets-days of returns with uniform weighting. Now the
expected returns are the total returns of the respective
securities over the whole of 2011.

\begin{table}[H]
  \centering
  \caption{\label{decomp5_2011U252}%
    Decomposition of return data -- 5 large ETF universe\\
    the 252 market-days of 2011 -- uniform weights
  }
{ \small
  \newcolumntype{.}{D{.}{.}{2}}
  \newcolumntype{W}{|>{\columncolor{white}}c|}
$
\begin{array}{|c|.....||r.|}
\cline{1-6}\rule{0mm}{4mm}
  \text{fund}
  & \multicolumn{1}{r}{\parbox{6.5ex}{\hfill\ticker{IEF}\hspace*{0.2ex}}}
  & \multicolumn{1}{r}{\parbox{7.0ex}{\hfill\ticker{IWB}\hspace*{0.2ex}}}
  & \multicolumn{1}{r}{\parbox{7.0ex}{\hfill\ticker{IWM}\hspace*{0.2ex}}}
  & \multicolumn{1}{r}{\parbox{7.0ex}{\hfill\ticker{EFA}\hspace*{0.2ex}}}
  & \multicolumn{1}{r|}{\parbox{7.0ex}{\hfill\ticker{EEM}\hspace*{0.8ex}}} \\
\hline\rule[-1.5mm]{0mm}{6mm}
  E &   15.65 &  1.24 &  -4.43 &  -12.23 & -18.79
  & \multicolumn{2}{c|}{\widehat{V}_\text{T}} \\
\hline\rowcolor{productive}
  \multicolumn{1}{W}{\rule{0mm}{4mm}}
  & 1.59 & -6.29 & -9.39 & -13.66 & -17.26 & 615 & 20.8\% \\
\cline{2-8}
  \rowcolor{nonproductive}\multicolumn{1}{W}{\rule{0mm}{4mm}}
  & -6.91 & 21.05 & 27.09 & 22.92 & 21.30 & 2203 & 74.6\% \\
  \rowcolor{nonproductive}\multicolumn{1}{W}{\raisebox{1.5ex}[0pt]{$F$}}
  & -1.80 & 0.21 & -5.70 & 7.11 & -1.19 & 88 & 3.0\% \\
  \rowcolor{nonproductive}\multicolumn{1}{W}{}
  & -1.84 & 1.27 & -3.35 & -2.36 & 4.96 & 46 & 1.6\% \\
\hhline{|=|=====#==|}\rule{0mm}{4mm}
  & \multicolumn{1}{r@{\hspace*{1.1ex}}}{57}
  & \multicolumn{1}{r@{\hspace*{1.1ex}}}{484}
  & \multicolumn{1}{r@{\hspace*{1.1ex}}}{866}
  & \multicolumn{1}{r@{\hspace*{1.1ex}}}{768}
  & \multicolumn{1}{r@{\hspace*{1.1ex}}||}{777}
  & \multicolumn{1}{r}{2952}
  & \multicolumn{1}{r|}{100\%\hspace*{-0.8ex}} \\
\cline{8-8} \raisebox{1.6ex}[0pt]{$\widehat{V}_\text{T}$}
  & 1.9\% & 16.4\% & 29.3\% & 26.0\% & 26.3\%
  & \multicolumn{1}{r|}{100\%\hspace*{-0.5ex}}\\
\cline{1-7}
  \multicolumn{7}{c}{\rule{0mm}{5mm}
    \text{with}~f_0 = 4.62,~ e_0 = 12.74,~\text{and}~e_F = 1.827
  }
\end{array}
$
}
\end{table}

\begin{table}[H]
  \centering
  \caption{\label{decomp4_2011U252}%
    Decomposition of return data -- 4 emerging market ETF universe\\
    the 252 market-days of 2011 -- uniform weights
  }
{ \small
  \newcolumntype{.}{D{.}{.}{2}}
  \newcolumntype{W}{|>{\columncolor{white}}c|}
$
\begin{array}{|c|....||r.|}
\cline{1-5}\rule{0mm}{4mm}
  \text{fund}
  & \multicolumn{1}{r}{\parbox{6.5ex}{\hfill\ticker{ECH}\hspace*{0.2ex}}}
  & \multicolumn{1}{r}{\parbox{7.0ex}{\hfill\ticker{EPU}\hspace*{0.2ex}}}
  & \multicolumn{1}{r}{\parbox{7.0ex}{\hfill\ticker{EWM}\hspace*{0.2ex}}}
  & \multicolumn{1}{r|}{\parbox{7.0ex}{\hfill\ticker{EWT}\hspace*{0.2ex}}} \\
\hline\rule[-1.5mm]{0mm}{6mm}
  E &   -26.45 &  -21.78 &  -2.60 &  -21.78
  & \multicolumn{2}{c|}{\widehat{V}_\text{T}} \\
\hline\rowcolor{productive}
  \multicolumn{1}{W}{\rule{0mm}{4mm}}
  & -10.85 & -8.00 & 3.70 & -8.01 & 260 & 28.8\% \\
\cline{2-7}
  \rowcolor{nonproductive}\multicolumn{1}{W}{\rule{0mm}{4mm}$F$}
  & 4.02 & -19.66 & 3.58 & 7.08 & 466 & 51.7\% \\
  \rowcolor{nonproductive}\multicolumn{1}{W}{}
  & 9.09 & -1.70 & -0.84 & -9.46 & 176 & 19.5\% \\
\hhline{|=|====#==|}\rule{0mm}{4mm}
  & \multicolumn{1}{r@{\hspace*{1.1ex}}}{217}
  & \multicolumn{1}{r@{\hspace*{1.1ex}}}{453}
  & \multicolumn{1}{r@{\hspace*{1.1ex}}}{27}
  & \multicolumn{1}{r@{\hspace*{1.1ex}}||}{204}
  & \multicolumn{1}{r}{901}
  & \multicolumn{1}{r|}{100\%\hspace*{-0.8ex}} \\
\cline{7-7} \raisebox{1.6ex}[0pt]{$\widehat{V}_\text{T}$}
  & 24.0\% & 50.3\% & 3.0\% & 22.6\%
  & \multicolumn{1}{r|}{100\%\hspace*{-0.5ex}}\\
\cline{1-6}
  \multicolumn{6}{c}{\rule{0mm}{5mm}
    \text{with}~f_0 = 19.88,~ e_0 = -8.66,~\text{and}~e_F = 1.639
  }
\end{array}
$
}
\end{table}

It is quite easy to show that the expected expected returns in Tables
\ref{decomp5_2011U252} and \ref{decomp4_2011U252} are, in fact, total
returns for the year. In the general the expected return of a security
is given by
\begin{equation}\label{totalrtn1}
  e = \rho\,\sum_{i=1}^M \omega_i\,r_i
\end{equation}
with the $r_i~ (i=1,\ldots,M)$ being the successive periodic returns.
When $\rho = M$, $\omega_i=1/M$, and $r_i = (a_i - a_{i-1})/a_0$, with
the $a_i~ (i=0,1,\ldots,M)$ being successive adjusted closing prices for
the security, equation \eqref{totalrtn1} simplifies to
\begin{equation}\label{totalrtn2}
  e = \frac{a_M}{a_0}-1,
\end{equation}
which is, essentially by the definition of adjusted closing prices
(\cite{Norton:2010uq}), the total return of the security over the $M$
periods.


\section{Summary}\label{summary}

We have described an orthogonal decomposition of the space of ex post
periodic returns for a given universe of securities. The risk space,
which is orthogonal to the expected return axis, is decomposed into
systemic, productive, and principal nonproductive dimensions. Our
\textbf{rtndecomp} function accomplishes this decomposition. A technical
discussion and listing of this algorithm is given in the appendix.

In the examples of Section \ref{output-examples} we have emphasized the
two-dimensional $XY$-projection of periodic return data. The
minimum-$|y|$ path through the portfolio simplex is easily obtained from
the $XY$-projection of the data. The minimum-$|y|$ path can very closely
approximate the minimum-variance path for a small universe of
securities.

In the future we hope to develop a minimum-variance algorithm of the
form
\begin{align*}
   \textbf{function:}\quad
   P &= ~\textbf{minvar}(E, F)\\
   \text{with}\quad
   [E, F] &= ~\textbf{rtndecomp}(R, \bm{\omega}, \rho).
\end{align*}
Here $P=[\mathbf{p}_1,\ldots,\mathbf{p}_{n_P}]$ would contain successive
corner portfolios of the minimum-variance path through $\bm{\Delta}$ and
include single security portfolios at either end.

Given $[E, F] = \textbf{rtndecomp}(R, \bm{\omega}, \rho)$ and such a
\textbf{minvar} function, the minimum-$|y|$ path would be given by
$P_2=\textbf{minvar}(E, F(1 : 2, :))$. More generally, the
piecewise-linear paths through $\bm{\Delta}$ determined by the
portfolios in $P_k=\textbf{minvar}(E, F(1 : k, :))~ (k=2,\ldots,m)$
would approximate the minimum-variance path $P = P_m$ with successively
better approximations; moreover each portfolio in $P_k$ would contain at most
$k$ securities.


\appendix
\numberwithin{table}{section}
\numberwithin{equation}{section}
\numberwithin{figure}{section}


\newpage
\section{The flow of the rtndecomp algorithm}
\label{alg-descript}

Section \ref{octave-listing} gives the complete
\href{http://www.gnu.org/software/octave/}{GNU Octave} listing of the
\textbf{rtndecomp} algorithm. In this section we describe the ideas
behind specific sections of the listing. Certain simplifying assumptions
have been made in the interest of clarity. For example we only consider $\rho
= 1$ periods per unit time. Scaling for different $\rho$, e.g. $\rho =
252$ market-days per year, occurs at the end of the algorithm, as
described in Section \ref{scaling}.

This section is arranged in blocks. Each block summarizes a specific
section of the \textbf{rtndecomp} code. \\

\textbf{Initial setup}  (rtndecomp: 105 -- 111)\\
\begin{tabular}{@{~~}p{2.5in}@{\#~~}p{3.0in}}
  $ E = \bm{\omega}^T R;$ & $ \bm{\omega} > \mathbf{0}_M,~
    \bm{\omega}^T \mathbf{1}_M = \sum \bm{\omega} = 1. $ \\
  $ \bm{\beta} = \sqrt{\bm{\omega}};$ & $\|\bm{\beta}\|^2
    = \bm{\beta}^T\bm{\beta} = 1. $ \\
  $ Z = \diag(\bm{\beta}) R - \bm{\beta} E; $ & $ \bm{\beta}^T Z = 0,~
    V = Z^T Z$ (covariance matrix).
\end{tabular}

The linear isometry
\[ \mathbf{r}\mapsto\mathbf{x}=\diag(\bm{\beta})\mathbf{r},~
\bm{\beta}=\sqrt{\bm{\omega}} \]
converts the $\omega$-metric into the standard,
sum-of-squares-metric for $\mathds{R}^M:
\langle\mathbf{r},\mathbf{r}\rangle_\omega
= \mathbf{x}^T\mathbf{x}$.
The risk vectors $\mathbf{z}_j$ in the $M\times n$ matrix
$Z = [\mathbf{z}_1,\ldots,\mathbf{z}_n]$ are the isometric images
of the risk vectors $\mathbf{z}_j$ of \eqref{risk-vectors};
now the covariance matrix $V=[v_{jk}]$ of \eqref{cov1} is
given by $V=[\mathbf{z}_j^T\mathbf{z}_k]=Z^T Z$.\\

\textbf{QR factorization of $Z$}  (rtndecomp: 113 -- 121)\\
\begin{tabular}{@{~~}p{2.5in}@{\#~~}p{3.0in}}
  $ Z = Q F; $ & compact QR factorization.
\end{tabular}\\
\hspace*{3.0ex}(now
  $F \sim [\mathbf{z}_1,\mathbf{z}_2,\dots,\mathbf{z}_n]$)

The Octave code for the QR factorization actually reads\\
\hspace*{5.0ex}{\ttfamily[Q, F, J] = \textbf{qr}(Z, 0);}\\ The output
consists of an $M\times n$ matrix $Q$ with orthonormal columns,
an upper-triangular, $n\times n$ matrix $F$, and a permutation,
$J$, of the index sequence $[1,\ldots,n]$. These matrices satisfy
$Z(:, J) = Q F.$

This is QR factorization with column pivoting. At the end of the actual
\textbf{rtndecomp} algorithm, $F$-columns are returned to the initial
order with the replacement $F := \mbox{$F(:, J^{-1})$}$.
In this description we assume that $M \ge n$ and that $Z$ is
not rank deficient. Thus we can skip column pivoting and
start with the QR factorization $Z=QF$.

The columns of $Q$ make up an orthonormal basis for the range of $Z$ in
$\mathds{R}^M$. The columns $\mathbf{f}_1,\ldots,\mathbf{f}_n$ of the
upper-triangular $F$ are the coordinate vectors of the
$\mathbf{z}_1,\ldots,\mathbf{z}_n$ in $Z$ with respect to this basis. We
signify this situation with the notation $F \sim
[\mathbf{z}_1,\ldots,\mathbf{z}_n]$. It simply says that ``$F$ is the
matrix of coordinate vectors for $[\mathbf{z}_1,\ldots,\mathbf{z}_n]$
with respect to some orthonormal basis for the range of
$[\mathbf{z}_1,\ldots,\mathbf{z}_n]$.'' Then, regardless of the
orthonormal basis $Q$,~
$\mathbf{z}_j^T \mathbf{z}_k
= \mathbf{f}_j^T Q^TQ\,\mathbf{f}_k
= \mathbf{f}_j^T \mathbf{f}_k$
for $j,k=1,\ldots,n$, since $Q^TQ = I_n$.\\

\textbf{$Z$-flat tangent space}  (rtndecomp: 158 -- 163)\\
\begin{tabular}{@{~~}l}
  $ F(1, j) := F(1, j) - F(1, 1);\quad(j=2,\ldots,n)$
\end{tabular}\\
\hspace*{3.0ex}(now
  $F \sim [\mathbf{z}_1,\mathbf{z}_2-\mathbf{z}_1,\dots,
    \mathbf{z}_n-\mathbf{z}_1]$)

The matrix $F$ is still upper triangular, but now the
$\mathbf{z}$-vectors corresponding to the second through the last
columns of $F$ span the $Z$-flat tangent space, $\mathcal{T}(Z)$. \\

\textbf{Hessenberg QR via Givens}  (rtndecomp: 165 -- 175)
\begin{multline}\label{givens}
  \small
  \hspace*{2.0ex} {\normalsize F} = 
  \left[
  \begin{array}{c|cccc}
    \times & \times & \times & \times & \times \\
    0 & \times & \times & \times & \times \\
    0 & 0 & \times & \times & \times \\
    0 & 0 & 0 & \times & \times \\
    0 & 0 & 0 & 0 & \times \\
  \end{array}
  \right]
  \to
  \left[
  \begin{array}{c|cccc}
    \times & \times & \times & \times & \times \\
    \times & 0 & \times & \times & \times \\
    0 & 0 & \times & \times & \times \\
    0 & 0 & 0 & \times & \times \\
    0 & 0 & 0 & 0 & \times \\
  \end{array}
  \right]
  \to
  \left[
  \begin{array}{c|cccc}
    \times & \times & \times & \times & \times \\
    \times & 0 & \times & \times & \times \\
    \times & 0 & 0 & \times & \times \\
    0 & 0 & 0 & \times & \times \\
    0 & 0 & 0 & 0 & \times \\
  \end{array}
  \right]
  \\[1.0ex]
  \to
  \left[
  \begin{array}{c|cccc}
    \times & \times & \times & \times & \times \\
    \times & 0 & \times & \times & \times \\
    \times & 0 & 0 & \times & \times \\
    \times & 0 & 0 & 0 & \times \\
    0 & 0 & 0 & 0 & \times \\
  \end{array}
  \right]
  \to
  \left[
  \begin{array}{c|cccc}
    \times & \times & \times & \times & \times \\
    \times & 0 & \times & \times & \times \\
    \times & 0 & 0 & \times & \times \\
    \times & 0 & 0 & 0 & \times \\
    \pm f_0 & 0 & 0 & 0 & 0 \\
  \end{array}
  \right]
  = {\normalsize F}\qquad
\end{multline}
\\[1.0ex]
\hspace*{3.0ex}(still
  $F \sim [\mathbf{z}_1,\mathbf{z}_2-\mathbf{z}_1,\dots,
    \mathbf{z}_n-\mathbf{z}_1]$)

Here we apply a sequence of $n-1$ Givens rotations,
$G_1,\ldots,G_{n-1}$, to zero the subdiagonal elements of 
the upper Hessenburg submatrix $F(:,2 : n)$.
Then
\[
  F := G_{n-1}^T\cdots G_1^T F\quad\text{and}\quad
  Q := Q\, G_1\cdots G_{n-1}
\]
with $QF = [\mathbf{z}_1,\mathbf{z}_2-\mathbf{z}_1,\dots,
\mathbf{z}_n-\mathbf{z}_1]$ at either end of the sequence. This process
is described in Section 5.2.4 of \cite{Golub:1989cs}.

Looking at the last row of the final $F$ in \eqref{givens} we see that
$\mathbf{q}_n$ is orthogonal to the $Z$-flat tangent space:
$\mathbf{q}_n^T(\mathbf{z}_j - \mathbf{z}_1) = 0$ for $j=2,\ldots,n$.
It follows that $\mathbf{z}_0 = \mathbf{q}_n f_{n1}$ is the point on
the $Z$-flat that is closest to origin, with
$f_0=|f_{n1}|=\|\mathbf{z}_0\|$ being the systemic risk of the system.\\

\textbf{Extract systemic risk $f_0$}  (rtndecomp: 177 -- 189)\\
\begin{tabular}{@{~~}l}
  $f_0 = |f_{n1}|;$\\
  discard the last row of $F$:
\end{tabular}\\
$ \small
  \hspace*{2.0ex} {\normalsize F} =
  \left[
  \begin{array}{c|cccc}
    \times & \times & \times & \times & \times \\
    \times & 0 & \times & \times & \times \\
    \times & 0 & 0 & \times & \times \\
    \times & 0 & 0 & 0 & \times \\
    \pm f_0 & 0 & 0 & 0 & 0 \\
  \end{array}
  \right]
  \to
  \left[
  \begin{array}{c|cccc}
    \times & \times & \times & \times & \times \\
    \times & 0 & \times & \times & \times \\
    \times & 0 & 0 & \times & \times \\
    \times & 0 & 0 & 0 & \times \\
  \end{array}
  \right]
  = {\normalsize F}
$\\[0.5ex]
\hspace*{2.0ex}$m = n-1;\quad(m = \text{the number of rows of~} F)$
\\[1.0ex]
\hspace*{3.0ex}(now
  $F \sim [\mathbf{z}_1-\mathbf{z}_0,\mathbf{z}_2-\mathbf{z}_1,\dots,
    \mathbf{z}_n-\mathbf{z}_1]$)
    
Here the last column of the previous $Q$ is discarded so that the
remaining $m = n-1$ columns form an orthonormal basis for the $Z$-flat
tangent space, $\mathcal{T}(Z)$.\\

\textbf{$Z$-flat gradient of expected return}  (rtndecomp: 204 -- 223)\\
\begin{tabular}{@{~~}p{2.5in}ll}
  $\mathbf{g}^T F(:, 2 : n) = E(2 : n) - e_1;$
    & \# \textit{eflag} = \textbf{false} if solution is exact. \\
  $e_0 = e_1 - \mathbf{g}^T F(:, 1);~ e_F = \|\mathbf{g}\|;$
    & \# when \textit{eflag} = \textbf{false}.
\end{tabular}

If expected return is an affine function of (vector) risk, then
\[ [e_2-e_1,\ldots,e_n-e_1] = \mathbf{g}^T
   [\mathbf{z}_2-\mathbf{z}_1,\ldots\mathbf{z}_n-\mathbf{z}_1] \]
can be solved exactly
for $\mathbf{g}\in\mathcal{T}(Z)$. This will be the case (and
\textit{eflag} will be \textbf{false}) unless
$\mathbf{1}_M$ is parallel to the $R$-flat (Proposition
\ref{isomorph}).

In this discussion we will assume that $\mathbf{1}_M$ is not
parallel to the $R$-flat. Then
$\mathbf{u}_1 = \mathbf{g}/\|\mathbf{g}\|\in\mathcal{T}(Z)$
is direction of steepest increase in expected return, and
\begin{equation}\label{efuncf}
  e = e_0 + e_F\,\mathbf{u_1}^T\mathbf{z}
\end{equation}
holds for all $\mathbf{z}$ in the $Z$-flat, where $e_0 = e_1 -
\mathbf{g}^T\mathbf{z}_1$ and $e_F =
\|\mathbf{g}\|$.

In this block we solve for the coordinates of
$\mathbf{g}$ with respect to the current orthonormal basis, $Q =
[\mathbf{q}_1,\ldots,\mathbf{q}_m]$, for the tangent space
$\mathcal{T}(Z)$. Then $e_0$ and $e_F$ are computed from the coordinate
representation. In the following block, $\mathbf{g}\in\mathbf{R}^m$
denotes the coordinate vector corresponding to
$\mathbf{g}\in\mathcal{T}(Z)$.\\

\begin{minipage}{3.7in}
\textbf{Householder reflection}
  (rtndecomp: 225 -- 235) \\
\begin{tabular}{@{~~}p{3.5in}l}
  $H = I_m - \beta\mathbf{v}\mathbf{v}^T;~~
  \beta = 2/\|\mathbf{v}\|^2;~~
  H \mathbf{g}=\bm{\delta_1}\|\mathbf{g}\|;$ \\
  $F := H F;$ \\
  \hspace*{1.0ex}(still
  $F \sim [\mathbf{z}_1-\mathbf{z}_0,\mathbf{z}_2-\mathbf{z}_1,\dots,
    \mathbf{z}_n-\mathbf{z}_1]$) \\
  $F(:, j) := F(:, j) + F(:, 1);~~(j = 2, \ldots, n)$ \\
  \hspace*{1.0ex}(now
  $F \sim [\mathbf{z}_1-\mathbf{z}_0,\mathbf{z}_2-\mathbf{z}_0,\dots,
    \mathbf{z}_n-\mathbf{z}_0]$) \\
\end{tabular}
\end{minipage}
\begin{minipage}{1.5in}
\begin{tikzpicture}[scale=1.00,>={angle 60}]
\draw[thick,->] (0,0) -- (1,0) node[left=14,below=0]{$\bm{\delta}_1$};
\draw[thick,red,->] (0,0) -- (-0.850,1.472) node[above]{$\mathbf{g}$};
\draw[dashed] (-0.250, -0.433) -- (0.750,1.299);
\draw[thick,->] (0,0) -- (-0.500,0.866) node[left=5,below=7]{$H \bm{\delta}_1$};
\draw[<->] (-0.433+0.500,0.250+0.866) -- (0.433+0.500,-0.250+0.866);
\node at (0.8,2) {\parbox{2in}{\centering Householder\\reflection $H$}};
\end{tikzpicture}
\end{minipage}

The first basis vector, $\mathbf{q_1}$, of the current basis $Q$ has
coordinate vector $\bm{\delta}_1$, the first column of the $m\times m$
identity matrix $I_m$. If $H$ is the Househoulder reflection of
$\mathds{R}^m$ that maps $\mathbf{g}$ to $\bm{\delta}_1\|\mathbf{g}\|$,
then the first column of $QH$ is the direction of maximum increase in
expected return in $\mathcal{T}(Z)$. Thus, after the replacements $Q :=
QH$ and $F := HF$, we still have
$F \sim [\mathbf{z}_1-\mathbf{z}_0,\mathbf{z}_2-\mathbf{z}_1,\dots,
    \mathbf{z}_n-\mathbf{z}_1]$,
but now $e_F$ times the first coordinates of $F$ produces the corresponding
changes in expected return:
\[ e_F F(1, :) = [e_1-e_0,e_2-e_1,\ldots,e_n-e_1].\]
Then the replacements~
$F(:, j) := F(:, j) + F(:, 1)~(j = 2, \ldots, n)$~ result in
\begin{align*}
  F &\sim [\mathbf{z}_1-\mathbf{z}_0,\mathbf{z}_2-\mathbf{z}_0,\dots,
    \mathbf{z}_n-\mathbf{z}_0] \\
\intertext{and}
  E &= [e_1,\ldots,e_n] = e_0 + e_F F(1, :)
\end{align*}
with~ $V = n f_0^2 + F^T F.$\\

\textbf{Principal components of nonproductive risk}
(rtndecom: 237 -- 243) \\
\begin{tabular}{@{~~}p{2.5in}p{3.2in}}
  $ F(2 : m, :) = U \Sigma V^T; $ & \#~ compact singular value decomposition.\\
  $ \Sigma = \diag(\tau_2,\ldots,\tau_m); $ & \#~ the $\tau_i$ of Definition \ref{principal-npr},
  Section \ref{nonproductive-risk}.\\
  $ F(2 : m, :) := \Sigma V^T; $
\end{tabular}\\
\hspace*{3.0ex}(final
  $F \sim [\mathbf{z}_1-\mathbf{z}_0,\mathbf{z}_2-\mathbf{z}_0,\dots,
    \mathbf{z}_n-\mathbf{z}_0]$)

Entering this block the first row of $F$ consists of the coordinates of
$[\mathbf{z}_1-\mathbf{z}_0,\mathbf{z}_2-\mathbf{z}_0,\dots,
    \mathbf{z}_n-\mathbf{z}_0]$
in the productive risk direction. The remaining $m-1$ rows represent
the nonproductive risks.

The principal components of nonproductive risk are computed with the
code\\
\hspace*{5.0ex}{\ttfamily[U, S, V] = \textbf{svd}(F(2 : m, :), 0);}\hfill\\
The output matrices $U$ and $V$ (not to be confused with the covariance
matrix $V$) have orthonormal columns and dimensions $(m-1)\times(m-1)$
and $n\times(m-1)$, respectively. The principal nonproductive risks,
$\tau_2\ge\ldots\ge\tau_n>0$, of Definition \ref{principal-npr}
are contained in the diagonal matrix
$\Sigma$ ( = \texttt{S} in the \textbf{rtndecomp} code). Now the
replacements\\
\hspace*{5.0ex}$F(2:m, :):=\Sigma\,V^T,\quad Q(:, 2:m):=Q(:, 2:m)\,U$\\
maintain~
  $F \sim [\mathbf{z}_1-\mathbf{z}_0,\mathbf{z}_2-\mathbf{z}_0,\dots,
    \mathbf{z}_n-\mathbf{z}_0]$, but
organize the last $m-1$ rows of $F$ in the principal directions of
nonproductive risk, with row norms $\|F(i, :)\|=\tau_i$ for
$i=2,\dots,m$.\\

In this discussion we have described how the basis matrix $Q$ changes
from one block to another. It is always the case that $Q^TQ=I$, though
the size of the identity matrix diminishes from $n\times n$ as the
algorithm progresses. The actual \textbf{rtndecomp} algorithm
makes no attempt to keep track of the changing $Q$.


\newpage
\section{The Octave algorithm}
\label{octave-listing}

{\fontsize{11}{14}\selectfont

\begin{lstlisting}
#   Function: [E, F, f0, e0, eF] = rtndecomp(rtns, wgts, pput)
#---------------------------------------------------------------
#   Purpose
#     To decompose financial return data into orthogonal
#     "systemic", "productive", and "nonproductive"
#     risk-factors.
#   Input
#     rtns   - M by n matrix of periodic returns.
#     wgts   - M by 1 vector of positive weights or a scalar,
#              default: wgts = ones(M, 1) / M if wgts is a
#              scalar or rtns is the only input.
#     pput   - periods per unit time. default: pput = 1.
#              example: pput = 252 market-days/year.
#   Output
#     E      - 1 by n matrix of expected returns.
#     F      - m by n matrix of risk coefficients.
#              rank(F) = m unless F = zeros(1, n).
#     f0     - systemic risk (nonnegative).
#     e0     - systemic expected return.
#     eF     - return per unit change in the F(1, :)-direction.
#              (eF >= 0)
#   Global output 
#     eflag  - true iff the constant return vector, ones(M, 1),
#              parallels the returns flat or, said another way,
#                 ones(M, 1) = rtns * x
#              for some n-vector x with sum(x) = 0.
#   Variables and relationships
#     The 1 x n expected return matrix is
#       E = (wgts * pput)' * rtns.
#     The n x n covariance of return matrix is
#       V = Z' * diag(wgts * pput) * Z,
#     where Z is the M x n matrix of risk vectors
#       Z = rtns - ones(M, 1) * E / pput.
#     The output variables satisfy
#       1)  V = f0^2 + F' * F.
#       2)  E = e0 + eF * F(1, :) unless eflag is true;
#           then this relationship is approximate. However
#           mean(E) = mean(e0 + eF * F(1, :)) is always true.
#           We refer to norm(F(1, :)) as the productive risk
#           when eF is nonzero.
#       3)  The row norms,
#             tau(i) = norm(F(i, :)) (i = i0, .., m),
#           are the principal nonproductive risks, where
#             i0 = 2 if eF > 0, and i0 = 1 if eF = 0.
#           The nonproductive risks tau(i) decrease with
#           increasing i, and the corresponding row vectors
#           are pairwise orthogonal:
#             F(i, :)' * F(j, :) = 0 (i0 <= i < j <= m).
#
#---------------------------------------------------------------
#  Copyright (C) 2012  Vic Norton <mailto:vic@norton.name>
#
#  This GNU Octave program
#           http://www.gnu.org/software/octave/
#  is free software: you can redistribute it and/or modify it
#  under the terms of the GNU General Public License
#           http://www.gnu.org/copyleft/gpl.html
#  as published by the Free Software Foundation
#           http://www.fsf.org/
#  
#  This program is distributed in the hope that it will be
#  useful, but WITHOUT ANY WARRANTY; without even the implied
#  warranty of MERCHANTABILITY or FITNESS FOR A PARTICULAR
#  PURPOSE.  See the GNU General Public License for more
#  details.
#---------------------------------------------------------------

function [E, F, f0, e0, eF] = rtndecomp(rtns, wgts, pput)

  if !(nargin >= 1 || nargin <= 3)
    usage("[E, F, f0, e0, eF] = rtndecomp(rtns, wgts, pput)");
  endif
  
  [M, n] = size(rtns);
  global eflag;
  eps0 = eps * 1e2;           # precision
  
  ## check/set wgts & pput
  if nargin == 1 || isscalar(wgts)
    wgts = ones(M, 1) / M;
  else
    if !isvector(wgts)
      error("wgts must be a scalar or a vector");
    endif
    if length(wgts) != M
      error("wgts must have length M");
    endif
    if any(wgts <= 0)
      error("wgts must be positive");
    endif
    if rows(wgts) == 1
      wgts = wgts';
    endif
    wgts /= sum(wgts);
  endif
  if nargin == 3
    if pput < 1
      error("g must be 1 or greater");
    endif
  else
    pput = 1;
  endif
  sqrtpput = sqrt(pput);
  
  ## expected return matrix E and "risk" vectors Z for pput = 1
  E = wgts' * rtns;
  if nargout <= 1;
    E *= pput;
    return;
  endif
  sqrtwgts = sqrt(wgts);
  Z = diag(sqrtwgts) * rtns;
  Z -= sqrtwgts * E;          # V = Z' * Z = covariance matrix
  sigZ = norm(Z, "fro");      # square root of total variance
  
  ## QR factorization of Z with column pivoting
  [Q, F, J] = qr(Z, 0);       # Q is not used
  m = rows(F);
  epsZ = eps0 * sigZ;
  while m > 1 && norm(F(m, m : n)) <= epsZ
    F(m, :) = [];
    m -= 1;
  endwhile
  Jinv = [1 : n] * eye(n)(:, J)';
  
  ## finish up when m = 1 and all coefficients of F are the same
  if n == 1
    f0 = abs(F);
    if pput > 1
      E *= pput; f0 *= sqrtpput;
    endif
    e0 = E; eF = 0; F = 0; eflag = 0;
    return;
  elseif m == 1
    Fmin = min(F); Fmax = max(F);
    if Fmax - Fmin < eps0 * max(abs(Fmin), abs(Fmax))
      e0 = mean(E);
      f0 = abs((Fmin + Fmax)/2);
      eF = 0;
      F = zeros(1, n);
      ee = sum(E .* E);
      eflag = ee - n * e0 * e0 > eps0 * ee;
      if pput > 1
        E *= pput; e0 *= pput; f0 *= sqrtpput;
      endif
      return;
    endif
  endif
  
  ## At this point the columns of F are the coordinate vectors
  #  of the columns of Y = [y1, y2, ..., yn] = Z(:, J) with
  #  respect an orthonormal basis for the linear space, L(Z),
  #  spanned by the columns of Y or Z. We represent this
  #  situation with the notation
  #    F ~ [y1, y2, ..., yn].
  #  As the algorithm progesses the orthonomal basis for L(Z)
  #  changes. We do not keep track of the changing basis, but
  #  we do continue to note the elements of L(Z) represented by
  #  the coordinate-vector columns of F.
  
  ## E-flat & Z-flat tangent spaces
  e1 = E(J(1));               # base expected return
  y1 = Z(:, J(1));            # base risk vector
  nm1 = n - 1;
  B = E(J(2 : n)) - e1;       # differential expected returns
  F(1, 2 : n) -= F(1, 1);     # differential risks
  
  ## QR decomposition of Z-flat tangent space by Givens
  #  rotations
  for j = 2 : m
    jm1 = j - 1;
    [cs, sn] = givens(F(jm1, j), F(j, j));
    GV = [cs, sn; -sn, cs];   # Givens rotation
    F(jm1 : j, 1) = GV * F(jm1 : j, 1);
    F(jm1 : j, j : n) = GV * F(jm1 : j, j : n);
    F(j , j) = 0;
  endfor
  #    F ~ [y1, y2 - y1, ..., yn - y1]
  
  ## systemic risk f0
  if norm(F(m, m + 1 : n)) <= epsZ
    ## the Z-flat tangent space has dimension m - 1
    f0 = abs(F(m, 1));
    F(m, :) = [];
    m -= 1;
  else
    ## the Z-flat tangent space has dimension m
    f0 = 0;
  endif
  #    F ~ [y1 - y0, y2 - y1, ..., yn - y1]
  # where y0 = z0 is the point in the Z-flat that is closest
  # to the origin. 
  
  ## check for constant E
  Emin = min(E); Emax = max(E);
  if Emax - Emin < eps0 * max(abs(Emin), abs(Emax))
    ## E is constant: finish up
    F(:, 2 : n) += F(:, 1) * ones(1, n - 1);
    #    F ~ [y1 - y0, y2 - y0, ..., yn - y0]
    [U, S, V] = svd(F(:, Jinv), 0);
    F = S * V';
    E *= pput; e0 = mean(E);
    eF = 0; f0 *= sqrtpput; F *= sqrtpput;
    return;
  endif
  
  ## Z-flat direction of maximum increase
  #  in expected return
  g = (B / F(:, 2 : n))';     # E-gradient if eflag
  eflag = norm(g' * F(:, 2 : n) - B)/norm(B) > eps0;
  if !eflag  # E = e0 + eF * F(1, :) will be exact
    e0 = e1 - g' * F(:, 1);
    eF = norm(g);
    #    F ~ [y1 - y0, y2 - y1, ..., yn - y1]
  else       # E = e0 + eF * F(1, :) will be approximate
    F(:, 2 : n) += F(:, 1) * ones(1, n - 1);
    #    F ~ [y1 - y0, y2 - y0, ..., yn - y0]
    v = mean(F')';            # v = ymean - y0
    G = F - v * ones(1, n);
    #    G ~ [y1 - ymean, y2 - ymean, ..., yn - ymean]
    emean = mean(E);
    C = E(J) - emean;
    g = (C / G)';             # E-gradient if !eflag
    e0 = emean - g' * v;
    eF = norm(g);
  endif
  
  ## the Householder reflection, H g = (x, 0, ..., 0) 
  #  puts the productive risk coordinates in F(1, :)
  [v, b] = housh(g, 1, 0);
  signx = sign((g - (b * v) * (v' * g))(1));
  F = F - (b * v) * (v' * F);
  F(1, :) *= signx;
  if ~eflag
    #    F ~ [y1 - y0, y2 - y1, ..., yn - y1]
    F(:, 2 : n) += F(:, 1) * ones(1, n - 1);
  endif
  #    F ~ [y1 - y0, y2 - y0, ..., yn - y0]
  
  ## put the principal components of nonproductive risk
  #  in F(2 : m, :) 
  if m > 1
    [U, S, V] = svd(F(2 : m, :), 0);
    F(2 : m, :) = S * V';
  endif
  #    F ~ [y1 - y0, y2 - y0, ..., yn - y0]
  
  ## unpermute columns of F
  F = F(:, Jinv);
  #    F ~ [z1 - z0, z2 - z0, ..., zn - z0]
  
  ## scale these results if pput != 1
  if pput != 1
    E *= pput; e0 *= pput;
    eF *= sqrtpput; f0 *= sqrtpput; F *= sqrtpput;
  endif

endfunction
\end{lstlisting}
}


\clearpage

\end{document}